\documentclass[aps,reprint,superscriptaddress,pra,10pt]{revtex4-2}
\usepackage[utf8]{inputenc}
\usepackage[T1]{fontenc}
\usepackage{physics}
\usepackage{amsmath}
\usepackage{amssymb}
\usepackage{amsthm}
\usepackage{tikz}
\usepackage{algorithm}
\usepackage{algpseudocode}
\usepackage[breaklinks=true,hidelinks]{hyperref}
\usepackage{dsfont}
\usepackage{enumitem}
\usepackage{cases}
\usetikzlibrary{quantikz2}

\newtheorem{theorem}{Theorem}
\newtheorem{problem}{Problem}
\newtheorem{definition}[theorem]{Definition}

\newtheorem{lemma}[theorem]{Lemma}
\newtheorem{corollary}[theorem]{Corollary}

\DeclareMathOperator{\sgn}{sgn}
\newcommand{\eq}[1]{(\ref{eq:#1})}
\newcommand{\fig}[1]{{\bf\ref{fig:#1}}}
\newcommand{\sect}[1]{{\bf\ref{sec:#1}}}
\newcommand{\tab}[1]{{\bf\ref{tab:#1}}}

\newcommand{\vv}[1]{\boldsymbol{#1}}
\newcommand{\bpi}{\boldsymbol{\pi}}
\newcommand{\epsrel}{\epsilon_{\text{r}}}

\newcommand{\Discuss}[1]{\textcolor{red}{{\fontencoding{U}\fontfamily{futs}\selectfont\char 49\relax} \textit{discuss} {\fontencoding{U}\fontfamily{futs}\selectfont\char 49\relax}}}

\begin{document}
\title{Randomized semi-quantum matrix processing}

\author{Allan Tosta}
\affiliation{Federal University of Rio de Janeiro, Caixa Postal 68528, Rio de Janeiro, RJ 21941-972, Brazil}
\author{Thais de Lima Silva}
\affiliation{Quantum Research Center, Technology Innovation Institute, Abu Dhabi, UAE}

\author{Giancarlo Camilo}
\affiliation{Quantum Research Center, Technology Innovation Institute, Abu Dhabi, UAE}

\author{Leandro Aolita}
\affiliation{Quantum Research Center, Technology Innovation Institute, Abu Dhabi, UAE}
\affiliation{Federal University of Rio de Janeiro, Caixa Postal 68528, Rio de Janeiro, RJ 21941-972, Brazil}

\begin{abstract}
We present a hybrid quantum-classical framework for simulating generic matrix
functions more amenable to early fault-tolerant quantum hardware than standard
quantum singular-value transformations. The method is based on randomization
over the Chebyshev approximation of the target function while keeping the matrix
oracle quantum, and is assisted by a variant of the Hadamard test that removes
the need for post-selection. The resulting statistical overhead is similar to the fully
quantum case and does not incur any circuit depth degradation. On the contrary,
the average circuit depth is shown to get smaller, yielding equivalent reductions
in noise sensitivity, as explicitly shown for depolarizing noise and coherent errors.
We apply our technique to partition-function estimation, linear system solvers,
and ground-state energy estimation. For these cases, we prove advantages on
average depths, including quadratic speed-ups on costly parameters and even the
removal of the approximation-error dependence.  
\end{abstract}

\maketitle


\section{Introduction}\label{sec:intro}
Faster algorithms for linear algebra are a major promise of quantum computation, holding the potential for precious runtime speed-ups over classical methods. 
A modern, unified framework for such algorithms is given by the quantum signal processing (QSP) \cite{LowChuang2017,LowChuangQuantum2019} and, more generally, quantum singular-value transformation (QSVT) \cite{Gilyen2019} formalisms. These are powerful techniques to manipulate a matrix, coherently given by a quantum oracle, via polynomial transformations on its eigenvalues and singular values, respectively.
The class of matrix arithmetic attained is remarkably broad, encompassing primitives as diverse as 
Hamiltonian simulation, matrix inversion, ground-state energy estimation, Gibbs-state sampling, among others \cite{ChuangGrandUnification}.
Moreover, the framework often offers the state-of-the-art in asymptotic query complexities (i.e. number of oracle calls), in some cases matching known complexity lower bounds.
Nevertheless, the experimental requirements for full implementations are prohibitive for current devices, and it is not clear if the framework will be useful in practice before large-scale fault-tolerant quantum computers appear.

This has triggered a quest for {\it early fault-tolerant algorithms} for matrix processing that allow one to trade performance for nearer-term feasibility in a controlled way, i.e. with provable runtime guarantees \cite{silva_fragmented_2022,silva2022fourierbased,lin_heisenberg-limited_2022,wang2023quantum,campbell_random_2019,wan_randomized_2022,wang2023qubitefficient,Campbell_2021,Dong_2022,wang2023faster}. 
Particularly promising are randomized hybrid quantum-classical schemes to statistically simulate a matrix function via quantum implementations of more elementary ones \cite{lin_heisenberg-limited_2022,wang2023quantum,campbell_random_2019,wan_randomized_2022,wang2023qubitefficient,Campbell_2021}.
For instance, this has been applied to the Heaviside step function $\theta(H)$ of a Hamiltonian $H$, which allows for eigenvalue thresholding, a practical technique for Heisenberg-limited spectral analysis \cite{lin_heisenberg-limited_2022}. 
Two input access models have been considered there: quantum oracles as a controlled unitary evolution of $H$ \cite{lin_heisenberg-limited_2022,wang2023quantum,Dong_2022} and classical ones given by a decomposition of $H$ as a linear combination of Pauli operators \cite{campbell_random_2019,wan_randomized_2022,wang2023qubitefficient,Campbell_2021}. 
In the former, one Monte-Carlo simulates the Fourier series of $\theta(H)$ by randomly sampling its harmonics. In the latter -- in an additional level of randomization -- one also probabilistically samples the Pauli terms from the linear combination.

Curiously, however, randomized quantum algorithms for matrix processing have been little explored beyond the specific case of the Heaviside function. Ref. \cite{wang2023qubitefficient} put forward a randomized, qubit-efficient technique for Fourier-based QSP \cite{silva2022fourierbased,Dong_2022} for generic functions. 
However, the additional level of randomization can detrimentally affect the circuit depth per run compared to coherent oracles. 
On the other hand, in the quantum-oracle setting, the randomized algorithms above have focused mainly on controlled unitary evolution as the input access model. 
This is convenient in specific cases where $H$ can be analogically implemented. 
However, it leaves aside the powerful class of {\it block-encoding oracles}, i.e., unitary matrices with the input matrix as one of its blocks \cite{LowChuangQuantum2019}. 
Besides having a broader scope of applicability (including non-Hermitean matrices), such oracle types are also a more natural choice for digital setups. 
Moreover, randomized quantum algorithms have so far not addressed Chebyshev polynomials, the quintessential basis functions for approximation theory \cite{trefethen_approx}, which often attain better accuracy than Fourier series \cite{boyd_spectral_methods}. 
Chebyshev polynomials, together with block-encoding oracles, provide the most sophisticated and general arena for quantum matrix arithmetic \cite{LowChuang2017,LowChuangQuantum2019,Gilyen2019,ChuangGrandUnification}. 

Here, we fill in this gap. We derive a semi-quantum algorithm for Monte-Carlo simulations of QSVT with provably better circuit complexities than fully-quantum schemes as well as advantages in terms of experimental feasibility.
Our method estimates state amplitudes and expectation values involving a generic matrix function $f(A)$, leveraging three main ingredients: 
$i$) it samples each component of a Chebyshev series for $f$ with a probability proportional to its coefficient in the series;
$ii$) it assumes coherent access to $A$ via a block-encoding oracle; and
$iii$) $f(A)$ is automatically extracted from its block-encoding without post-selection, 
using a Hadamard test.
The combination of $i$) and $ii$) leaves untouched the maximal query complexity $k$ per run native from the Chebyshev expansion. 
In addition, the statistical overhead we pay for end-user estimations scales only with the $l_1$-norm of the Chebyshev coefficients. 
For the use cases we consider, this turns out to be similar (at worst up to logarithmic factors) to the operator norm of $f(A)$, which would govern the statistical overhead if we used fully-quantum QSVT with a Hadamard test. That is, our scheme does not incur any significant degradation with respect to the fully-quantum case either in runtime or circuit depth. 
On the contrary, the average query complexity can be significantly smaller than $k$. 
We prove interesting speed-ups of the former over the latter for practical use cases. 

These speed-ups translate directly into equivalent reductions in noise sensitivity: 
For simple models such as depolarization or coherent errors in the quantum oracle, we show that the estimation inaccuracy due to noise scales with the average query depth. In comparison, it scales with the maximal depth in standard QSVT implementations. 
Importantly, we implement each sampled Chebyshev polynomial with a simple sequence of queries to the oracle using qubitization; no QSP pulses are required throughout. 
Finally, $iii$) circumvents the need for repeating until success or quantum amplitude amplification. That is, no statistical run is wasted, and no overhead in circuit depth is incurred. 
In addition, the fully-quantum scheme requires an extra ancillary qubit controlling the oracle in order to implement the QSP pulses. All this renders our hybrid approach more experimentally friendly than coherent QSVT. 

As use cases, we benchmark our framework on four end-user applications: partition-function estimation of classical Hamiltonians via quantum Markov-chain Monte Carlo (MCMC);  partition-function estimation of quantum Hamiltonians via quantum imaginary-time evolution (QITE); linear system solvers (LSSs); and ground-state energy estimation (GSEE). 
The maximal and expected query depths per run as well as the total expected runtime (taking into account sample complexity) are displayed in Table \tab{main_table}. In all cases, we systematically obtain the following advantages (both per run and in total) of expected versus maximal query complexities. 

For MCMC, we prove a quadratic speed-up on a factor $\mathcal{O}(\log (Z_\beta\, e^{\beta}/\epsrel))$, where $Z_\beta$ is the partition function to estimate, at inverse temperature $\beta$, and $\epsrel$ is the tolerated relative error. 
For QITE, we remove a factor  $\mathcal{O}(\log (D\,e^{\beta}/Z_\beta\,\epsrel))$ from the scaling, where $D$ is the system dimension. For LSSs we consider two sub-cases: estimation of an entry of the (normalized) solution vector and of the expectation value of an observable $O$ on it. 
We prove quadratic speed-ups on factors $\mathcal{O}\big(\log(\kappa/\epsilon)\big)$ and $\mathcal{O}\big(\log(\kappa^2\, \|O\|/\epsilon)\big)$ for the first and second sub-cases, respectively, where $\|O\|$ is the operator norm of $O$, $\kappa$ is the condition number of the matrix, and $\epsilon$ the tolerated additive error. 
 This places our query depth at an intermediate position between that of the best known Chebyshev-based method \cite{Childs_2017} and the optimal one in general 
\cite{costa2021optimal}.
In turn, compared to the results obtained in \cite{wang2023qubitefficient} via full randomization, our scaling is one power of $\kappa$ superior. 
Finally, for GSEE, we prove a speed-up on a factor that depends on the overlap $\eta$ between the probe state and the ground state: the average query depth is $\mathcal{O}\big(\frac{1}{\xi}\sqrt{\log(1/\eta)}/\log(1/\xi)\big)$, whereas the maximal query depth is $\mathcal{O}\big(\frac{1}{\xi}\log (1/\eta)\big)$, with $\xi$ the additive error in the energy estimate.


\section{Preliminaries}\label{sec:prelim}

We consider the basic setup of Quantum Singular Value Transformation (QSVT) \cite{Gilyen2019,ChuangGrandUnification}. 
This is a powerful technique for synthesizing polynomial functions of a linear operator embedded in a block of a unitary matrix, via polynomial transformations on its singular values. 
Combined with approximation theory \cite{Vishnoi2013}, this leads to state-of-the-art query complexities and an elegant unifying structure for a variety of quantum algorithms of interest. 
For simplicity of the presentation, in the main text we focus explicitly on the case of Hermitian matrices. There, QSVT reduces to the simpler setup of Quantum Signal Processing (QSP) \cite{LowChuang2017,LowChuangQuantum2019}, describing eigenvalue transformations. The extension of our algorithms to QSVT for generic matrices is straightforward and is left for App. \sect{QSVT}. 
Throughout the paper, we adopt the short-hand notation $[l]:=\{0,\ldots,l-1\}$ for any $l\in\mathbb{N}$.

The basic input taken by QSP is a block-encoding $U_{A}$ of the Hermitian operator $A$ of interest (the \textit{signal}). 
A block-encoding is a unitary acting on $\mathcal{H}_{sa}:=\mathcal{H}_s\otimes\mathcal{H}_{a}$, where $\mathcal{H}_s$ is the system Hilbert space where $A$ acts and $\mathcal{H}_{a}$ is an ancillary Hilbert space (with dimensions $D$ and $D_a$, respectively), satisfying 
\begin{equation}\label{eq:block_encoding}
\big(\bra{0}_a\otimes \mathds{1}_s\big)\,U_A\,\big(\ket{0}_a\otimes \mathds{1}_s\big) = A
\end{equation}
for some suitable state $\ket{0}_{a}\in \mathcal{H}_{a}$ (here $\mathds{1}_s$ is the identity operator in $\mathcal{H}_s$). 
Designing such an oracle for arbitrary $A$ is a non-trivial task \cite{camps2023explicit}, but efficient block-encoding schemes are known in cases where some special structure is present, e.g., when $A$ is sparse or expressible as a linear combination of unitaries \cite{LowChuangQuantum2019,Gilyen2019,sunderhauf2023blockencoding}. 
In particular, we will need the following particular form of $U_A$ that makes it amenable for dealing with Chebyshev polynomials. 

\begin{definition}[Qubitized block-encoding oracle]\label{def:qubitizedoracle}
Let 
 $A$ be a Hermitian matrix on $\mathcal{H}_s$ with spectral norm $\norm{A}\leq1$, eigenvalues $\{\lambda_\gamma\}_{\gamma\in [D]}$, and eigenstates $\{\ket{\lambda}_s\}$. 
A unitary $U_{A}$ acting on $\mathcal{H}_{sa}$
is called a (exact) \textit{qubitized block-encoding} of $A$ if 
it has the form
\begin{align}\label{eq:qubitizedoracle}
U_{A} = \bigoplus_{\gamma\in[D]} e^{-i\,\vartheta_\gamma\,Y_\gamma}
\,,
\end{align}
where $\vartheta_\gamma:=\arccos(\lambda_\gamma)$ and $Y_{\gamma}$ is the second Pauli matrix acting on the two-dimensional subspace spanned by $\big\{\ket{0}_a\otimes\ket{\lambda_\gamma}_s, \ket{\perp_{\lambda_\gamma}}_{sa}\big\}$ with ${}_{sa}\!\bra{\perp_{\lambda_\gamma}}\big(\ket{0}_a\otimes\ket{\lambda_\gamma}_s\big)=0$.
\end{definition}

\noindent A qubitized oracle of the form \eq{qubitizedoracle} can be constructed from any other block-encoding $U'_A$ of $A$ using at most one query to $U'_A$ and ${U'_A}^{-1}$, at most one additional ancillary qubit, and $\mathcal{O}(\log(D_a))$ quantum gates \cite{LowChuangQuantum2019}. 

Standard QSP takes as input the qubitized oracle $U_A$ and transforms it into (a block-encoding of) a polynomial function $\Tilde{f}(A)$. With the help of function approximation theory \cite{trefethen_approx}, this allows the approximate implementation of generic non-polynomial functions $f(A)$. 
The algorithm complexity is measured by the number of queries to $U_A$, which allows for rigorous quantitative statements agnostic to details of $A$ or to hardware-specific circuit compilations. 
For our purposes, only a simple QSP result will be needed, namely the observation \cite{LowChuangQuantum2019} that repeated applications of $U_A$ give rise to Chebyshev polynomials of $A$ (see App. \sect{chebyshev_proof} for a proof).

\begin{lemma}[Block encoding of Chebyshev polynomials] \label{cheb}
Let $U_A$ be a qubitized block-encoding of $A$. Then 
\begin{equation}
\big(\bra{0}_a\otimes \mathds{1}_s\big)\,U_A^j\,\big(\ket{0}_a\otimes \mathds{1}_s\big) = \mathcal{T}_j(A)\,,
\end{equation}
for $j\in\mathbb{N}$, where $\mathcal{T}_{j}(\cdot)$ is the $j$-th order Chebyshev polynomial of the first kind. 
\end{lemma}

We are interested in a truncated Chebyshev series 
\begin{equation}
\label{eq:ftilde}
\Tilde{f}(x)=\sum^{k}_{j=0}a_{j}\mathcal{T}_{j}(x) 
\end{equation}
providing a $\nu$-approximation to the target real-valued function $f:[-1,1]\rightarrow\mathbb{R}$, that is, $\max_{x\in [-1,1]}\abs{f(x)-\Tilde{f}(x)} \le \nu$. 
The Chebyshev polynomials $\mathcal{T}_{j}$ form a key basis for function approximation, often leading to near-optimal approximation errors \cite{trefethen_approx}. In particular, unless the target function is periodic and smooth, they tend to outperform Fourier approximations \cite{boyd_spectral_methods}. 
The case of complex-valued functions can be treated similarly by splitting it into its real and imaginary parts. 
The truncation order $k$ is controlled by the desired accuracy $\nu$ in a problem-specific way (see Sec. \sect{usecases} for explicit examples). 
We denote by $\vv{a}:=\big\{a_0,\ldots,a_k\big\}$ the vector of  Chebyshev coefficients of $\Tilde{f}$ and by $\norm{\vv{a}}_1:=\sum_{j=0}^k|a_j|$ its $\ell_1$-norm.

\section{Results} \label{sec:results}

We are now in a position to state our main results. First, we set up explicitly the two problems {in question} 
and then proceed to describe 
our randomized semi-quantum algorithm to solve each one of them, proving correctness, runtime, and performing an error-robustness analysis. We conclude by applying our general framework to a number of exemplary use cases of interest.

\begin{figure*}[ht!]
    \centering
    \includegraphics[width=0.95\textwidth]{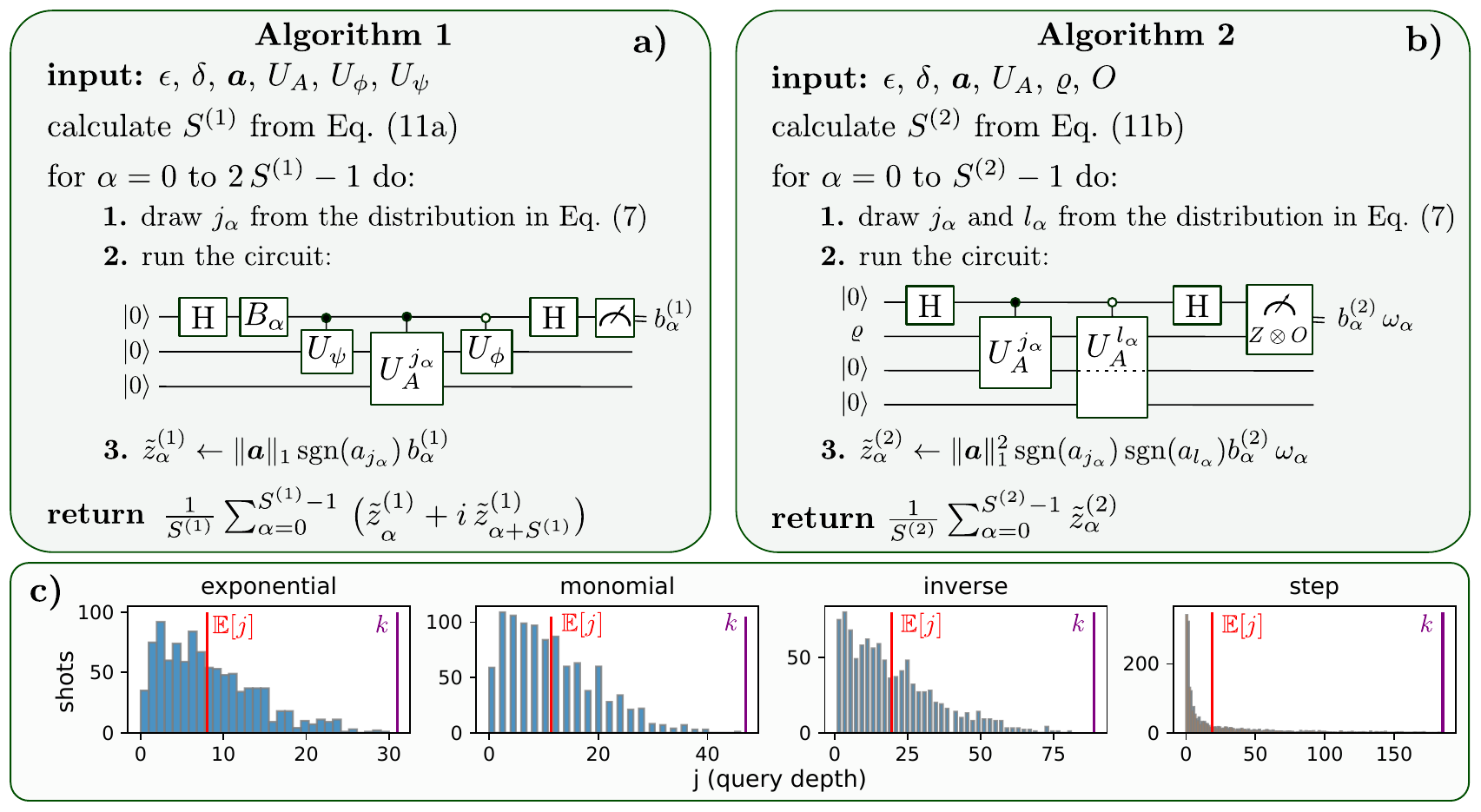}
    \caption{Alg. 1 in panel {\bf a)} solves Problem \ref{problem.1}, whereas Alg. 2 in panel {\bf b)} solves Problem \ref{problem.2}. {\bf a-b)} The algorithms receive as inputs: 
    $i)$ a qubitized block-encoding $U_A$ of $A$; 
    $ii)$ the vector $\vv{a}$ of Chebyshev coefficients defining the polynomial approximation $\tilde{f}$ to the target function $f$; 
    $iii)$ state preparation unitaries $U_\phi$ and $U_\psi$ (for Alg. 1), or the state $\varrho$ and the observable $O$ (for Alg. 2); 
    $iv)$ the tolerated error $\epsilon$ and failure probability $\delta$ for the statistical estimation. The algorithm repeats a number $\frac{2}{P}S^{(P)}$ of times two basic sampling steps. 
    The first step is to classically sample a Chebyshev polynomial degree $j_\alpha$ with probability $p(j_\alpha)=|a_{j_\alpha}|/\norm{\vv{a}}_1$. 
    The second step -- the only quantum subroutine -- is a Hadamard test (including a measurement of $O$, for Alg. 2) containing $j_\alpha$ successive queries to the controlled version of $U_A$ (plus another sequence of $l_\alpha$ queries but with the control negated and a different oracle-ancilla register for Alg. 2). 
    Finally, the average over all the measurement outcomes gives  the statistical estimate of the quantity of interest $z^{(P)}$, for $P=1$ or 2. 
    Interestingly, the Hadamard test automatically extracts the correct block of $U_A$, which relaxes the need for post-selection on the oracle ancillae. 
    Therefore, every experimental run contributes to the statistics (i.e., no measurement shot is wasted).  
    {\bf c)} Histograms of number of times (shots) a Chebyshev polynomial degree $j$ is drawn out of $1000$ samples, for the four use cases described in Sec. \sect{usecases}. 
    The vertical lines show the  maximal Chebyshev degree $k$ (purple) and the average degree $\mathbb{E}[j]$ (red). Importantly, for this figure, we do not estimate $k$ analytically using approximation theory. The values of $k$ plotted are numerically obtained as the minimum degree of $\tilde{f}$ such that the target error $\nu$ is attained. 
    The parameters used are: $\nu=10^{-2}$ (all examples),  $\beta=100$ (exponential function), 
 $t=200$ (monomial), $\kappa=8$ (inverse function), $\xi=20$ (step function). 
 In all cases, we observe a significant reduction in query complexity. This translates in practice into shallower circuits and hence less accumulated noise (see Sec. \sect{error}).}
    \label{fig:main}
\end{figure*}

\subsection{Problem statement}\label{sec:problem_statement}

We consider the following two concrete problems (throughout the paper we will use superscripts $^{(1)}$ or $^{(2)}$ on quantities referring to Problems \ref{problem.1} or \ref{problem.2}, respectively): 
\begin{problem}[Transformed vector amplitudes]
\label{problem.1}
Given access to state preparation unitaries $U_{\phi}$ and $U_{\psi}$ such that $U_{\psi}\ket{0}=\ket{\psi}$, $U_{\phi}\ket{0}=\ket{\phi}$, a Hermitean matrix $A$, and a real-valued function $f$, obtain an estimate of 
\begin{equation}\label{eq:problem1}
z^{(1)} = 
\mel{\phi}{f(A)}{\psi}
\end{equation}
to additive precision $\epsilon$ with failure probability at most $\delta$.
\end{problem} 
\noindent This class of problems is relevant for estimating the overlap between a linearly transformed state and another state of interest. This is the case, e.g., in linear system solving, where one is interested in the $i$-th computational basis component of a quantum state of the form $A^{-1}\ket{\vv{b}}$ encoding the solution to the linear system (see Sec. \sect{linearsystem} for details). 
The unitary $U_\phi$ preparing the computational-basis state $\ket{i}$, in that case, is remarkably simple, given by a sequence of bit flips. 

\begin{problem}[Transformed observable expectation values]
\label{problem.2}
Given access to a state preparation $\varrho$, a Hermitian matrix $A$, an observable $O$, and a real-valued function $f$, obtain an estimate of 
\begin{equation}\label{eq:problem2}
z^{(2)} = 
\Tr[O\,f(A)\,\varrho\,f(A)^{\dagger}]
\end{equation}
to additive precision $\epsilon$ with failure probability at most $\delta$. 
\end{problem} 
\noindent This is of relevance, e.g., when $A=H$ is a Hamiltonian, to estimate the partition function corresponding to $H$, as discussed below in Sec. \sect{partitionfunction}. 

We present randomized hybrid classical-quantum algorithms for these problems using Chebyshev-polynomial approximations of $f$ and coherent access to a block-encoding of $A$. Similar problems have been addressed in \cite{wang2023qubitefficient} but using Fourier approximations and randomizing also over a classical description of $A$ in the Pauli basis.

\subsection{Randomized semi-quantum matrix processing} 
\label{sec:framework}

Our framework is based on the Chebyshev approximation $\Tilde{f}$ of the function $f$ and a modified Hadamard test involving the qubitized block-encoding oracle $U_A$. 
The idea is to statistically simulate the coherent QSP algorithm using a hybrid classical/quantum procedure based on randomly choosing $j\in[k+1]$ according to its importance for Eq. \eqref{eq:ftilde} and then running a Hadamard test involving the block encoding $U_A^j$ of $\mathcal{T}_j(A)$.
Pseudo-codes for the algorithms are presented in Fig. \fig{main}.~{\bf a)} and \fig{main}.~{\bf b)} for Problems \ref{problem.1} and \ref{problem.2}, respectively. 
In both cases, the Hadamard test is the only quantum sub-routine. 
The total number of statistical runs will be $\frac{2}{P}\,S^{(P)}$, with $P=1$ or $2$, where $S^{(P)}$ will be given in Eqs. \eqref{eq:SP} below. 
The factor $\frac{2}{P}$ is a subtle difference between Algorithms \ref{problem.1} and \ref{problem.2} coming from the fact that the target quantity is a complex-valued amplitude in the former case, while in the latter it is a real number. This implies that two different types of Hadamard tests (each with $S^{(1)}$ shots) are needed to estimate the real and imaginary parts of $z^{(1)}$, while $z^{(2)}$ requires a single one. More technically, the procedure goes as follows. First, for every $\alpha\in[\frac{2}{P}\,S^{(P)}]$ run the following two steps:

\begin{enumerate}[wide, labelwidth=!, labelindent=10pt]
\item[$i)$]{Classical subroutine: sample a Chebyshev polynomial degree $j_\alpha\in[k+1]$ (and also $l_\alpha$ for $P=2$) from a probability distribution weighted by the coefficients $\vv{a}$ of $\Tilde{f}$, defined by
\begin{align}\label{eq:pj}
p({j})=\frac{|a_{j}|_{\phantom{1}}}{\norm{\vv{a}}_1},\quad \text{ for all } j\in[k+1]\,.
\end{align}
This has classical runtime ${\tilde{\mathcal{O}}}(k)$.}
\item[$ii)$]{Quantum subroutine: if $P=1$, run the Hadamard test in Fig. \fig{main} {\bf a)}
with $B_\alpha=\mathds{1}$ for $\alpha< S^{(1)}$ or $B_\alpha=S^\dagger:=\ketbra{0}-i\ketbra{1}$ for $\alpha\ge S^{(1)}$ and use the resulting random bit $b^{(1)}_{\alpha}\in\{-1,1\}$ to record a sample of the variable
\begin{equation}\label{eq:z1tilde}
\Tilde{z}^{(1)}_\alpha:= \norm{\vv{a}}_{1}\,\sgn(a_{j_\alpha})\,b^{(1)}_{\alpha}\,.
\end{equation}
If $P=2$, in turn, run the test in Fig. \fig{main} {\bf b)} to get as outcomes a random bit $b^{(2)}_{\alpha}\in\{-1,1\}$ and a random number $\omega_{\alpha}\in\{o_m\}_{m\in [D]}$ where $o_m$ is the $m$-th eigenvalue of $O$, and use this to record a sample of 
\begin{equation}\label{eq:z2tilde}
\Tilde{z}^{(2)}_\alpha:= \norm{\vv{a}}_{1}^2\,\sgn(a_{j_\alpha})\sgn(a_{l_\alpha})\,b^{(2)}_{\alpha}\,\omega^{\phantom{(2)}}_{\alpha}.
\end{equation}
}
\end{enumerate}
Then, in a final classical step, obtain the desired estimate $\Tilde{z}^{(P)}$ by computing the empirical mean over all the recorded samples as follows
\begin{subequations}
\label{eq:ztildeP}
\begin{align}
    \Tilde{z}^{(1)} &= \frac{1}{S^{(1)}}\sum_{\alpha=0}^{S^{(1)}-1} \left(\Tilde{z}^{(1)}_{\alpha}+i\,\Tilde{z}^{(1)}_{\alpha+S^{(1)}}\right),\label{eq:ztildeP1}\\
    \Tilde{z}^{(2)} &= \frac{1}{S^{(2)}}\sum_{\alpha=0}^{S^{(2)}-1} \Tilde{z}^{(2)}_{\alpha}\,.\label{eq:ztildeP2}
\end{align}
\end{subequations}

The following two theorems respectively prove the correctness of the estimator and establish the complexity of the algorithm. 
A simple but crucial auxiliary result for the correctness is the observation that the Hadamard test statistics (i.e. the expectation value of $b^{(P)}_{\alpha}$)  depends only on the correct block of $U_A^j$, removing the need of post-selection.  
With this, in App. \sect{main_lemma_proof}, we prove the following. 

\begin{theorem}[Correctness of the estimator]\label{main_lemma}
The empirical means $\Tilde{z}^{(1)}$ and $\Tilde{z}^{(2)}$
are unbiased estimators of  $\mel{\phi}{\Tilde{f}(A)}{\psi}$ and $\Tr[O\,\Tilde{f}(A)\,\varrho\,\Tilde{f}(A)^{\dagger}]$, respectively.     
\end{theorem}
\noindent Importantly, since $\Tilde{f}$ is a $\nu$-approximation to $f$, the obtained $\Tilde{z}^{(P)}$ are actually biased estimators of the ultimate quantities of interest $z^{(P)}$ in Eqs. \eq{problem1} and \eq{problem2}.
Such biases are always present in quantum algorithms based on approximate matrix functions, including the fully-coherent schemes for QSP \cite{LowChuang2017,LowChuangQuantum2019} and QSVT \cite{Gilyen2019,ChuangGrandUnification}.
Nevertheless, they can be made arbitrarily small in a tunable manner by increasing the truncation order $k$ in Eq. \eqref{eq:ftilde}.

Here, it is convenient to set $k$ so that $\nu^{(P)}\leq\epsilon/2$, where $\nu^{(1)}:=\nu$ and $\nu^{(2)}:=\nu\,(2\,\norm{f(A)}\,\norm{O}+\nu)$. 
This limits the approximation error in Eqs. \eq{problem1} or \eq{problem2} to at most $\epsilon/2$. In addition, demanding the statistical error to be also $\epsilon/2$, leads to (see App. \sect{main_lemma2_proof}) the following end-to-end sample and oracle-query complexities for the algorithm. 

\begin{theorem}[Complexity of the estimation]\label{main_lemma2}
Let $\epsilon>0$ and $\delta>0$ be respectively the tolerated additive error and failure probability; let $\vv{a}$ be the vector of coefficients in Eq. \eq{ftilde} and $\nu^{(P)}\leq\epsilon/2$ the error in $z^{(P)}$ from approximating $f$ with $\Tilde{f}$. Then, if the number of samples is at least
\begin{subequations}\label{eq:SP}
\begin{numcases}{S^{(P)} =}
        \frac{16\norm{\vv{a}}^{2}_{1}}{\epsilon^{2}}\log\frac{4}{\delta}\,, & \text{for } P=1, \label{eq:SP1}\\
        \frac{8\,\norm{O}^2\norm{\vv{a}}^{4}_{1}}{\epsilon^{2}}\log\frac{2}{\delta}\,, & \text{for } P=2, \label{eq:SP2}
\end{numcases}
\end{subequations}
Eqs. \eqref{eq:ztildeP} give an $\epsilon$-precise estimate of $z^{(P)}$ with confidence $1-\delta$. 
Moreover, the total expected runtime is $Q^{(P)}:= 2\,\mathbb{E}[j]\,S^{(P)}$, where $\mathbb{E}[j]:=\sum_{j=0}^k j\,p(j)$. 
\end{theorem}

A remarkable consequence of this theorem is that the expected number of queries per statistical run is $P\times\mathbb{E}[j]$. 
Instead, if we used standard QSVT (together with a similar Hadamard test to avoid post-selection), each statistical run would take $P\times k$ queries (and an extra ancillary qubit coherently controlling everything else would be required).
As shown in Fig. \fig{main}.~{\bf c)}, $\mathbb{E}[j]$ can be significantly smaller than $k$ in practice. 
In fact, in Sec. \sect{usecases}, we prove scaling advantages of $\mathbb{E}[j]$ over $k$.
These query-complexity advantages translate directly into reductions in circuit depth and, hence, also in noise sensitivity (see next sub-section).
As for sample complexity, the statistical overhead of our semi-quantum algorithms scales with $\norm{\vv{a}}_{1}$, while that of fully-quantum ones would have a similar scaling with $\norm{f(A)}$, due to the required normalization for block encoding. 
Interestingly, in all the use cases analyzed, $\norm{\vv{a}}_1$ and $\norm{f(A)}$ differ at most by a logarithmic factor. 
Finally, another appealing feature is that our approach relaxes the need to compute the QSP/QSVT angles, which is currently tackled with an extra classical pre-processing stage of runtime $\mathcal{O}\big(\text{poly}(k)\big)$ \cite{LowChuang2017,LowChuangQuantum2019,Gilyen2019,ChuangGrandUnification}.

We emphasize that here we have assumed Hermitian $A$ for the sake of clarity, but a straightforward extension of our randomized scheme from QSP to QSVT (see App. \sect{QSVT})
gives the generalization to generic $A$. Moreover, in Lemma \ref{lemma_hadamardsecondkind} in App. \ref{sec:Hadamard_test}, we also extend the construction to Chebyshev polynomials of the second kind. This is useful for ground-state energy estimation, in Sec. \sect{groundstate}.


\subsection{Intrinsic noise-sensitivity reduction} 
\label{sec:error}

Here we study how the  reduction in query complexity per run from $k$ to the average value $\mathbb{E}[j]$ translates into sensitivity to experimental noise.
The aim is to make a quantitative but general comparison between our randomized semi-quantum approach and fully-quantum schemes, remaining agnostic to the specific choice of operator function, circuit compilation, or physical platform.
To this end, we consider two toy error models that allow one to allocate one unit of noise per oracle query. 

Our first error model consists of a faulty quantum oracle given by the ideal oracle followed by a globally depolarizing channel $\Lambda$ of noise strength $p$, defined by \cite{Aolita15review}
\begin{equation}
\label{eq:depol}
\Lambda[\varrho]:=(1-p)\,\varrho + p\,\frac{\mathds{1}}{D_\text{tot}}.
\end{equation}
Here, $\varrho$ is the joint state of the total Hilbert space in Fig. \fig{main}{\bf a} (system register, oracle ancilla, and Hadamard test ancilla) and $D_\text{tot}$ its dimension.
In App. \sect{error_analysis_proof} we prove:
\begin{theorem}[Average noise sensitivity]\label{thm:error_robust}
Let $\Tilde{z}^{(P)}$  be the ideal estimators \eq{ztildeP} and $\Tilde{z}^{(P),\Lambda}$ 
their noisy version with $\Lambda$ acting after each oracle query in Fig. \fig{main}. Then 
\begin{subequations}\label{eq:noise_robustness}
\begin{align}
\abs{\mathbb{E}\big[\Tilde{z}^{(1)}\big]-\mathbb{E}\big[\Tilde{z}^{(1),\Lambda}\big]} &\le
p\,E^{(1)}_\text{sq}\leq\, p\,\norm{\vv{a}}_{1}\,\mathbb{E}[j] \,,\label{eq:noise_robustness1}\\    
\abs{\mathbb{E}\big[\Tilde{z}^{(2)}\big]-\mathbb{E}\big[\Tilde{z}^{(2),\Lambda}\big]} &\le p\,E^{(2)}_\text{sq}
\leq\, 2\,p\,\norm{\vv{a}}_{1}^2\,\mathbb{E}[j],\label{eq:noise_robustness2}
\end{align}    
\end{subequations}
where $E^{(1)}_\text{sq} := \abs{\sum^{k}_{j=0}j\,a_{j}\mel**{\phi}{\mathcal{T}_j(A)}{\psi}}$ and $E^{(2)}_\text{sq} := \abs{\sum^{k}_{j,l=0}(j+l)\,a_{j}\,a_{l}\Tr{O\,\mathcal{T}_j(A)\,\varrho\,\mathcal{T}_l(A)}}$.
\end{theorem}
\noindent Our second model is coherent errors that make the quantum oracle no longer the exact block encoding $U_A$ of $A$ but only a $\varepsilon$-approximate block encoding (a unitary with operator-norm distance $\varepsilon$ from $U_A$).
In App. \sect{error_analysis_proof}, we show that Eq. \eq{noise_robustness} holds also there  with  $p$ replaced by $2\varepsilon$.

It is instructive to compare Eq.\ \eq{noise_robustness} with the inaccuracy for the corresponding fully-quantum scheme. 
A fair scenario for that comparison (in the case of Problem 1) is to equip the standard QSVT with a Hadamard test similar to the ones in Fig. \fig{main} so as to also circumvent the need for post-selection.
Notice that, while in our randomized method, only the Hadamard ancilla controls the calls to the oracle, the standard QSVT circuit involves two-qubit control to also implement the pulses that determine the Chebyshev coefficients.
As a consequence, the underlying gate complexity per oracle query would be considerably higher than for our schemes (with single-qubit gates becoming two-qubit gates, two-qubit gates becoming Toffoli gates, etc). For this reason, the resulting noise strength  $p_\text{fq}$ is expected to be larger than $p$.
The left-hand side of Eq. \eq{noise_robustness1} would then (see App. \sect{error_analysis_proof}) be upper-bounded by $p_\text{fq}\,E_\text{fq}$, with $E_\text{fq} = k\,|\mel{\phi}{\Tilde{f}(A)}{\psi}|$, where $p_\text{fq}> p$ and $k>\mathbb{E}[j]$. 

Another natural scenario for comparison is that where the fully-quantum algorithm does not leverage a Hadamard test but implements post-selection measurements on the oracle ancilla, in a repeat-until-success strategy. This comparison applies only to Problem \ref{problem.2}, since one cannot directly measure the complex amplitudes for Problem \ref{problem.1}. The advantage though is that the circuits are now directly comparable because the gate complexities per oracle query are essentially the same (the fully-quantum scheme has extra QSP pulses, but these are single-qubit gates whose error contribution is low).
Hence, similar error rates to $p$ are expected here, so that one would have the equivalent of Eq. \eq{noise_robustness2} being $\mathcal{O}(k\, p)$.
This is already worse than Eq. \eq{noise_robustness2} because $k>\mathbb{E}[j]$, as already discussed. 
Moreover, with post-selection, one additionally needs to estimate normalizing constants with an independent set of experimental runs, which incurs in extra systematic and statistical errors. 
In contrast, our method does not suffer from this issue, as it directly gives the estimates in Eqs. \eq{problem1} or \eq{problem2} regardless of state normalization (see Sec. \sect{usecases}). 

Finally, a third possibility could be to combine the fully-quantum scheme with quantum amplitude amplification to manage the post-selection. This would quadratically improve the dependence on the post-selection probability.
However, it would then be the circuit depth that would gain a factor inversely proportional to the square root of the post-selection probability.
Unfortunately, this is far out of reach of early-fault tolerant hardware.

\subsection{End-user applications}
\label{sec:usecases}

Here we illustrate the usefulness of our framework with four use cases of practical relevance: partition function estimation (both for classical or general Hamiltonians), linear system solving, and ground-state energy estimation. These correspond to $f(x)=x^t$, $e^{-\beta x}$, $x^{-1}$, and $\theta(x)$, respectively. 
The end-to-end complexities for each case are summarized in Table \tab{main_table}.

\begin{table*}[ht]
    \centering
    \begin{tabular}{||c||c|c|c|c||}
     \hline
     \rule{0pt}{3ex}
     \rule[-1.5ex]{0pt}{0pt}
     {\bf Problem} & {\bf App.} & {\bf Maximal query depth} 
     & {\bf Expected query depth} & {\bf Total expected runtime} \\
     \hline
     \hline
     \rule{0pt}{4ex}
     \rule[-1.5ex]{0pt}{0pt}
     Part. funct. (MCMC) & \sect{explicit_examples_monomials} & $\sqrt{\frac{2}{\Delta}}\log\Big(\frac{12\,Z_\beta\, e^{\beta E_{\vv{y}}}}{\epsrel}\Big)$
     & $\sqrt{\frac{2}{\pi\,\Delta}\log\Big(\frac{12\,Z_\beta\, e^{\beta E_{\vv{y}}}}{\epsrel}\Big)}$ 
     & $\mathcal{O}\left(\frac{e^{2\beta E_{\vv{y}}} Z_\beta^2}{\sqrt{\Delta}}\sqrt{\log\Big(\frac{Z_\beta e^{\beta E_{\vv{y}}}}{\epsrel}\Big)}\frac{\log(1/\delta)}{\epsrel^2}\right)$ \\
     \hline
     \rule{0pt}{4ex}
     \rule[-2ex]{0pt}{0pt}
     Part. funct. (QITE) & \sect{explicit_examples_qite} & $\mathcal{O}\Big(\sqrt{\beta}\log\big(\frac{D\, e^{\beta}}{Z_\beta\,\epsrel}\big)\Big)$ 
     & $\mathcal{O}\left(\sqrt{\beta}\right)$ & ${\mathcal{O}}\left(\frac{D^2 \sqrt{\beta}\,e^{2\beta}}{Z_{\beta}^2}\frac{\log(1/\delta)}{\epsrel^2}\right)$ \\
     \hline
     \rule{0pt}{4ex}
     \rule[-2ex]{0pt}{0pt}     
     QLSS: $\bra{i}A^{-1}\ket{\vv{b}}$ & \sect{explicit_examples_inverse} & $\mathcal{O}\left(\kappa\log\big(\frac{\kappa}{\epsilon}\big)\right)$ 
     & $\mathcal{O}\left(\kappa\,\sqrt{\log\big(\frac{\kappa}{\epsilon}\big)}\right)$ & $\mathcal{O}\left(\kappa^3\log^{5/2}\!\big(\frac{\kappa}{\epsilon}\big)\frac{\log(1/\delta)}{\epsilon^2}\right)$ \\
     \hline
     \rule{0pt}{4ex}
     \rule[-2ex]{0pt}{0pt} 
     QLSS: $\bra{\vv{b}}A^{-1}OA^{-1}\ket{\vv{b}}$ & \sect{explicit_examples_inverse} & $\mathcal{O}\left(\kappa\log\Big(\frac{\kappa^2\norm{O}}{\epsilon}\Big)\right)$ 
     & $\mathcal{O}\left(\kappa\,\sqrt{\log\Big(\frac{\kappa^2\norm{O}}{\epsilon}\Big)}\right)$ & $\mathcal{O}\left(\kappa^5\norm{O}^2\log^{9/2}\!\Big(\frac{\kappa^2\norm{O}}{\epsilon}\Big)\frac{\log(1/\delta)}{\epsilon^2}\right)$ \\
     \hline
     \rule{0pt}{4ex}
     \rule[-2ex]{0pt}{0pt}      
         Ground-state energy & \sect{explicit_examples_step} & $\mathcal{O}\left(\frac{1}{\xi}\log\big(\frac{1}{\eta}\big)\right)$ 
         & $\quad\mathcal{O}\left(\frac{1}{\xi}\frac{\sqrt{\log(1/\eta)}}{\log((1/\xi)\log(1/\eta))}\right)$
         & $
    \mathcal{O}\left(\frac{1}{\eta^{2}\xi}\sqrt{\log\big(\frac{1}{\eta}\big)}\log\big(\frac{1}{\xi}
    \big)\log\big(\frac{1}{\delta}
    \big)\right)
    $ \\   
     \hline
    \end{tabular}
    \caption{
    {\bf Complexities of our algorithms for end-user applications}. 
    The first column indicates the specific use case (see Sec. \sect{usecases}).
    The second one indicates the appendix with the corresponding derivations. 
    The third column shows the maximal query complexity per run $k$. Chebyshev-based fully-quantum matrix processing (using the same Hadamard tests as us) would require the same query depth but in {\it every} run.  
    The fourth column displays the average query complexity per run $P\,\mathbb{E}[j]$, with $P=1$ for Alg. 1 and $P=2$ for Alg. 2. 
    We notice that in the last row we use $\xi$ for the additive error in the ground state energy (coming from the $\mathcal{O}\big(\log\big(\frac{1}{\xi}\big)\big)$ steps in the binary search) to distinguish from the $\epsilon$ (which here is $\mathcal{O}(\eta)$) reserved for the estimation error in the quantities $z^{(P)}$. 
    As can be seen, in all use cases, the average query depth features a better scaling than $k$ on certain parameters. 
    This is an interesting speed-up specific to the randomization over the Chebyshev expansion.
    Finally, the fourth column shows the expected runtime, given by $Q^{(P)}$ in Theorem \ref{main_lemma2}, namely the average query depth times the sample complexity $S^{(P)}$. 
    Here, $S^{(P)}$ scales with $\norm{\vv{a}}_1$ exactly as it would with $\norm{f(A)}$ had we used the fully-quantum algorithm. Interestingly, $\norm{\vv{a}}_1$ and  $\norm{f(A)}$  happen to be of the same order for the use cases studied, except for small logarithmic corrections for QLSSs and ground-state energy estimation (see Table \tab{chebyshev_data} in App. \sect{explicit_examples} for details). 
    All in all, the total expected runtimes are either similar or slightly superior to the corresponding runtimes of Chebyshev-based fully-quantum approaches. 
    Remarkably, this is achieved in tandem with important advantages in terms of quantum hardware (see, e.g., Sec. \sect{error}). 
    }
    \label{tab:main_table}
\end{table*}

\subsubsection{Relative-error partition function estimation} \label{sec:partitionfunction}

Partition function estimation is a quintessential hard computational problem, with applications {ranging from} statistical physics {to generative} machine learning, as in Markov
random fields \cite{Ma_Peng_Wang}, Boltzmann machines \cite{KRAUSE2020103195}, and even the celebrated transformer architecture \cite{shim2022probabilistic} from large language models. 
Partition functions also appear naturally in other problems of practical relevance, such as constraint satisfaction problems \cite{BULATOV2005148}.

The partition function of a Hamiltonian $H$ at inverse temperature $\beta$ is defined as
\begin{equation}\label{eq:partitionfunction}
Z_{\beta}=\Tr\left[ e^{-\beta H}\right]\,.
\end{equation}
One is typically interested in the problem of estimating $Z_\beta$ to relative error $\epsrel$, that is, finding $\Tilde{Z}_\beta$ such that
\begin{align}
   \big|\Tilde{Z}_\beta - Z_\beta\big| \le \epsrel\,Z_\beta\,.
\end{align}
This allows for the estimation of relevant thermodynamic functions, such as the Helmholtz free energy $F=\frac{1}{\beta}\log Z_{\beta}$, to additive precision. The naive classical algorithm based on direct diagonalization runs in time $\mathcal{O}(D^{3})$, where $D=\text{dim}(\mathcal{H}_s)$ is the Hilbert space dimension. Although it can be improved to $\mathcal{O}(D)$ using the kernel polynomial method \cite{RevModPhys.78.275} if $H$ is sparse, one expects no general-case efficient algorithm to be possible due to complexity theory arguments \cite{bravyi2021complexity}. In turn, if the Hamiltonian is classical (diagonal), $Z_{\beta}$ can be obtained exactly in classical runtime $\mathcal{O}(D)$.  General-purpose quantum algorithms (that work for any inverse temperature and any Hamiltonian) have been proposed 
\cite{PhysRevLett.103.220502,chowdhury_computing_2021,PhysRevA.107.012421}. The list includes another algorithm \cite{chowdhury_computing_2021} that, like ours, utilizes the Hadamard test and a block-encoding of the Hamiltonian. 

In the following, we present two different quantum algorithms for partition function estimation: one for classical Ising models, based on the Markov-Chain Monte-Carlo (MCMC) method, and another for generic non-commuting Hamiltonians, based on quantum imaginary-time evolution (QITE) simulation \cite{Sunetal21,silva_fragmented_2022}.

\paragraph{Partition function estimation via MCMC:}

Here, we take $H$ as the Hamiltonian of a classical Ising model. As such, spin configurations, denoted by $\ket{\vv{y}}$, are eigenstates of $H$ with corresponding energies $E_{\vv{y}}$. Let us define the {coherent version of the} Gibbs state $\ket{\sqrt{\bpi}} := Z_\beta^{-1/2}\sum_{\vv{y}}e^{-\beta E_{\vv{y}}/2}\ket{\vv{y}}$.  
 Then, for any $\ket{\vv{y}}$, the partition function satisfies the identity
\begin{align}\label{eq:Zidentity}
    Z_\beta = \frac{e^{-\beta E_{\vv{y}}}}{\mel{\vv{y}}{\Pi_{{\bpi}}}{\vv{y}}} 
\end{align}
with $\Pi_{{\bpi}}:=\ketbra{\sqrt{\bpi}}$. Below we discuss how to use our framework to obtain an estimation of $\mel{\vv{y}}{\Pi_{{\bpi}}}{\vv{y}}$ for a randomly sampled $\ket{\vv{y}}$ and, therefore, approximate the partition function. 

Let $A$ be the discriminant matrix\ \cite{szegedy2004} of a Markov chain having the Gibbs state of $H$ at inverse temperature $\beta$ as its unique stationary state. The Szegedy quantum walk unitary\ \cite{szegedy2004} provides a qubitized block-encoding $U_A$ of $A$ that can be efficiently implemented\ \cite{Lemieux2020efficientquantum}. A useful property of $A$ is that the monomial $A^t$ approaches $\Pi_{\bpi}$  for sufficiently large integer $t$ (the precise statement is given by Lemma\ \ref{lemma_partition_function_power_method} in App.\ \sect{explicit_examples_monomials}). This implies that  $\mel{\vv{y}}{\Pi_{{\bpi}}}{\vv{y}}$ can be estimated using Alg. 1  with $f(A)=A^t$ and $\ket{\psi}=\ket{\phi}=\ket{\vv{y}}$. 
In this case, the state preparation unitaries $U_{\psi}=U_{\phi}$ will be simple bit flips. 

A $\nu$-approximation $\Tilde{f}(A)$ can be constructed by truncating the Chebyshev representation of $A^t$ to order $k=\sqrt{2\,t\log(2/\nu)}$  \cite{Vishnoi2013}. The {$l_1$}-norm of the {corresponding} coefficient vector is $\norm{\vv{a}}_1=1-\nu$.  For this Chebyshev series, the ratio $\mathbb{E}[j]/k$ between the average and the maximum query complexities can be shown (see Lemma \ref{lem:av_Q_mon} in App.\ \sect{explicit_examples_monomials}) to be at most  $(1-\nu)^{-1}/{\sqrt{\pi\,\log(2/\nu)}}
$ for large $t$. 
This implies that the more precise the estimation, the larger the advantage of the randomized algorithm in terms of total expected runtime.
{For instance, for $\nu=10^{-2}$, the ratio is roughly equal to 0.25.}

To estimate the partition function up to relative error $\epsrel$, Alg. 1 needs to estimate $\mel{\vv{y}}{\Pi_{{\bpi}}}{\vv{y}}$ with additive error $\epsilon=\frac{ e^{-\beta E_{\vv{y}}}}{2\,Z_\beta}\epsrel$ (see Lemma \ref{lem:estimate_z} in App.\ \sect{explicit_examples_monomials}).
In Lemma \ref{lem:complexity_mon}, in App.\ \sect{explicit_examples_monomials}, we show that the necessary $t$ and $\nu$ required for that yield a  maximum query complexity per run of $k=
\sqrt{\frac{2}{\Delta}}\log(\frac{12\,Z_\beta\, e^{\beta E_{\vv{y}}}}{\epsrel})$ and an average query complexity of $\mathbb{E}[j]=\sqrt{\frac{2}{\pi\,\Delta}\log(\frac{12\,Z_\beta\, e^{\beta E_{\vv{y}}}}{\epsrel})}$, where $\Delta$ is the spectral gap of $A$.
Moreover, from Theorem \ref{main_lemma2}, the necessary sample complexity is $S^{(1)}=64\,e^{2\beta E_{\vv{y}}}\, Z_\beta^2\,\frac{\log(2/\delta)}{\epsrel^2}$. 
This leads to the total expected runtime in Table \tab{main_table}.  

Three important observations about the algorithm's complexities are in place. First, the total expected runtime has no explicit dependence on the Hilbert space dimension $D$ and maintains the square-root dependence on $\Delta$ (a Szegedy-like quadratic quantum speed-up \cite{szegedy2004}).
Second, all three complexities in the first row of the  table depend on the product $Z_\beta\, e^{\beta E_{\vv{y}}}=\mathcal{O}\big(e^{\beta (E_{\vv{y}}-E_\text{min})}\big)
$, with $E_\text{min}$ the minimum eigenvalue of $H$, where the scaling holds for large $\beta$.
This  scaling plays more in our favor the lower the energy $E_{\vv{y}}$ of the initial state $\vv{y}$ is.
Hence, by uniformly sampling a constant number of different bit-strings $\vv{y}$ and picking the lowest energy one, one ensures to start with a convenient initial state.
Third, the quadratic advantage featured by $\mathbb{E}[j]$ over $k$ on the logarithmic term is an interesting type of speed-up entirely due to the randomization over the components of the Chebyshev series.

To end up with, the total expected runtime obtained can potentially provide a quantum advantage over classical estimations in regimes where $\frac{e^{2\beta (E_{\vv{y}}-E_\text{min})}}{\sqrt{\Delta}\,\epsrel^2}<D$.

\paragraph{Partition function estimation via QITE:} 
Alternatively, the partition function associated with a Hamiltonian $H$ can be estimated by quantum simulation of imaginary time evolution (QITE). This method applies to any Hamiltonian (not just classical ones), assuming a block-encoding of $H$.  
$Z_{\beta}$ can be written in terms of the expectation value of the QITE propagator $e^{-\beta H}$ over the maximally mixed state $\varrho_0:=\frac{\mathds{1}}{D}$, that is,
\begin{equation}
Z_{\beta}=D\,\Tr\big[e^{-\beta H}\varrho_0\big]\,.
\end{equation}
Therefore, we can apply our Alg. 2 with $A=H$, $O=D\mathds{1},\varrho=\varrho_0$, and $f(H)=e^{-\beta H/2}$ to estimate $Z_{\beta}$ with relative precision $\epsrel$ and confidence $1-\delta$. The sample complexity is obtained from Eq.\ \eqref{eq:SP2}  as $S^{(2)}=\frac{8\, D^2e^{2\beta}}{\epsrel^2 Z_{\beta}^2}\log\frac{2}{\delta}$, by setting the additive error equal to $Z_\beta\,\epsrel$. 

We use the Chebyshev approximation of the exponential function introduced in Ref.\ \cite{Vishnoi2013}, which has a quadratically better asymptotic dependence on $\beta$ than other well-known expansions such as the Jacobi-Anger decomposition \cite{silva_fragmented_2022}.  This expansion was used before to implement the QITE propagator using QSVT coherently \cite{Gilyen2019}.
The resulting truncated Chebyshev series has order $k=\sqrt{2\,\max\left\{\frac{e^2\beta}{2},\log\big(\frac{8D\,e^{\beta}}{Z_\beta\,\epsrel}\big)\right\}\,\log\big(\frac{16 D\,e^{\beta}}{Z_\beta\,\epsrel}\big)}$  and coefficient $l_1$-norm $\norm{\vv{a}}_1\leq e^{\beta/2}+\nu$ (see Lemmas \ref{lem:approx_exp} and \ref{lem:1norm_exp} in App.\ \sect{explicit_examples_qite}). 
 Interestingly, the average query depth does not depend on the precision of the estimation but scales as $\mathcal{O}(\sqrt{\beta})$ with a modest constant factor for any $\epsrel$ (see Lemma \ref{lem:query_exp} in App.\ \sect{explicit_examples_qite}). This implies an advantage of $\mathcal{O}\left(\log\big(\frac{D\,e^{\beta}}{Z_\beta\,\epsrel}\big)\right)$ in terms of overall runtime as compared to coherent QSVT, which is again entirely due to our randomization scheme. 
 
 Overall, {this} gives our algorithm a total {expected} runtime of $\mathcal{O}\left(\frac{D^2 \sqrt{\beta}\,e^{2\beta}}{Z_{\beta}^2}\frac{\log(2/\delta)}{\epsrel^2}\right)$. The previous state-of-the-art algorithm from Ref.\ \cite{chowdhury_computing_2021} has runtime $\Tilde{\mathcal{O}}\left(\frac{D^2D_a^2e^{2\beta}\beta^2}{\epsrel^2 Z_{\beta}^2}\log\frac{1}{\delta}\right)$. 
 Compared with that, we get an impressive
  quartic speed-up in $\beta$ together with the entire removal of the dependence on $D_a^2$. 
  The improvement comes from not estimating each Chebyshev term individually and allowing the ancillas to be pure while only the system is initialized in the maximally mixed state.

Finally, compared to the  $\order{D^3}$ scaling of the classical algorithm based on exact diagonalization, our expected runtime has a better dependence on $D$. Moreover, in the regime of small $\beta$ such that  $Z^2_\beta>\mathcal{O}\big(\sqrt{\beta}\,e^{2\beta}\log(1/\delta)/\epsrel^2\big)$, the expected runtime can be even better than that of the kernel method, which scales as $\mathcal{O}(D)$.

\subsubsection{Quantum linear-system solvers} \label{sec:linearsystem}

Given a matrix $A\in\mathbb{C}^{D}\times \mathbb{C}^{D}$ and a vector $\vv{b}\in\mathbb{C}^{D}$, the task is to find a vector $\vv{x}\in\mathbb{C}^{D}$ such that
\begin{equation}
A\,\vv{x}=\vv{b}\,.
\end{equation}
The best classical algorithm for a generic $A$ is based on Gaussian elimination, with a runtime $\mathcal{O}(D^3)$ \cite{trefethen97}. For $A$ positive semi-definite and sparse, with sparsity (i.e. maximal number of non-zero elements per row or column) $s$, the conjugate gradient algorithm \cite{book_iterative_methods} can reduce this to $\mathcal{O}(Ds\kappa)$, where $\kappa:=\norm{A}\,\norm{A^{-1}}$ is the condition number of $A$. In turn, the randomized Kaczmarz algorithm \cite{strohmer2007randomized} can yield an $\epsilon$-precise approximation of a single component of $\vv{x}$ in $\mathcal{O}\left(s\,\kappa_{F}^{2}\log(1/\epsilon)\right)$, with $\kappa_F := \norm{A}_{F}\norm{A^{-1}}$ and $\norm{A}_{F}$ the Frobenius norm of $A$. 

In contrast, quantum linear-system solvers (QLSSs) \cite{harrow_quantum_2009,Childs_2017,Gilyen2019,ChuangGrandUnification,Lin_2020,An_2022,Suba_2019,costa2021optimal,Wossnig_2018} prepare a quantum state that encodes the normalized version of the solution vector $\vv{x}$ in its amplitudes. More precisely, given quantum oracles for $A$ and $\ket{\vv{b}}:=\frac{1}{\norm{\vv{b}}_2}\sum_{i}b_i\ket{i}$ as inputs, they output the state $\ket{\vv{x}}:=\frac{1}{\|\vv{x}\|_2}\sum_{i}x_i\ket{i}$, where $\norm{\cdot}_2$ is the $l_2$-norm and we assume $\norm{A}\le1$ for simplicity of presentation (see App. \sect{QSVT} for the case of unnormalized $A$). 
Interestingly, circuit compilations of block encoding oracles for $A$ with gate complexity $\mathcal{O}\left(\log(D/\epsilon)\right)$ have been explicitly worked out assuming a QRAM access model to the classical entries of $A$ \cite{Clader_2022}. 
This can be used for extracting relevant features -- such as an amplitude $\braket{\phi}{\vv{x}}$ or an expectation value $\mel{\vv{x}}{O}{\vv{x}}$ -- from the solution state, with potential exponential speed-ups over known classical algorithms, assuming that the oracles are efficiently implementable and $\kappa=\mathcal{O}\big(\text{polylog}(D)\big)$. 

Ref.~\cite{costa2021optimal} proposed an asymptotically optimal QLSS based on a discrete version of the adiabatic theorem with query complexity $\mathcal{O}\left(\kappa\log(1/\epsilon)\right)$. 
Within the Chebyshev-based QSP framework, the best known QLSS 
uses $\mathcal{O}\left(\kappa\log(\kappa/\epsilon)\right)$ oracle queries \cite{Childs_2017}.  
If the final goal is, for instance, to reconstruct a computational-basis component $\braket{i}{\vv{x}}$ of the solution vector, the resulting runtime becomes $\mathcal{O}\left((\kappa^3/\epsilon^{2})\log(\kappa/\epsilon)\right)$, since this requires $\mathcal{O}\left(\kappa^2/\epsilon^2\right)$ measurements on $\ket{\vv{x}}$. 
Importantly, however, in order to relate the abovementioned features of $\ket{\vv{x}}$ to the corresponding ones from the (unnormalized) classical solution vector $\vv{x}$, one must also independently estimate $\|\vv{x}\|_2$. 
This can still be done with QLSSs (e.g., with quantum amplitude estimation techniques), but requires extra runs. 
Our algorithms do not suffer from this issue, providing direct estimates from the unnormalized vector $A^{-1}\ket{\vv{b}}$. 

More precisely, with $f$ being the inverse function on the cut-off interval $\mathcal{I}_\kappa:=[1,-1/\kappa]\cup[1/\kappa,1]$, our Algs. 1 and 2 readily estimate amplitudes $\bra{\phi}A^{-1}\ket{\vv{b}}$ and expectation values $\bra{\vv{b}}A^{-1}OA^{-1}\ket{\vv{b}}$, respectively. 
The technical details of the polynomial approximation $\Tilde{f}$ and complexity analysis are deferred to App. \sect{explicit_examples_inverse}. 
In particular, there we show that, to approximate $f$ to error $\nu$, one needs a polynomial of degree $k=\mathcal{O}\left(\kappa\,\log(\kappa/\nu)\right)$ and $\norm{\vv{a}}_1=\mathcal{O}\left(\kappa\sqrt{\log(\kappa^{2}/\nu)}\right)$. 
For our purposes, as discussed before theorem \ref{main_lemma2}, to ensure a target estimation error $\epsilon$ on the quantity of interest one must have $\nu=\mathcal{O}(\epsilon)$ for Alg. 1 and $\nu=\mathcal{O}((\kappa\norm{O})^{-1}\epsilon)$ for Alg. 2. 
This leads to the sample complexities $S^{(1)}=\mathcal{O}\left((\kappa^2/\epsilon^2)\log^2(\kappa^{2}/\epsilon)\log(4/\delta)\right)$ and $S^{(2)}=\mathcal{O}\left((\kappa^4\norm{O}^2/\epsilon^2)\log^4(\kappa^3\,\norm{O}/\epsilon))\log(4/\delta)\right)$, respectively. 
    
The expected query depth and total expected runtimes are shown in Table \tab{main_table}. 
In particular, the former exhibits a quadratic improvement in the error dependence with respect to the maximal query depth $k$. 
This places our algorithm in between the $\mathcal{O}(\kappa\log(\kappa/\epsilon))$ \cite{Childs_2017} scaling of the fully quantum algorithm and the asymptotically optimal $\mathcal{O}(\kappa\log(1/\epsilon))$ scaling of \cite{costa2021optimal}, therefore making it more suitable for the early fault-tolerance era. 
In fact, our expected query depth can even beat this optimal scaling for $\kappa\lesssim(1/\epsilon)^{\log(1/\epsilon)-1}$. 
Note also that our total expected runtimes are only logarithmically worse in $\kappa$ than the ones in the fully-quantum case. 
For the case of Alg. 1, an interesting sub-case is that of $\bra{\phi}=\bra{i}$, as this directly gives the $i$-th component of the solution vector $\vv{x}$. 
The quantum oracle $U_{\phi}$ is remarkably simple there, corresponding to the preparation of a computational-basis state. As for the runtime, we recall that $\norm{A}\leq\norm{A}_{F}$ in general and $\norm{A}_{F}=\mathcal{O}(\sqrt{D}\,\norm{A})$ for high-rank matrices. Hence, Alg. 1 has potential for significant speed-ups over the randomized Kaczmarz algorithm mentioned above. In turn, for the case of Alg. 2, we stress that the estimates obtained refer directly to the target expectation values for a generic observable $O$, with no need to estimate the normalizing factor $\|\vv{x}\|_2$ separately (although, if desired, the latter can be obtained by taking $O=\openone$). 
 
It is also interesting to compare our results with those of the fully randomized scheme of \cite{wang2023qubitefficient}. There, for $A$ given in terms of a Pauli decomposition with total Pauli weight $\lambda$, they also offer direct estimates, with no need of $\|\vv{x}\|_2$. However, their total runtime of $\Tilde{O}\big(\norm{A^{-1}}^{4}\lambda^{2}/\epsilon^{2}\big)$ is worse than the scaling presented here by a factor $\Tilde{O}\big(\norm{A^{-1}}\,\lambda^{2}\big)$ (recall that here $\kappa=\norm{A^{-1}}$ since we are assuming $\norm{A}=1$). In turn, compared to the solver in Ref. \cite{wang2023qubitefficient}, the scaling of our query depth per run is one power of $\kappa$ superior. In their case, the scaling refers readily to circuit depth, instead of query depth, but this is approximately compensated by the extra dependence on $\lambda^{2}$ in their circuit depth.

\subsubsection{Ground-state energy estimation} \label{sec:groundstate}

The task of estimating the ground-state energy of a quantum Hamiltonian holds paramount importance in condensed matter physics, quantum chemistry, material science, and optimization. 
In fact, it is considered one of the most promising use cases for quantum computing in the near term \cite{clinton2022nearterm}. 
However, the problem in its most general form is known to be QMA-hard \cite{kempe2005complexity}. A typical assumption -- one we will also use here -- is that one is given a Hamiltonian $H$ with $\norm{H}\leq1$ and a promise state $\varrho$ having non-vanishing overlap $\eta$ with the ground state subspace. The {\it ground state energy estimation} (GSEE) problem \cite{lin_heisenberg-limited_2022} then consists in finding an estimate of the ground state energy $E_{0}$ to additive precision $\xi$. 

If the overlap $\eta$ is reasonably large (which is often the case in practice, e.g., for small molecular systems using the Hartree-Fock state \cite{tubman2018postponing}), the problem is known to be efficiently solvable, but without any guarantee on $\eta$ the problem is challenging. 
A variety of quantum algorithms for GSEE have been proposed (see, e.g., \cite{PhysRevLett.83.5162,ge2018faster,Lin_2020,Poulin_2009}), but the substantial resources required are prohibitive for practical implementation before full-fledged fault tolerant devices become available. 
Recent works have tried to simplify the complexity of quantum algorithms for GSEE with a view towards early fault-tolerant quantum devices. 
Notably, a semi-randomized quantum scheme was proposed in \cite{lin_heisenberg-limited_2022} with query complexity $\mathcal{O}\big(\frac{1}{\xi}\log\!\big(\frac{1}{\xi\eta}\big)\big)$ achieving Heisenberg-limited scaling in $\epsilon$ \cite{Atia_2017}. 
Importantly, their algorithm assumes access to the Hamiltonian $H$ through a time evolution oracle $e^{-iH\tau}$ (for some fixed time $\tau$), which makes it more appropriate for implementation in analog devices. 
The similar fully-randomized approach of \cite{wan_randomized_2022} gives rise to an expected circuit (not query) complexity of $\mathcal{O}\big(\frac{1}{\xi^2}\log\!\big(\frac{1}{\eta}\big)\big)$.  

Here we approach the GSEE problem within our Chebyshev-based randomized semi-quantum framework. 
We follow the same strategy used in \cite{lin_heisenberg-limited_2022,wan_randomized_2022,wang2023faster,wang2023quantum} of reducing GSEE to the so-called {\it eigenvalue thresholding problem}. The problem reduces to the estimation up to additive precision $\frac{\eta}{2}$ of the filter function $F_\varrho(y):=\Tr[\varrho\,\theta(y\mathds{1}-H)]$ for a set of $\log\big(\frac{1}{\xi}\big)$ different values of $y$ chosen from a uniform grid of cell size $\xi$ (times the length $E_\text{max}-E_0$ of the interval of energies of $H$). This allows one to find $E_0$ up to additive error $\xi$ with $\log\big(\frac{1}{\xi}\big)$ steps of a binary-like search over $y$ \cite{lin_heisenberg-limited_2022}. 
At each step, we apply our Alg. \fig{main} with $f(x)=\theta(y-x)$, $A=H$, and $\ket{\phi}=\ket{\psi}$ to estimate $F_{\varrho}(y)$, with $\varrho=\ketbra{\psi}{\psi}$. Here, $\ket{\psi}$ is any state with promised overlap $\eta>0$ with the ground state subspace. 
The requirement of additive precision $\frac{\eta}{2}$ for $F_\varrho(x)$ requires an approximation error $\nu\le\frac{\eta}{4}$ for $f$ and a statistical error $\epsilon\le\frac{\eta}{4}$ for the estimation. 

Interestingly, our approach does not need to estimate $F_\varrho(y)$ at different $y$'s for the search. In Lemma \ref{lemma_ychoice} in App. \sect{explicit_examples_step}, we show that estimating $F_\varrho$ at a special point $y_*=1/\sqrt{2}$ and increasing the number of samples suffices to obtain $F_\varrho(y)$ at any other $y$. 
As a core auxiliary ingredient for that, 
we develop a $\nu$-approximation $\Tilde{f}$ to the step function with a shifted argument, $\theta(y-x)$, given in Lemma \ref{lemma_theta} in App. \sect{explicit_examples_step}. It has the appealing property that the $x$ and $y$ dependence are separated, namely $\Tilde{f}(y-x)=\sum_{j\in[k]}\big[a_j(y)\,\mathcal{T}_j(x)\big]+\sum_{j\in[k]}\big[b_j(y)\sqrt{1-x^2}\,\mathcal{U}_j(x)\big]$, where $\mathcal{U}_j$ is the $j$-th Chebyshev polynomial of the second kind. 
The first contribution to $\Tilde{f}$ takes the usual form \eq{ftilde} and can be directly implemented by our Alg. \fig{main}; the second contribution containing the $\mathcal{U}_j$'s can also be implemented in a similar way, with the caveat that the required Hadamard test needs a minor modification described in Lemma \ref{lemma_hadamardsecondkind}, App. \sect{Hadamard_test}. 
The maximal degree $k=\mathcal{O}(\frac{1}{\xi}\log\big(\frac{1}{\eta}\big))$ is the same for both contributions and the coefficient 1-norms are $\norm{\vv{a}}_1=\norm{\vv{b}}_1=\mathcal{O}\left(\log\big(\frac{1}{\xi}\log\big(\frac{1}{\eta}\big)\big)\right)$. 
Putting all together and taking into account also the $\mathcal{O}\big(\log\big(\frac{1}{\xi}\big)\big)$ steps of the binary search, one obtains a total sample complexity $S^{(1)}=\mathcal{O}\left(\frac{1}{\eta^{2}}\log^2\big(\frac{1}{\xi}\log\big(\frac{1}{\eta}\big)\big)\log(\frac{4}{\delta}\log(\frac{1}{\xi}))\right)$. 

The corresponding expected query depth and total runtime are shown in Table \tab{main_table}. 
Remarkably, the query depth exhibits a speed-up with respect to the maximal value $k$, namely a square root improvement in the $\eta$ dependence and a logarithmic improvement in the $\frac{1}{\xi}$ dependence (see Lemma \ref{lemma_step_query} in App. \sect{explicit_examples_step} for details). 
In addition, as can be seen in the table, our expected runtime displays the same Heisenberg-scaling of \cite{wan_randomized_2022}. This is interesting given that our algorithm is based on block-encoded oracles rather than the time-evolution oracles used in \cite{wan_randomized_2022}, which may be better suited for digital platforms as discussed previously. Finally, it is interesting to note that there have been recent improvements in the precision dependence, e.g. based on a derivative Gaussian filter \cite{wang2023quantum}. Those matrix functions are also within the scope of applicability of our approach.

\section{Final discussion} 

We presented a randomized hybrid quantum-classical framework to efficiently estimate state amplitudes and expectation values involving a generic matrix function $f(A)$. 
More precisely, our algorithms perform a Monte-Carlo simulation of the powerful quantum signal processing (QSP) and singular-value transformation (QSVT) techniques \cite{LowChuang2017,LowChuangQuantum2019,Gilyen2019,ChuangGrandUnification}. 
Our toolbox is based on three main ingredients: 
$i$) it samples each component of a Chebyshev series for $f$ weighed by its coefficient in the series;
$ii$) it assumes coherent access to $A$ via a block-encoding oracle; and
$iii$) $f(A)$ is automatically extracted from its block-encoding without post-selection, 
using a Hadamard test. 
This combination allows us to deliver provably better circuit complexities than the standard QSP and QSVT algorithms while maintaining comparable total runtimes. 

We illustrated our algorithms on four specific end-user applications: partition-function estimation via quantum Markov-chain Monte Carlo and via imaginary-time evolution; linear system solvers; and ground-state energy estimation (GSEE). 
The full end-to-end complexity scalings are detailed in Table \tab{main_table}.

For GSEE, the reduction in query complexity (and consequently also in the total gate count) due to randomization is by a factor $k/\mathbb{E}[j]=\mathcal{O}\big(\sqrt{\log(1/\eta)}\,\log(\xi^{-1}\log(1/\eta))\big)$. 
We estimate this factor explicitly for the iron-molybdenum cofactor (FeMoco) molecule, which is the primary cofactor of nitrogenase and one of the main target use cases in chemistry for early quantum computers \cite{Reiher_2017,RevModPhys.92.015003}.  
For GSEE within chemical accuracy (see Sec. \ref{app:gate_count} in the Methods for details), the resulting reduction factor is approximately $28$ (from $k\approx1.12\times10^7$ to $\mathbb{E}[j]\approx4.01\times10^5$), while the sample complexity overhead is a factor of $\norm{\vv{a}}_1^2\approx2.35$. 
Importantly, these estimates do not take into account the overhead for quantum error correction (QEC). 
Our reductions in query depth imply larger tolerated logical-gate error levels, which translate into lower code distances and, hence, lower QEC overheads. 
E.g., for the surface code, the overhead in physical-qubit number is usually quadratic with the code distance \cite{Kim_2022}.

An interesting future direction is to explore other matrix functions with our framework. This includes recent developments such as Gaussian and derivative-Gaussian filters for precision improvements in ground-state energy estimation \cite{wang2023quantum} or Green function estimation \cite{wang2023qubitefficient}, and a diversity of other more-established use cases \cite{ChuangGrandUnification}. 
Another possibility is to explore the applicability of our methods in the context of hybrid quantum-classical rejection sampling \cite{wang2023faster}.
Moreover, further studies on the interplay between our framework and Fourier-based matrix processing \cite{silva2022fourierbased,Dong_2022} may be in place too. Fourier-based approaches have so far focused mainly on the eigenvalue thresholding for ground-state energy estimation \cite{lin_heisenberg-limited_2022,wang2023quantum,wan_randomized_2022,wang2023qubitefficient}. 

Our findings open a promising arena to build and optimize early fault-tolerant quantum algorithms towards practical linear-algebra applications in a nearer term.

\begin{acknowledgements}
    AT {and LA} acknowledge financial support from the Serrapilheira Institute (grant number Serra-1709-17173).
    We thank Lucas Borges, Samson Wang, Sam McArdle, Mario Berta, Daniel Stilck-Fran\c{c}a,  and Juan Miguel Arrazola for helpful discussions.
\end{acknowledgements}

\bibliography{paper/sn-bibliography}

\begin{thebibliography}{62}%
\makeatletter
\providecommand \@ifxundefined [1]{%
 \@ifx{#1\undefined}
}%
\providecommand \@ifnum [1]{%
 \ifnum #1\expandafter \@firstoftwo
 \else \expandafter \@secondoftwo
 \fi
}%
\providecommand \@ifx [1]{%
 \ifx #1\expandafter \@firstoftwo
 \else \expandafter \@secondoftwo
 \fi
}%
\providecommand \natexlab [1]{#1}%
\providecommand \enquote  [1]{``#1''}%
\providecommand \bibnamefont  [1]{#1}%
\providecommand \bibfnamefont [1]{#1}%
\providecommand \citenamefont [1]{#1}%
\providecommand \href@noop [0]{\@secondoftwo}%
\providecommand \href [0]{\begingroup \@sanitize@url \@href}%
\providecommand \@href[1]{\@@startlink{#1}\@@href}%
\providecommand \@@href[1]{\endgroup#1\@@endlink}%
\providecommand \@sanitize@url [0]{\catcode `\\12\catcode `\$12\catcode `\&12\catcode `\#12\catcode `\^12\catcode `\_12\catcode `\%12\relax}%
\providecommand \@@startlink[1]{}%
\providecommand \@@endlink[0]{}%
\providecommand \url  [0]{\begingroup\@sanitize@url \@url }%
\providecommand \@url [1]{\endgroup\@href {#1}{\urlprefix }}%
\providecommand \urlprefix  [0]{URL }%
\providecommand \Eprint [0]{\href }%
\providecommand \doibase [0]{https://doi.org/}%
\providecommand \selectlanguage [0]{\@gobble}%
\providecommand \bibinfo  [0]{\@secondoftwo}%
\providecommand \bibfield  [0]{\@secondoftwo}%
\providecommand \translation [1]{[#1]}%
\providecommand \BibitemOpen [0]{}%
\providecommand \bibitemStop [0]{}%
\providecommand \bibitemNoStop [0]{.\EOS\space}%
\providecommand \EOS [0]{\spacefactor3000\relax}%
\providecommand \BibitemShut  [1]{\csname bibitem#1\endcsname}%
\let\auto@bib@innerbib\@empty
\bibitem [{\citenamefont {Low}\ and\ \citenamefont {Chuang}(2017)}]{LowChuang2017}%
  \BibitemOpen
  \bibfield  {author} {\bibinfo {author} {\bibfnamefont {G.~H.}\ \bibnamefont {Low}}\ and\ \bibinfo {author} {\bibfnamefont {I.~L.}\ \bibnamefont {Chuang}},\ }\bibfield  {title} {\bibinfo {title} {Optimal hamiltonian simulation by quantum signal processing},\ }\href {https://doi.org/10.1103/PhysRevLett.118.010501} {\bibfield  {journal} {\bibinfo  {journal} {Phys. Rev. Lett.}\ }\textbf {\bibinfo {volume} {118}},\ \bibinfo {pages} {010501} (\bibinfo {year} {2017})}\BibitemShut {NoStop}%
\bibitem [{\citenamefont {Low}\ and\ \citenamefont {Chuang}(2019)}]{LowChuangQuantum2019}%
  \BibitemOpen
  \bibfield  {author} {\bibinfo {author} {\bibfnamefont {G.~H.}\ \bibnamefont {Low}}\ and\ \bibinfo {author} {\bibfnamefont {I.~L.}\ \bibnamefont {Chuang}},\ }\bibfield  {title} {\bibinfo {title} {Hamiltonian {S}imulation by {Q}ubitization},\ }\href {https://doi.org/10.22331/q-2019-07-12-163} {\bibfield  {journal} {\bibinfo  {journal} {{Quantum}}\ }\textbf {\bibinfo {volume} {3}},\ \bibinfo {pages} {163} (\bibinfo {year} {2019})}\BibitemShut {NoStop}%
\bibitem [{\citenamefont {Gily\'{e}n}\ \emph {et~al.}(2019)\citenamefont {Gily\'{e}n}, \citenamefont {Su}, \citenamefont {Low},\ and\ \citenamefont {Wiebe}}]{Gilyen2019}%
  \BibitemOpen
  \bibfield  {author} {\bibinfo {author} {\bibfnamefont {A.}~\bibnamefont {Gily\'{e}n}}, \bibinfo {author} {\bibfnamefont {Y.}~\bibnamefont {Su}}, \bibinfo {author} {\bibfnamefont {G.~H.}\ \bibnamefont {Low}},\ and\ \bibinfo {author} {\bibfnamefont {N.}~\bibnamefont {Wiebe}},\ }\bibfield  {title} {\bibinfo {title} {Quantum singular value transformation and beyond: Exponential improvements for quantum matrix arithmetics},\ }\href {https://doi.org/10.1145/3313276.3316366} {\bibfield  {journal} {\bibinfo  {journal} {Proceedings of the 51st Annual ACM SIGACT Symposium on Theory of Computing}\ }\bibinfo {series} {STOC 2019},\ \bibinfo {pages} {193–204} (\bibinfo {year} {2019})}\BibitemShut {NoStop}%
\bibitem [{\citenamefont {Martyn}\ \emph {et~al.}(2021)\citenamefont {Martyn}, \citenamefont {Rossi}, \citenamefont {Tan},\ and\ \citenamefont {Chuang}}]{ChuangGrandUnification}%
  \BibitemOpen
  \bibfield  {author} {\bibinfo {author} {\bibfnamefont {J.~M.}\ \bibnamefont {Martyn}}, \bibinfo {author} {\bibfnamefont {Z.~M.}\ \bibnamefont {Rossi}}, \bibinfo {author} {\bibfnamefont {A.~K.}\ \bibnamefont {Tan}},\ and\ \bibinfo {author} {\bibfnamefont {I.~L.}\ \bibnamefont {Chuang}},\ }\bibfield  {title} {\bibinfo {title} {Grand unification of quantum algorithms},\ }\href {https://doi.org/10.1103/PRXQuantum.2.040203} {\bibfield  {journal} {\bibinfo  {journal} {PRX Quantum}\ }\textbf {\bibinfo {volume} {2}},\ \bibinfo {pages} {040203} (\bibinfo {year} {2021})}\BibitemShut {NoStop}%
\bibitem [{\citenamefont {Silva}\ \emph {et~al.}(2023)\citenamefont {Silva}, \citenamefont {Taddei}, \citenamefont {Carrazza},\ and\ \citenamefont {Aolita}}]{silva_fragmented_2022}%
  \BibitemOpen
  \bibfield  {author} {\bibinfo {author} {\bibfnamefont {T.~d.~L.}\ \bibnamefont {Silva}}, \bibinfo {author} {\bibfnamefont {M.~M.}\ \bibnamefont {Taddei}}, \bibinfo {author} {\bibfnamefont {S.}~\bibnamefont {Carrazza}},\ and\ \bibinfo {author} {\bibfnamefont {L.}~\bibnamefont {Aolita}},\ }\bibfield  {title} {\bibinfo {title} {Fragmented imaginary-time evolution for early-stage quantum signal processors},\ }\href {https://doi.org/10.1038/s41598-023-45540-2} {\bibfield  {journal} {\bibinfo  {journal} {Scientific Reports}\ }\textbf {\bibinfo {volume} {13}},\ \bibinfo {pages} {18258} (\bibinfo {year} {2023})}\BibitemShut {NoStop}%
\bibitem [{\citenamefont {de~Lima~Silva}\ \emph {et~al.}(2022)\citenamefont {de~Lima~Silva}, \citenamefont {Borges},\ and\ \citenamefont {Aolita}}]{silva2022fourierbased}%
  \BibitemOpen
  \bibfield  {author} {\bibinfo {author} {\bibfnamefont {T.}~\bibnamefont {de~Lima~Silva}}, \bibinfo {author} {\bibfnamefont {L.}~\bibnamefont {Borges}},\ and\ \bibinfo {author} {\bibfnamefont {L.}~\bibnamefont {Aolita}},\ }\href@noop {} {\bibinfo {title} {Fourier-based quantum signal processing}} (\bibinfo {year} {2022}),\ \Eprint {https://arxiv.org/abs/2206.02826} {arXiv:2206.02826 [quant-ph]} \BibitemShut {NoStop}%
\bibitem [{\citenamefont {Lin}\ and\ \citenamefont {Tong}(2022)}]{lin_heisenberg-limited_2022}%
  \BibitemOpen
  \bibfield  {author} {\bibinfo {author} {\bibfnamefont {L.}~\bibnamefont {Lin}}\ and\ \bibinfo {author} {\bibfnamefont {Y.}~\bibnamefont {Tong}},\ }\bibfield  {title} {\bibinfo {title} {Heisenberg-{Limited} {Ground}-{State} {Energy} {Estimation} for {Early} {Fault}-{Tolerant} {Quantum} {Computers}},\ }\href {https://doi.org/10.1103/PRXQuantum.3.010318} {\bibfield  {journal} {\bibinfo  {journal} {PRX Quantum}\ }\textbf {\bibinfo {volume} {3}},\ \bibinfo {pages} {010318} (\bibinfo {year} {2022})},\ \bibinfo {note} {publisher: American Physical Society}\BibitemShut {NoStop}%
\bibitem [{\citenamefont {Wang}\ \emph {et~al.}(2023{\natexlab{a}})\citenamefont {Wang}, \citenamefont {Fran{\c{c}}a}, \citenamefont {Zhang}, \citenamefont {Zhu},\ and\ \citenamefont {Johnson}}]{wang2023quantum}%
  \BibitemOpen
  \bibfield  {author} {\bibinfo {author} {\bibfnamefont {G.}~\bibnamefont {Wang}}, \bibinfo {author} {\bibfnamefont {D.~S.}\ \bibnamefont {Fran{\c{c}}a}}, \bibinfo {author} {\bibfnamefont {R.}~\bibnamefont {Zhang}}, \bibinfo {author} {\bibfnamefont {S.}~\bibnamefont {Zhu}},\ and\ \bibinfo {author} {\bibfnamefont {P.~D.}\ \bibnamefont {Johnson}},\ }\bibfield  {title} {\bibinfo {title} {Quantum algorithm for ground state energy estimation using circuit depth with exponentially improved dependence on precision},\ }\href {https://doi.org/10.22331/q-2023-11-06-1167} {\bibfield  {journal} {\bibinfo  {journal} {{Quantum}}\ }\textbf {\bibinfo {volume} {7}},\ \bibinfo {pages} {1167} (\bibinfo {year} {2023}{\natexlab{a}})}\BibitemShut {NoStop}%
\bibitem [{\citenamefont {Campbell}(2019)}]{campbell_random_2019}%
  \BibitemOpen
  \bibfield  {author} {\bibinfo {author} {\bibfnamefont {E.}~\bibnamefont {Campbell}},\ }\bibfield  {title} {\bibinfo {title} {Random {Compiler} for {Fast} {Hamiltonian} {Simulation}},\ }\href {https://doi.org/10.1103/PhysRevLett.123.070503} {\bibfield  {journal} {\bibinfo  {journal} {Physical Review Letters}\ }\textbf {\bibinfo {volume} {123}},\ \bibinfo {pages} {070503} (\bibinfo {year} {2019})},\ \bibinfo {note} {publisher: American Physical Society}\BibitemShut {NoStop}%
\bibitem [{\citenamefont {Wan}\ \emph {et~al.}(2022)\citenamefont {Wan}, \citenamefont {Berta},\ and\ \citenamefont {Campbell}}]{wan_randomized_2022}%
  \BibitemOpen
  \bibfield  {author} {\bibinfo {author} {\bibfnamefont {K.}~\bibnamefont {Wan}}, \bibinfo {author} {\bibfnamefont {M.}~\bibnamefont {Berta}},\ and\ \bibinfo {author} {\bibfnamefont {E.~T.}\ \bibnamefont {Campbell}},\ }\bibfield  {title} {\bibinfo {title} {Randomized {Quantum} {Algorithm} for {Statistical} {Phase} {Estimation}},\ }\href {https://doi.org/10.1103/PhysRevLett.129.030503} {\bibfield  {journal} {\bibinfo  {journal} {Physical Review Letters}\ }\textbf {\bibinfo {volume} {129}},\ \bibinfo {pages} {030503} (\bibinfo {year} {2022})},\ \bibinfo {note} {publisher: American Physical Society}\BibitemShut {NoStop}%
\bibitem [{\citenamefont {Wang}\ \emph {et~al.}(2024)\citenamefont {Wang}, \citenamefont {McArdle},\ and\ \citenamefont {Berta}}]{wang2023qubitefficient}%
  \BibitemOpen
  \bibfield  {author} {\bibinfo {author} {\bibfnamefont {S.}~\bibnamefont {Wang}}, \bibinfo {author} {\bibfnamefont {S.}~\bibnamefont {McArdle}},\ and\ \bibinfo {author} {\bibfnamefont {M.}~\bibnamefont {Berta}},\ }\bibfield  {title} {\bibinfo {title} {Qubit-efficient randomized quantum algorithms for linear algebra},\ }\href {https://doi.org/10.1103/PRXQuantum.5.020324} {\bibfield  {journal} {\bibinfo  {journal} {PRX Quantum}\ }\textbf {\bibinfo {volume} {5}},\ \bibinfo {pages} {020324} (\bibinfo {year} {2024})}\BibitemShut {NoStop}%
\bibitem [{\citenamefont {Campbell}(2021)}]{Campbell_2021}%
  \BibitemOpen
  \bibfield  {author} {\bibinfo {author} {\bibfnamefont {E.~T.}\ \bibnamefont {Campbell}},\ }\bibfield  {title} {\bibinfo {title} {Early fault-tolerant simulations of the hubbard model},\ }\href {https://doi.org/10.1088/2058-9565/ac3110} {\bibfield  {journal} {\bibinfo  {journal} {Quantum Science and Technology}\ }\textbf {\bibinfo {volume} {7}},\ \bibinfo {pages} {015007} (\bibinfo {year} {2021})}\BibitemShut {NoStop}%
\bibitem [{\citenamefont {Dong}\ \emph {et~al.}(2022)\citenamefont {Dong}, \citenamefont {Lin},\ and\ \citenamefont {Tong}}]{Dong_2022}%
  \BibitemOpen
  \bibfield  {author} {\bibinfo {author} {\bibfnamefont {Y.}~\bibnamefont {Dong}}, \bibinfo {author} {\bibfnamefont {L.}~\bibnamefont {Lin}},\ and\ \bibinfo {author} {\bibfnamefont {Y.}~\bibnamefont {Tong}},\ }\bibfield  {title} {\bibinfo {title} {Ground-state preparation and energy estimation on early fault-tolerant quantum computers via quantum eigenvalue transformation of unitary matrices},\ }\bibfield  {journal} {\bibinfo  {journal} {{PRX} Quantum}\ }\textbf {\bibinfo {volume} {3}},\ \href {https://doi.org/10.1103/prxquantum.3.040305} {10.1103/prxquantum.3.040305} (\bibinfo {year} {2022})\BibitemShut {NoStop}%
\bibitem [{\citenamefont {Wang}\ \emph {et~al.}(2023{\natexlab{b}})\citenamefont {Wang}, \citenamefont {França}, \citenamefont {Rendon},\ and\ \citenamefont {Johnson}}]{wang2023faster}%
  \BibitemOpen
  \bibfield  {author} {\bibinfo {author} {\bibfnamefont {G.}~\bibnamefont {Wang}}, \bibinfo {author} {\bibfnamefont {D.~S.}\ \bibnamefont {França}}, \bibinfo {author} {\bibfnamefont {G.}~\bibnamefont {Rendon}},\ and\ \bibinfo {author} {\bibfnamefont {P.~D.}\ \bibnamefont {Johnson}},\ }\href@noop {} {\bibinfo {title} {Faster ground state energy estimation on early fault-tolerant quantum computers via rejection sampling}} (\bibinfo {year} {2023}{\natexlab{b}}),\ \Eprint {https://arxiv.org/abs/2304.09827} {arXiv:2304.09827 [quant-ph]} \BibitemShut {NoStop}%
\bibitem [{\citenamefont {Trefethen}(2012)}]{trefethen_approx}%
  \BibitemOpen
  \bibfield  {author} {\bibinfo {author} {\bibfnamefont {L.~N.}\ \bibnamefont {Trefethen}},\ }\href@noop {} {\emph {\bibinfo {title} {Approximation Theory and Approximation Practice.}}}\ (\bibinfo  {publisher} {SIAM},\ \bibinfo {year} {2012})\ pp.\ \bibinfo {pages} {I--VII, 1--305}\BibitemShut {NoStop}%
\bibitem [{\citenamefont {Boyd}(2001)}]{boyd_spectral_methods}%
  \BibitemOpen
  \bibfield  {author} {\bibinfo {author} {\bibfnamefont {J.~P.}\ \bibnamefont {Boyd}},\ }\href@noop {} {\emph {\bibinfo {title} {{Chebyshev} and {Fourier} Spectral Methods}}},\ \bibinfo {edition} {2nd}\ ed.\ (\bibinfo  {publisher} {Dover},\ \bibinfo {address} {Mineola, New York},\ \bibinfo {year} {2001})\BibitemShut {NoStop}%
\bibitem [{\citenamefont {Childs}\ \emph {et~al.}(2017{\natexlab{a}})\citenamefont {Childs}, \citenamefont {Kothari},\ and\ \citenamefont {Somma}}]{Childs_2017}%
  \BibitemOpen
  \bibfield  {author} {\bibinfo {author} {\bibfnamefont {A.~M.}\ \bibnamefont {Childs}}, \bibinfo {author} {\bibfnamefont {R.}~\bibnamefont {Kothari}},\ and\ \bibinfo {author} {\bibfnamefont {R.~D.}\ \bibnamefont {Somma}},\ }\bibfield  {title} {\bibinfo {title} {Quantum algorithm for systems of linear equations with exponentially improved dependence on precision},\ }\href {https://doi.org/10.1137/16m1087072} {\bibfield  {journal} {\bibinfo  {journal} {{SIAM} Journal on Computing}\ }\textbf {\bibinfo {volume} {46}},\ \bibinfo {pages} {1920} (\bibinfo {year} {2017}{\natexlab{a}})}\BibitemShut {NoStop}%
\bibitem [{\citenamefont {Costa}\ \emph {et~al.}(2022)\citenamefont {Costa}, \citenamefont {An}, \citenamefont {Sanders}, \citenamefont {Su}, \citenamefont {Babbush},\ and\ \citenamefont {Berry}}]{costa2021optimal}%
  \BibitemOpen
  \bibfield  {author} {\bibinfo {author} {\bibfnamefont {P.~C.}\ \bibnamefont {Costa}}, \bibinfo {author} {\bibfnamefont {D.}~\bibnamefont {An}}, \bibinfo {author} {\bibfnamefont {Y.~R.}\ \bibnamefont {Sanders}}, \bibinfo {author} {\bibfnamefont {Y.}~\bibnamefont {Su}}, \bibinfo {author} {\bibfnamefont {R.}~\bibnamefont {Babbush}},\ and\ \bibinfo {author} {\bibfnamefont {D.~W.}\ \bibnamefont {Berry}},\ }\bibfield  {title} {\bibinfo {title} {Optimal scaling quantum linear-systems solver via discrete adiabatic theorem},\ }\href {https://doi.org/10.1103/PRXQuantum.3.040303} {\bibfield  {journal} {\bibinfo  {journal} {PRX Quantum}\ }\textbf {\bibinfo {volume} {3}},\ \bibinfo {pages} {040303} (\bibinfo {year} {2022})}\BibitemShut {NoStop}%
\bibitem [{\citenamefont {Sachdeva}\ and\ \citenamefont {Vishnoi}(2013{\natexlab{a}})}]{Vishnoi2013}%
  \BibitemOpen
  \bibfield  {author} {\bibinfo {author} {\bibfnamefont {S.}~\bibnamefont {Sachdeva}}\ and\ \bibinfo {author} {\bibfnamefont {N.~K.}\ \bibnamefont {Vishnoi}},\ }\bibfield  {title} {\bibinfo {title} {Faster algorithms via approximation theory},\ }\href {https://doi.org/10.1561/0400000065} {\bibfield  {journal} {\bibinfo  {journal} {Found. Trends Theor. Comput. Sci.}\ }\textbf {\bibinfo {volume} {9}},\ \bibinfo {pages} {125} (\bibinfo {year} {2013}{\natexlab{a}})}\BibitemShut {NoStop}%
\bibitem [{\citenamefont {Camps}\ \emph {et~al.}(2024)\citenamefont {Camps}, \citenamefont {Lin}, \citenamefont {Van~Beeumen},\ and\ \citenamefont {Yang}}]{camps2023explicit}%
  \BibitemOpen
  \bibfield  {author} {\bibinfo {author} {\bibfnamefont {D.}~\bibnamefont {Camps}}, \bibinfo {author} {\bibfnamefont {L.}~\bibnamefont {Lin}}, \bibinfo {author} {\bibfnamefont {R.}~\bibnamefont {Van~Beeumen}},\ and\ \bibinfo {author} {\bibfnamefont {C.}~\bibnamefont {Yang}},\ }\bibfield  {title} {\bibinfo {title} {Explicit quantum circuits for block encodings of certain sparse matrices},\ }\href {https://doi.org/10.1137/22M1484298} {\bibfield  {journal} {\bibinfo  {journal} {SIAM Journal on Matrix Analysis and Applications}\ }\textbf {\bibinfo {volume} {45}},\ \bibinfo {pages} {801} (\bibinfo {year} {2024})},\ \Eprint {https://arxiv.org/abs/https://doi.org/10.1137/22M1484298} {https://doi.org/10.1137/22M1484298} \BibitemShut {NoStop}%
\bibitem [{\citenamefont {Sünderhauf}\ \emph {et~al.}(2024)\citenamefont {Sünderhauf}, \citenamefont {Campbell},\ and\ \citenamefont {Camps}}]{sunderhauf2023blockencoding}%
  \BibitemOpen
  \bibfield  {author} {\bibinfo {author} {\bibfnamefont {C.}~\bibnamefont {Sünderhauf}}, \bibinfo {author} {\bibfnamefont {E.}~\bibnamefont {Campbell}},\ and\ \bibinfo {author} {\bibfnamefont {J.}~\bibnamefont {Camps}},\ }\bibfield  {title} {\bibinfo {title} {Block-encoding structured matrices for data input in quantum computing},\ }\href {https://doi.org/10.22331/q-2024-01-11-1226} {\bibfield  {journal} {\bibinfo  {journal} {Quantum}\ }\textbf {\bibinfo {volume} {8}},\ \bibinfo {pages} {1226} (\bibinfo {year} {2024})}\BibitemShut {NoStop}%
\bibitem [{\citenamefont {Aolita}\ \emph {et~al.}(2015)\citenamefont {Aolita}, \citenamefont {de~Melo},\ and\ \citenamefont {Davidovich}}]{Aolita15review}%
  \BibitemOpen
  \bibfield  {author} {\bibinfo {author} {\bibfnamefont {L.}~\bibnamefont {Aolita}}, \bibinfo {author} {\bibfnamefont {F.}~\bibnamefont {de~Melo}},\ and\ \bibinfo {author} {\bibfnamefont {L.}~\bibnamefont {Davidovich}},\ }\bibfield  {title} {\bibinfo {title} {Open-system dynamics of entanglement:a key issues review},\ }\href@noop {} {\bibfield  {journal} {\bibinfo  {journal} {Rep. Prog. Phys.}\ }\textbf {\bibinfo {volume} {78}},\ \bibinfo {pages} {042001} (\bibinfo {year} {2015})}\BibitemShut {NoStop}%
\bibitem [{\citenamefont {Ma}\ \emph {et~al.}(2013)\citenamefont {Ma}, \citenamefont {Peng}, \citenamefont {Wang},\ and\ \citenamefont {Xu}}]{Ma_Peng_Wang}%
  \BibitemOpen
  \bibfield  {author} {\bibinfo {author} {\bibfnamefont {J.}~\bibnamefont {Ma}}, \bibinfo {author} {\bibfnamefont {J.}~\bibnamefont {Peng}}, \bibinfo {author} {\bibfnamefont {S.}~\bibnamefont {Wang}},\ and\ \bibinfo {author} {\bibfnamefont {J.}~\bibnamefont {Xu}},\ }\bibfield  {title} {\bibinfo {title} {Estimating the partition function of graphical models using langevin importance sampling},\ }in\ \href {https://proceedings.mlr.press/v31/ma13a.html} {\emph {\bibinfo {booktitle} {Proceedings of the Sixteenth International Conference on Artificial Intelligence and Statistics}}},\ \bibinfo {series} {Proceedings of Machine Learning Research}, Vol.~\bibinfo {volume} {31},\ \bibinfo {editor} {edited by\ \bibinfo {editor} {\bibfnamefont {C.~M.}\ \bibnamefont {Carvalho}}\ and\ \bibinfo {editor} {\bibfnamefont {P.}~\bibnamefont {Ravikumar}}}\ (\bibinfo  {publisher} {PMLR},\ \bibinfo {address} {Scottsdale, Arizona, USA},\ \bibinfo {year} {2013})\ pp.\ \bibinfo {pages} {433--441}\BibitemShut {NoStop}%
\bibitem [{\citenamefont {Krause}\ \emph {et~al.}(2020)\citenamefont {Krause}, \citenamefont {Fischer},\ and\ \citenamefont {Igel}}]{KRAUSE2020103195}%
  \BibitemOpen
  \bibfield  {author} {\bibinfo {author} {\bibfnamefont {O.}~\bibnamefont {Krause}}, \bibinfo {author} {\bibfnamefont {A.}~\bibnamefont {Fischer}},\ and\ \bibinfo {author} {\bibfnamefont {C.}~\bibnamefont {Igel}},\ }\bibfield  {title} {\bibinfo {title} {Algorithms for estimating the partition function of restricted boltzmann machines},\ }\href {https://doi.org/https://doi.org/10.1016/j.artint.2019.103195} {\bibfield  {journal} {\bibinfo  {journal} {Artificial Intelligence}\ }\textbf {\bibinfo {volume} {278}},\ \bibinfo {pages} {103195} (\bibinfo {year} {2020})}\BibitemShut {NoStop}%
\bibitem [{\citenamefont {Shim}(2022)}]{shim2022probabilistic}%
  \BibitemOpen
  \bibfield  {author} {\bibinfo {author} {\bibfnamefont {A.}~\bibnamefont {Shim}},\ }\href@noop {} {\bibinfo {title} {A probabilistic interpretation of transformers}} (\bibinfo {year} {2022}),\ \Eprint {https://arxiv.org/abs/2205.01080} {arXiv:2205.01080 [cs.LG]} \BibitemShut {NoStop}%
\bibitem [{\citenamefont {Bulatov}\ and\ \citenamefont {Grohe}(2005)}]{BULATOV2005148}%
  \BibitemOpen
  \bibfield  {author} {\bibinfo {author} {\bibfnamefont {A.}~\bibnamefont {Bulatov}}\ and\ \bibinfo {author} {\bibfnamefont {M.}~\bibnamefont {Grohe}},\ }\bibfield  {title} {\bibinfo {title} {The complexity of partition functions},\ }\href {https://doi.org/https://doi.org/10.1016/j.tcs.2005.09.011} {\bibfield  {journal} {\bibinfo  {journal} {Theoretical Computer Science}\ }\textbf {\bibinfo {volume} {348}},\ \bibinfo {pages} {148} (\bibinfo {year} {2005})},\ \bibinfo {note} {automata, Languages and Programming: Algorithms and Complexity (ICALP-A 2004)}\BibitemShut {NoStop}%
\bibitem [{\citenamefont {Wei\ss{}e}\ \emph {et~al.}(2006)\citenamefont {Wei\ss{}e}, \citenamefont {Wellein}, \citenamefont {Alvermann},\ and\ \citenamefont {Fehske}}]{RevModPhys.78.275}%
  \BibitemOpen
  \bibfield  {author} {\bibinfo {author} {\bibfnamefont {A.}~\bibnamefont {Wei\ss{}e}}, \bibinfo {author} {\bibfnamefont {G.}~\bibnamefont {Wellein}}, \bibinfo {author} {\bibfnamefont {A.}~\bibnamefont {Alvermann}},\ and\ \bibinfo {author} {\bibfnamefont {H.}~\bibnamefont {Fehske}},\ }\bibfield  {title} {\bibinfo {title} {The kernel polynomial method},\ }\href {https://doi.org/10.1103/RevModPhys.78.275} {\bibfield  {journal} {\bibinfo  {journal} {Rev. Mod. Phys.}\ }\textbf {\bibinfo {volume} {78}},\ \bibinfo {pages} {275} (\bibinfo {year} {2006})}\BibitemShut {NoStop}%
\bibitem [{\citenamefont {Bravyi}\ \emph {et~al.}(2022)\citenamefont {Bravyi}, \citenamefont {Chowdhury}, \citenamefont {Gosset},\ and\ \citenamefont {Wocjan}}]{bravyi2021complexity}%
  \BibitemOpen
  \bibfield  {author} {\bibinfo {author} {\bibfnamefont {S.}~\bibnamefont {Bravyi}}, \bibinfo {author} {\bibfnamefont {A.}~\bibnamefont {Chowdhury}}, \bibinfo {author} {\bibfnamefont {D.}~\bibnamefont {Gosset}},\ and\ \bibinfo {author} {\bibfnamefont {P.}~\bibnamefont {Wocjan}},\ }\bibfield  {title} {\bibinfo {title} {Quantum hamiltonian complexity in thermal equilibrium},\ }\href {https://doi.org/10.1038/s41567-022-01742-5} {\bibfield  {journal} {\bibinfo  {journal} {Nature Physics}\ }\textbf {\bibinfo {volume} {18}},\ \bibinfo {pages} {1367–1370} (\bibinfo {year} {2022})}\BibitemShut {NoStop}%
\bibitem [{\citenamefont {Poulin}\ and\ \citenamefont {Wocjan}(2009{\natexlab{a}})}]{PhysRevLett.103.220502}%
  \BibitemOpen
  \bibfield  {author} {\bibinfo {author} {\bibfnamefont {D.}~\bibnamefont {Poulin}}\ and\ \bibinfo {author} {\bibfnamefont {P.}~\bibnamefont {Wocjan}},\ }\bibfield  {title} {\bibinfo {title} {Sampling from the thermal quantum gibbs state and evaluating partition functions with a quantum computer},\ }\href {https://doi.org/10.1103/PhysRevLett.103.220502} {\bibfield  {journal} {\bibinfo  {journal} {Phys. Rev. Lett.}\ }\textbf {\bibinfo {volume} {103}},\ \bibinfo {pages} {220502} (\bibinfo {year} {2009}{\natexlab{a}})}\BibitemShut {NoStop}%
\bibitem [{\citenamefont {Chowdhury}\ \emph {et~al.}(2021)\citenamefont {Chowdhury}, \citenamefont {Somma},\ and\ \citenamefont {Subasi}}]{chowdhury_computing_2021}%
  \BibitemOpen
  \bibfield  {author} {\bibinfo {author} {\bibfnamefont {A.~N.}\ \bibnamefont {Chowdhury}}, \bibinfo {author} {\bibfnamefont {R.~D.}\ \bibnamefont {Somma}},\ and\ \bibinfo {author} {\bibfnamefont {Y.}~\bibnamefont {Subasi}},\ }\bibfield  {title} {\bibinfo {title} {Computing partition functions in the one clean qubit model},\ }\href {https://doi.org/10.1103/PhysRevA.103.032422} {\bibfield  {journal} {\bibinfo  {journal} {Physical Review A}\ }\textbf {\bibinfo {volume} {103}},\ \bibinfo {pages} {032422} (\bibinfo {year} {2021})},\ \bibinfo {note} {arXiv:1910.11842 [quant-ph]}\BibitemShut {NoStop}%
\bibitem [{\citenamefont {Jackson}\ \emph {et~al.}(2023)\citenamefont {Jackson}, \citenamefont {Kapourniotis},\ and\ \citenamefont {Datta}}]{PhysRevA.107.012421}%
  \BibitemOpen
  \bibfield  {author} {\bibinfo {author} {\bibfnamefont {A.}~\bibnamefont {Jackson}}, \bibinfo {author} {\bibfnamefont {T.}~\bibnamefont {Kapourniotis}},\ and\ \bibinfo {author} {\bibfnamefont {A.}~\bibnamefont {Datta}},\ }\bibfield  {title} {\bibinfo {title} {Partition-function estimation: Quantum and quantum-inspired algorithms},\ }\href {https://doi.org/10.1103/PhysRevA.107.012421} {\bibfield  {journal} {\bibinfo  {journal} {Phys. Rev. A}\ }\textbf {\bibinfo {volume} {107}},\ \bibinfo {pages} {012421} (\bibinfo {year} {2023})}\BibitemShut {NoStop}%
\bibitem [{\citenamefont {Sun}\ \emph {et~al.}(2021)\citenamefont {Sun}, \citenamefont {Motta}, \citenamefont {Tazhigulov}, \citenamefont {Tan}, \citenamefont {Chan},\ and\ \citenamefont {Minnich}}]{Sunetal21}%
  \BibitemOpen
  \bibfield  {author} {\bibinfo {author} {\bibfnamefont {S.-N.}\ \bibnamefont {Sun}}, \bibinfo {author} {\bibfnamefont {M.}~\bibnamefont {Motta}}, \bibinfo {author} {\bibfnamefont {R.~N.}\ \bibnamefont {Tazhigulov}}, \bibinfo {author} {\bibfnamefont {A.~T.}\ \bibnamefont {Tan}}, \bibinfo {author} {\bibfnamefont {G.-L.}\ \bibnamefont {Chan}},\ and\ \bibinfo {author} {\bibfnamefont {A.~J.}\ \bibnamefont {Minnich}},\ }\bibfield  {title} {\bibinfo {title} {Quantum computation of finite-temperature static and dynamical properties of spin systems using quantum imaginary time evolution},\ }\href {https://doi.org/10.1103/PRXQuantum.2.010317} {\bibfield  {journal} {\bibinfo  {journal} {PRX Quantum}\ }\textbf {\bibinfo {volume} {2}},\ \bibinfo {pages} {010317} (\bibinfo {year} {2021})}\BibitemShut {NoStop}%
\bibitem [{\citenamefont {Szegedy}(2004)}]{szegedy2004}%
  \BibitemOpen
  \bibfield  {author} {\bibinfo {author} {\bibfnamefont {M.}~\bibnamefont {Szegedy}},\ }\bibfield  {title} {\bibinfo {title} {Quantum speed-up of markov chain based algorithms},\ }\href {https://doi.org/10.1109/FOCS.2004.53} {\bibfield  {journal} {\bibinfo  {journal} {45th Annual IEEE Symposium on Foundations of Computer Science}\ ,\ \bibinfo {pages} {32}} (\bibinfo {year} {2004})}\BibitemShut {NoStop}%
\bibitem [{\citenamefont {Lemieux}\ \emph {et~al.}(2020)\citenamefont {Lemieux}, \citenamefont {Heim}, \citenamefont {Poulin}, \citenamefont {Svore},\ and\ \citenamefont {Troyer}}]{Lemieux2020efficientquantum}%
  \BibitemOpen
  \bibfield  {author} {\bibinfo {author} {\bibfnamefont {J.}~\bibnamefont {Lemieux}}, \bibinfo {author} {\bibfnamefont {B.}~\bibnamefont {Heim}}, \bibinfo {author} {\bibfnamefont {D.}~\bibnamefont {Poulin}}, \bibinfo {author} {\bibfnamefont {K.}~\bibnamefont {Svore}},\ and\ \bibinfo {author} {\bibfnamefont {M.}~\bibnamefont {Troyer}},\ }\bibfield  {title} {\bibinfo {title} {Efficient {Q}uantum {W}alk {C}ircuits for {M}etropolis-{H}astings {A}lgorithm},\ }\href {https://doi.org/10.22331/q-2020-06-29-287} {\bibfield  {journal} {\bibinfo  {journal} {{Quantum}}\ }\textbf {\bibinfo {volume} {4}},\ \bibinfo {pages} {287} (\bibinfo {year} {2020})}\BibitemShut {NoStop}%
\bibitem [{\citenamefont {Trefethen}\ and\ \citenamefont {Bau}(1997)}]{trefethen97}%
  \BibitemOpen
  \bibfield  {author} {\bibinfo {author} {\bibfnamefont {L.~N.}\ \bibnamefont {Trefethen}}\ and\ \bibinfo {author} {\bibfnamefont {D.}~\bibnamefont {Bau}},\ }\href@noop {} {\emph {\bibinfo {title} {Numerical Linear Algebra}}}\ (\bibinfo  {publisher} {SIAM},\ \bibinfo {year} {1997})\BibitemShut {NoStop}%
\bibitem [{\citenamefont {Saad}(2003)}]{book_iterative_methods}%
  \BibitemOpen
  \bibfield  {author} {\bibinfo {author} {\bibfnamefont {Y.}~\bibnamefont {Saad}},\ }\href {https://doi.org/10.1137/1.9780898718003} {\emph {\bibinfo {title} {Iterative Methods for Sparse Linear Systems}}},\ \bibinfo {edition} {2nd}\ ed.\ (\bibinfo  {publisher} {Society for Industrial and Applied Mathematics},\ \bibinfo {year} {2003})\ \Eprint {https://arxiv.org/abs/https://epubs.siam.org/doi/pdf/10.1137/1.9780898718003} {https://epubs.siam.org/doi/pdf/10.1137/1.9780898718003} \BibitemShut {NoStop}%
\bibitem [{\citenamefont {Strohmer}\ and\ \citenamefont {Vershynin}(2007)}]{strohmer2007randomized}%
  \BibitemOpen
  \bibfield  {author} {\bibinfo {author} {\bibfnamefont {T.}~\bibnamefont {Strohmer}}\ and\ \bibinfo {author} {\bibfnamefont {R.}~\bibnamefont {Vershynin}},\ }\bibfield  {title} {\bibinfo {title} {A randomized kaczmarz algorithm with exponential convergence},\ }\href@noop {} {\bibfield  {journal} {\bibinfo  {journal} {Journal of Fourier Analysis and Applications}\ }\textbf {\bibinfo {volume} {15}},\ \bibinfo {pages} {262} (\bibinfo {year} {2007})},\ \Eprint {https://arxiv.org/abs/math/0702226} {arXiv:math/0702226 [math.NA]} \BibitemShut {NoStop}%
\bibitem [{\citenamefont {Harrow}\ \emph {et~al.}(2009)\citenamefont {Harrow}, \citenamefont {Hassidim},\ and\ \citenamefont {Lloyd}}]{harrow_quantum_2009}%
  \BibitemOpen
  \bibfield  {author} {\bibinfo {author} {\bibfnamefont {A.~W.}\ \bibnamefont {Harrow}}, \bibinfo {author} {\bibfnamefont {A.}~\bibnamefont {Hassidim}},\ and\ \bibinfo {author} {\bibfnamefont {S.}~\bibnamefont {Lloyd}},\ }\bibfield  {title} {\bibinfo {title} {Quantum {Algorithm} for {Linear} {Systems} of {Equations}},\ }\href {https://doi.org/10.1103/PhysRevLett.103.150502} {\bibfield  {journal} {\bibinfo  {journal} {Physical Review Letters}\ }\textbf {\bibinfo {volume} {103}},\ \bibinfo {pages} {150502} (\bibinfo {year} {2009})},\ \bibinfo {note} {publisher: American Physical Society}\BibitemShut {NoStop}%
\bibitem [{\citenamefont {Lin}\ and\ \citenamefont {Tong}(2020)}]{Lin_2020}%
  \BibitemOpen
  \bibfield  {author} {\bibinfo {author} {\bibfnamefont {L.}~\bibnamefont {Lin}}\ and\ \bibinfo {author} {\bibfnamefont {Y.}~\bibnamefont {Tong}},\ }\bibfield  {title} {\bibinfo {title} {Near-optimal ground state preparation},\ }\href {https://doi.org/10.22331/q-2020-12-14-372} {\bibfield  {journal} {\bibinfo  {journal} {Quantum}\ }\textbf {\bibinfo {volume} {4}},\ \bibinfo {pages} {372} (\bibinfo {year} {2020})}\BibitemShut {NoStop}%
\bibitem [{\citenamefont {An}\ and\ \citenamefont {Lin}(2022)}]{An_2022}%
  \BibitemOpen
  \bibfield  {author} {\bibinfo {author} {\bibfnamefont {D.}~\bibnamefont {An}}\ and\ \bibinfo {author} {\bibfnamefont {L.}~\bibnamefont {Lin}},\ }\bibfield  {title} {\bibinfo {title} {Quantum linear system solver based on time-optimal adiabatic quantum computing and quantum approximate optimization algorithm},\ }\href {https://doi.org/10.1145/3498331} {\bibfield  {journal} {\bibinfo  {journal} {{ACM} Transactions on Quantum Computing}\ }\textbf {\bibinfo {volume} {3}},\ \bibinfo {pages} {1} (\bibinfo {year} {2022})}\BibitemShut {NoStop}%
\bibitem [{\citenamefont {Suba{\c{s}}{\i}}\ \emph {et~al.}(2019)\citenamefont {Suba{\c{s}}{\i}}, \citenamefont {Somma},\ and\ \citenamefont {Orsucci}}]{Suba_2019}%
  \BibitemOpen
  \bibfield  {author} {\bibinfo {author} {\bibfnamefont {Y.}~\bibnamefont {Suba{\c{s}}{\i}}}, \bibinfo {author} {\bibfnamefont {R.~D.}\ \bibnamefont {Somma}},\ and\ \bibinfo {author} {\bibfnamefont {D.}~\bibnamefont {Orsucci}},\ }\bibfield  {title} {\bibinfo {title} {Quantum algorithms for systems of linear equations inspired by adiabatic quantum computing},\ }\bibfield  {journal} {\bibinfo  {journal} {Physical Review Letters}\ }\textbf {\bibinfo {volume} {122}},\ \href {https://doi.org/10.1103/physrevlett.122.060504} {10.1103/physrevlett.122.060504} (\bibinfo {year} {2019})\BibitemShut {NoStop}%
\bibitem [{\citenamefont {Wossnig}\ \emph {et~al.}(2018)\citenamefont {Wossnig}, \citenamefont {Zhao},\ and\ \citenamefont {Prakash}}]{Wossnig_2018}%
  \BibitemOpen
  \bibfield  {author} {\bibinfo {author} {\bibfnamefont {L.}~\bibnamefont {Wossnig}}, \bibinfo {author} {\bibfnamefont {Z.}~\bibnamefont {Zhao}},\ and\ \bibinfo {author} {\bibfnamefont {A.}~\bibnamefont {Prakash}},\ }\bibfield  {title} {\bibinfo {title} {Quantum linear system algorithm for dense matrices},\ }\bibfield  {journal} {\bibinfo  {journal} {Physical Review Letters}\ }\textbf {\bibinfo {volume} {120}},\ \href {https://doi.org/10.1103/physrevlett.120.050502} {10.1103/physrevlett.120.050502} (\bibinfo {year} {2018})\BibitemShut {NoStop}%
\bibitem [{\citenamefont {Clader}\ \emph {et~al.}(2022)\citenamefont {Clader}, \citenamefont {Dalzell}, \citenamefont {Stamatopoulos}, \citenamefont {Salton}, \citenamefont {Berta},\ and\ \citenamefont {Zeng}}]{Clader_2022}%
  \BibitemOpen
  \bibfield  {author} {\bibinfo {author} {\bibfnamefont {B.~D.}\ \bibnamefont {Clader}}, \bibinfo {author} {\bibfnamefont {A.~M.}\ \bibnamefont {Dalzell}}, \bibinfo {author} {\bibfnamefont {N.}~\bibnamefont {Stamatopoulos}}, \bibinfo {author} {\bibfnamefont {G.}~\bibnamefont {Salton}}, \bibinfo {author} {\bibfnamefont {M.}~\bibnamefont {Berta}},\ and\ \bibinfo {author} {\bibfnamefont {W.~J.}\ \bibnamefont {Zeng}},\ }\bibfield  {title} {\bibinfo {title} {Quantum resources required to block-encode a matrix of classical data},\ }\href {https://doi.org/10.1109/tqe.2022.3231194} {\bibfield  {journal} {\bibinfo  {journal} {{IEEE} Transactions on Quantum Engineering}\ }\textbf {\bibinfo {volume} {3}},\ \bibinfo {pages} {1} (\bibinfo {year} {2022})}\BibitemShut {NoStop}%
\bibitem [{\citenamefont {Clinton}\ \emph {et~al.}(2024)\citenamefont {Clinton}, \citenamefont {Cubitt}, \citenamefont {Flynn}, \citenamefont {Gambetta}, \citenamefont {Klassen}, \citenamefont {Montanaro}, \citenamefont {Piddock}, \citenamefont {Santos},\ and\ \citenamefont {Sheridan}}]{clinton2022nearterm}%
  \BibitemOpen
  \bibfield  {author} {\bibinfo {author} {\bibfnamefont {L.}~\bibnamefont {Clinton}}, \bibinfo {author} {\bibfnamefont {T.}~\bibnamefont {Cubitt}}, \bibinfo {author} {\bibfnamefont {B.}~\bibnamefont {Flynn}}, \bibinfo {author} {\bibfnamefont {F.~M.}\ \bibnamefont {Gambetta}}, \bibinfo {author} {\bibfnamefont {J.}~\bibnamefont {Klassen}}, \bibinfo {author} {\bibfnamefont {A.}~\bibnamefont {Montanaro}}, \bibinfo {author} {\bibfnamefont {S.}~\bibnamefont {Piddock}}, \bibinfo {author} {\bibfnamefont {R.~A.}\ \bibnamefont {Santos}},\ and\ \bibinfo {author} {\bibfnamefont {E.}~\bibnamefont {Sheridan}},\ }\bibfield  {title} {\bibinfo {title} {Towards near-term quantum simulation of materials},\ }\href {https://doi.org/10.1038/s41467-023-43479-6} {\bibfield  {journal} {\bibinfo  {journal} {Nat. Commun.}\ }\textbf {\bibinfo {volume} {15}},\ \bibinfo {pages} {211} (\bibinfo {year} {2024})}\BibitemShut {NoStop}%
\bibitem [{\citenamefont {Kempe}\ \emph {et~al.}(2006)\citenamefont {Kempe}, \citenamefont {Kitaev},\ and\ \citenamefont {Regev}}]{kempe2005complexity}%
  \BibitemOpen
  \bibfield  {author} {\bibinfo {author} {\bibfnamefont {J.}~\bibnamefont {Kempe}}, \bibinfo {author} {\bibfnamefont {A.}~\bibnamefont {Kitaev}},\ and\ \bibinfo {author} {\bibfnamefont {O.}~\bibnamefont {Regev}},\ }\bibfield  {title} {\bibinfo {title} {The complexity of the local hamiltonian problem},\ }\href {https://doi.org/10.1137/S0097539704445226} {\bibfield  {journal} {\bibinfo  {journal} {SIAM Journal on Computing}\ }\textbf {\bibinfo {volume} {35}},\ \bibinfo {pages} {1070} (\bibinfo {year} {2006})},\ \Eprint {https://arxiv.org/abs/https://doi.org/10.1137/S0097539704445226} {https://doi.org/10.1137/S0097539704445226} \BibitemShut {NoStop}%
\bibitem [{\citenamefont {Tubman}\ \emph {et~al.}(2018)\citenamefont {Tubman}, \citenamefont {Mejuto-Zaera}, \citenamefont {Epstein}, \citenamefont {Hait}, \citenamefont {Levine}, \citenamefont {Huggins}, \citenamefont {Jiang}, \citenamefont {McClean}, \citenamefont {Babbush}, \citenamefont {Head‐Gordon},\ and\ \citenamefont {Whaley}}]{tubman2018postponing}%
  \BibitemOpen
  \bibfield  {author} {\bibinfo {author} {\bibfnamefont {N.~M.}\ \bibnamefont {Tubman}}, \bibinfo {author} {\bibfnamefont {C.}~\bibnamefont {Mejuto-Zaera}}, \bibinfo {author} {\bibfnamefont {J.~M.}\ \bibnamefont {Epstein}}, \bibinfo {author} {\bibfnamefont {D.}~\bibnamefont {Hait}}, \bibinfo {author} {\bibfnamefont {D.~S.}\ \bibnamefont {Levine}}, \bibinfo {author} {\bibfnamefont {W.~J.}\ \bibnamefont {Huggins}}, \bibinfo {author} {\bibfnamefont {Z.}~\bibnamefont {Jiang}}, \bibinfo {author} {\bibfnamefont {J.~R.}\ \bibnamefont {McClean}}, \bibinfo {author} {\bibfnamefont {R.}~\bibnamefont {Babbush}}, \bibinfo {author} {\bibfnamefont {M.}~\bibnamefont {Head‐Gordon}},\ and\ \bibinfo {author} {\bibfnamefont {K.~B.}\ \bibnamefont {Whaley}},\ }\bibfield  {title} {\bibinfo {title} {Postponing the orthogonality catastrophe: efficient state preparation for electronic structure simulations on quantum devices},\ }\href {https://api.semanticscholar.org/CorpusID:119473829} {\bibfield  {journal} {\bibinfo  {journal}
  {Bulletin of the American Physical Society}\ } (\bibinfo {year} {2018})}\BibitemShut {NoStop}%
\bibitem [{\citenamefont {Abrams}\ and\ \citenamefont {Lloyd}(1999)}]{PhysRevLett.83.5162}%
  \BibitemOpen
  \bibfield  {author} {\bibinfo {author} {\bibfnamefont {D.~S.}\ \bibnamefont {Abrams}}\ and\ \bibinfo {author} {\bibfnamefont {S.}~\bibnamefont {Lloyd}},\ }\bibfield  {title} {\bibinfo {title} {Quantum algorithm providing exponential speed increase for finding eigenvalues and eigenvectors},\ }\href {https://doi.org/10.1103/PhysRevLett.83.5162} {\bibfield  {journal} {\bibinfo  {journal} {Phys. Rev. Lett.}\ }\textbf {\bibinfo {volume} {83}},\ \bibinfo {pages} {5162} (\bibinfo {year} {1999})}\BibitemShut {NoStop}%
\bibitem [{\citenamefont {Ge}\ \emph {et~al.}(2019)\citenamefont {Ge}, \citenamefont {Tura},\ and\ \citenamefont {Cirac}}]{ge2018faster}%
  \BibitemOpen
  \bibfield  {author} {\bibinfo {author} {\bibfnamefont {Y.}~\bibnamefont {Ge}}, \bibinfo {author} {\bibfnamefont {J.}~\bibnamefont {Tura}},\ and\ \bibinfo {author} {\bibfnamefont {J.~I.}\ \bibnamefont {Cirac}},\ }\bibfield  {title} {\bibinfo {title} {{Faster ground state preparation and high-precision ground energy estimation with fewer qubits}},\ }\href {https://doi.org/10.1063/1.5027484} {\bibfield  {journal} {\bibinfo  {journal} {Journal of Mathematical Physics}\ }\textbf {\bibinfo {volume} {60}},\ \bibinfo {pages} {022202} (\bibinfo {year} {2019})},\ \Eprint {https://arxiv.org/abs/https://pubs.aip.org/aip/jmp/article-pdf/doi/10.1063/1.5027484/13434463/022202\_1\_online.pdf} {https://pubs.aip.org/aip/jmp/article-pdf/doi/10.1063/1.5027484/13434463/022202\_1\_online.pdf} \BibitemShut {NoStop}%
\bibitem [{\citenamefont {Poulin}\ and\ \citenamefont {Wocjan}(2009{\natexlab{b}})}]{Poulin_2009}%
  \BibitemOpen
  \bibfield  {author} {\bibinfo {author} {\bibfnamefont {D.}~\bibnamefont {Poulin}}\ and\ \bibinfo {author} {\bibfnamefont {P.}~\bibnamefont {Wocjan}},\ }\bibfield  {title} {\bibinfo {title} {Preparing ground states of quantum many-body systems on a quantum computer},\ }\bibfield  {journal} {\bibinfo  {journal} {Physical Review Letters}\ }\textbf {\bibinfo {volume} {102}},\ \href {https://doi.org/10.1103/physrevlett.102.130503} {10.1103/physrevlett.102.130503} (\bibinfo {year} {2009}{\natexlab{b}})\BibitemShut {NoStop}%
\bibitem [{\citenamefont {Atia}\ and\ \citenamefont {Aharonov}(2017)}]{Atia_2017}%
  \BibitemOpen
  \bibfield  {author} {\bibinfo {author} {\bibfnamefont {Y.}~\bibnamefont {Atia}}\ and\ \bibinfo {author} {\bibfnamefont {D.}~\bibnamefont {Aharonov}},\ }\bibfield  {title} {\bibinfo {title} {Fast-forwarding of hamiltonians and exponentially precise measurements},\ }\bibfield  {journal} {\bibinfo  {journal} {Nature Communications}\ }\textbf {\bibinfo {volume} {8}},\ \href {https://doi.org/10.1038/s41467-017-01637-7} {10.1038/s41467-017-01637-7} (\bibinfo {year} {2017})\BibitemShut {NoStop}%
\bibitem [{\citenamefont {Reiher}\ \emph {et~al.}(2017)\citenamefont {Reiher}, \citenamefont {Wiebe}, \citenamefont {Svore}, \citenamefont {Wecker},\ and\ \citenamefont {Troyer}}]{Reiher_2017}%
  \BibitemOpen
  \bibfield  {author} {\bibinfo {author} {\bibfnamefont {M.}~\bibnamefont {Reiher}}, \bibinfo {author} {\bibfnamefont {N.}~\bibnamefont {Wiebe}}, \bibinfo {author} {\bibfnamefont {K.~M.}\ \bibnamefont {Svore}}, \bibinfo {author} {\bibfnamefont {D.}~\bibnamefont {Wecker}},\ and\ \bibinfo {author} {\bibfnamefont {M.}~\bibnamefont {Troyer}},\ }\bibfield  {title} {\bibinfo {title} {Elucidating reaction mechanisms on quantum computers},\ }\href {https://doi.org/10.1073/pnas.1619152114} {\bibfield  {journal} {\bibinfo  {journal} {Proceedings of the National Academy of Sciences}\ }\textbf {\bibinfo {volume} {114}},\ \bibinfo {pages} {7555–7560} (\bibinfo {year} {2017})}\BibitemShut {NoStop}%
\bibitem [{\citenamefont {McArdle}\ \emph {et~al.}(2020)\citenamefont {McArdle}, \citenamefont {Endo}, \citenamefont {Aspuru-Guzik}, \citenamefont {Benjamin},\ and\ \citenamefont {Yuan}}]{RevModPhys.92.015003}%
  \BibitemOpen
  \bibfield  {author} {\bibinfo {author} {\bibfnamefont {S.}~\bibnamefont {McArdle}}, \bibinfo {author} {\bibfnamefont {S.}~\bibnamefont {Endo}}, \bibinfo {author} {\bibfnamefont {A.}~\bibnamefont {Aspuru-Guzik}}, \bibinfo {author} {\bibfnamefont {S.~C.}\ \bibnamefont {Benjamin}},\ and\ \bibinfo {author} {\bibfnamefont {X.}~\bibnamefont {Yuan}},\ }\bibfield  {title} {\bibinfo {title} {Quantum computational chemistry},\ }\href {https://doi.org/10.1103/RevModPhys.92.015003} {\bibfield  {journal} {\bibinfo  {journal} {Rev. Mod. Phys.}\ }\textbf {\bibinfo {volume} {92}},\ \bibinfo {pages} {015003} (\bibinfo {year} {2020})}\BibitemShut {NoStop}%
\bibitem [{\citenamefont {Kim}\ \emph {et~al.}(2022)\citenamefont {Kim}, \citenamefont {Liu}, \citenamefont {Pallister}, \citenamefont {Pol}, \citenamefont {Roberts},\ and\ \citenamefont {Lee}}]{Kim_2022}%
  \BibitemOpen
  \bibfield  {author} {\bibinfo {author} {\bibfnamefont {I.~H.}\ \bibnamefont {Kim}}, \bibinfo {author} {\bibfnamefont {Y.-H.}\ \bibnamefont {Liu}}, \bibinfo {author} {\bibfnamefont {S.}~\bibnamefont {Pallister}}, \bibinfo {author} {\bibfnamefont {W.}~\bibnamefont {Pol}}, \bibinfo {author} {\bibfnamefont {S.}~\bibnamefont {Roberts}},\ and\ \bibinfo {author} {\bibfnamefont {E.}~\bibnamefont {Lee}},\ }\bibfield  {title} {\bibinfo {title} {Fault-tolerant resource estimate for quantum chemical simulations: Case study on li-ion battery electrolyte molecules},\ }\bibfield  {journal} {\bibinfo  {journal} {Physical Review Research}\ }\textbf {\bibinfo {volume} {4}},\ \href {https://doi.org/10.1103/physrevresearch.4.023019} {10.1103/physrevresearch.4.023019} (\bibinfo {year} {2022})\BibitemShut {NoStop}%
\bibitem [{\citenamefont {Sachdeva}\ and\ \citenamefont {Vishnoi}(2013{\natexlab{b}})}]{sachdeva_approximation_2013}%
  \BibitemOpen
  \bibfield  {author} {\bibinfo {author} {\bibfnamefont {S.}~\bibnamefont {Sachdeva}}\ and\ \bibinfo {author} {\bibfnamefont {N.}~\bibnamefont {Vishnoi}},\ }\href {https://doi.org/10.48550/arXiv.1309.4882} {\bibinfo {title} {Approximation {Theory} and the {Design} of {Fast} {Algorithms}}} (\bibinfo {year} {2013}{\natexlab{b}}),\ \bibinfo {note} {arXiv:1309.4882}\BibitemShut {NoStop}%
\bibitem [{abr()}]{abramowitz1966}%
  \BibitemOpen
  \href@noop {} {}\bibinfo {note} {M. Abramowitz, I. A. Stegun, ``Handbook of mathematical functions,'' Applied mathematics series 55, 62 (1966)}\BibitemShut {NoStop}%
\bibitem [{\citenamefont {Childs}\ \emph {et~al.}(2017{\natexlab{b}})\citenamefont {Childs}, \citenamefont {Kothari},\ and\ \citenamefont {Somma}}]{childs_quantum_2017}%
  \BibitemOpen
  \bibfield  {author} {\bibinfo {author} {\bibfnamefont {A.~M.}\ \bibnamefont {Childs}}, \bibinfo {author} {\bibfnamefont {R.}~\bibnamefont {Kothari}},\ and\ \bibinfo {author} {\bibfnamefont {R.~D.}\ \bibnamefont {Somma}},\ }\bibfield  {title} {\bibinfo {title} {Quantum algorithm for systems of linear equations with exponentially improved dependence on precision},\ }\href {https://doi.org/10.1137/16M1087072} {\bibfield  {journal} {\bibinfo  {journal} {SIAM Journal on Computing}\ }\textbf {\bibinfo {volume} {46}},\ \bibinfo {pages} {1920} (\bibinfo {year} {2017}{\natexlab{b}})}\BibitemShut {NoStop}%
\bibitem [{\citenamefont {Berry}\ \emph {et~al.}(2019)\citenamefont {Berry}, \citenamefont {Gidney}, \citenamefont {Motta}, \citenamefont {McClean},\ and\ \citenamefont {Babbush}}]{Berry2019qubitizationof}%
  \BibitemOpen
  \bibfield  {author} {\bibinfo {author} {\bibfnamefont {D.~W.}\ \bibnamefont {Berry}}, \bibinfo {author} {\bibfnamefont {C.}~\bibnamefont {Gidney}}, \bibinfo {author} {\bibfnamefont {M.}~\bibnamefont {Motta}}, \bibinfo {author} {\bibfnamefont {J.~R.}\ \bibnamefont {McClean}},\ and\ \bibinfo {author} {\bibfnamefont {R.}~\bibnamefont {Babbush}},\ }\bibfield  {title} {\bibinfo {title} {Qubitization of {A}rbitrary {B}asis {Q}uantum {C}hemistry {L}everaging {S}parsity and {L}ow {R}ank {F}actorization},\ }\href {https://doi.org/10.22331/q-2019-12-02-208} {\bibfield  {journal} {\bibinfo  {journal} {{Quantum}}\ }\textbf {\bibinfo {volume} {3}},\ \bibinfo {pages} {208} (\bibinfo {year} {2019})}\BibitemShut {NoStop}%
\bibitem [{\citenamefont {Lee}\ \emph {et~al.}(2021)\citenamefont {Lee}, \citenamefont {Berry}, \citenamefont {Gidney}, \citenamefont {Huggins}, \citenamefont {McClean}, \citenamefont {Wiebe},\ and\ \citenamefont {Babbush}}]{Lee2021}%
  \BibitemOpen
  \bibfield  {author} {\bibinfo {author} {\bibfnamefont {J.}~\bibnamefont {Lee}}, \bibinfo {author} {\bibfnamefont {D.~W.}\ \bibnamefont {Berry}}, \bibinfo {author} {\bibfnamefont {C.}~\bibnamefont {Gidney}}, \bibinfo {author} {\bibfnamefont {W.~J.}\ \bibnamefont {Huggins}}, \bibinfo {author} {\bibfnamefont {J.~R.}\ \bibnamefont {McClean}}, \bibinfo {author} {\bibfnamefont {N.}~\bibnamefont {Wiebe}},\ and\ \bibinfo {author} {\bibfnamefont {R.}~\bibnamefont {Babbush}},\ }\bibfield  {title} {\bibinfo {title} {Even more efficient quantum computations of chemistry through tensor hypercontraction},\ }\href {https://doi.org/10.1103/PRXQuantum.2.030305} {\bibfield  {journal} {\bibinfo  {journal} {PRX Quantum}\ }\textbf {\bibinfo {volume} {2}},\ \bibinfo {pages} {030305} (\bibinfo {year} {2021})}\BibitemShut {NoStop}%
\bibitem [{\citenamefont {Li}\ \emph {et~al.}(2019)\citenamefont {Li}, \citenamefont {Li}, \citenamefont {Dattani}, \citenamefont {Umrigar},\ and\ \citenamefont {Chan}}]{Li_2019}%
  \BibitemOpen
  \bibfield  {author} {\bibinfo {author} {\bibfnamefont {Z.}~\bibnamefont {Li}}, \bibinfo {author} {\bibfnamefont {J.}~\bibnamefont {Li}}, \bibinfo {author} {\bibfnamefont {N.~S.}\ \bibnamefont {Dattani}}, \bibinfo {author} {\bibfnamefont {C.~J.}\ \bibnamefont {Umrigar}},\ and\ \bibinfo {author} {\bibfnamefont {G.~K.-L.}\ \bibnamefont {Chan}},\ }\bibfield  {title} {\bibinfo {title} {{The electronic complexity of the ground-state of the FeMo cofactor of nitrogenase as relevant to quantum simulations}},\ }\href {https://doi.org/10.1063/1.5063376} {\bibfield  {journal} {\bibinfo  {journal} {The Journal of Chemical Physics}\ }\textbf {\bibinfo {volume} {150}},\ \bibinfo {pages} {024302} (\bibinfo {year} {2019})},\ \Eprint {https://arxiv.org/abs/https://pubs.aip.org/aip/jcp/article-pdf/doi/10.1063/1.5063376/14715992/024302\_1\_online.pdf} {https://pubs.aip.org/aip/jcp/article-pdf/doi/10.1063/1.5063376/14715992/024302\_1\_online.pdf} \BibitemShut {NoStop}%
\bibitem [{\citenamefont {Childs}\ and\ \citenamefont {Wiebe}(2012)}]{Childs_2012}%
  \BibitemOpen
  \bibfield  {author} {\bibinfo {author} {\bibfnamefont {A.~M.}\ \bibnamefont {Childs}}\ and\ \bibinfo {author} {\bibfnamefont {N.}~\bibnamefont {Wiebe}},\ }\bibfield  {title} {\bibinfo {title} {Hamiltonian simulation using linear combinations of unitary operations},\ }\bibfield  {journal} {\bibinfo  {journal} {Quantum Information and Computation}\ }\href {https://doi.org/10.26421/qic12.11-12} {10.26421/qic12.11-12} (\bibinfo {year} {2012})\BibitemShut {NoStop}%
\bibitem [{\citenamefont {Lee}\ \emph {et~al.}(2023)\citenamefont {Lee}, \citenamefont {Lee}, \citenamefont {Zhai}, \citenamefont {Tong}, \citenamefont {Dalzell}, \citenamefont {Kumar}, \citenamefont {Helms}, \citenamefont {Gray}, \citenamefont {Cui}, \citenamefont {Liu}, \citenamefont {Kastoryano}, \citenamefont {Babbush}, \citenamefont {Preskill}, \citenamefont {Reichman}, \citenamefont {Campbell}, \citenamefont {Valeev}, \citenamefont {Lin},\ and\ \citenamefont {Chan}}]{Lee_2023}%
  \BibitemOpen
  \bibfield  {author} {\bibinfo {author} {\bibfnamefont {S.}~\bibnamefont {Lee}}, \bibinfo {author} {\bibfnamefont {J.}~\bibnamefont {Lee}}, \bibinfo {author} {\bibfnamefont {H.}~\bibnamefont {Zhai}}, \bibinfo {author} {\bibfnamefont {Y.}~\bibnamefont {Tong}}, \bibinfo {author} {\bibfnamefont {A.~M.}\ \bibnamefont {Dalzell}}, \bibinfo {author} {\bibfnamefont {A.}~\bibnamefont {Kumar}}, \bibinfo {author} {\bibfnamefont {P.}~\bibnamefont {Helms}}, \bibinfo {author} {\bibfnamefont {J.}~\bibnamefont {Gray}}, \bibinfo {author} {\bibfnamefont {Z.-H.}\ \bibnamefont {Cui}}, \bibinfo {author} {\bibfnamefont {W.}~\bibnamefont {Liu}}, \bibinfo {author} {\bibfnamefont {M.}~\bibnamefont {Kastoryano}}, \bibinfo {author} {\bibfnamefont {R.}~\bibnamefont {Babbush}}, \bibinfo {author} {\bibfnamefont {J.}~\bibnamefont {Preskill}}, \bibinfo {author} {\bibfnamefont {D.~R.}\ \bibnamefont {Reichman}}, \bibinfo {author} {\bibfnamefont {E.~T.}\ \bibnamefont {Campbell}}, \bibinfo {author} {\bibfnamefont {E.~F.}\ \bibnamefont
  {Valeev}}, \bibinfo {author} {\bibfnamefont {L.}~\bibnamefont {Lin}},\ and\ \bibinfo {author} {\bibfnamefont {G.~K.-L.}\ \bibnamefont {Chan}},\ }\bibfield  {title} {\bibinfo {title} {Evaluating the evidence for exponential quantum advantage in ground-state quantum chemistry},\ }\bibfield  {journal} {\bibinfo  {journal} {Nature Communications}\ }\textbf {\bibinfo {volume} {14}},\ \href {https://doi.org/10.1038/s41467-023-37587-6} {10.1038/s41467-023-37587-6} (\bibinfo {year} {2023})\BibitemShut {NoStop}%
\bibitem [{\citenamefont {Gidney}\ and\ \citenamefont {Fowler}(2019)}]{Gidney2019efficientmagicstate}%
  \BibitemOpen
  \bibfield  {author} {\bibinfo {author} {\bibfnamefont {C.}~\bibnamefont {Gidney}}\ and\ \bibinfo {author} {\bibfnamefont {A.~G.}\ \bibnamefont {Fowler}},\ }\bibfield  {title} {\bibinfo {title} {Efficient magic state factories with a catalyzed {$|CCZ\rangle$} to {$2|T\rangle$} transformation},\ }\href {https://doi.org/10.22331/q-2019-04-30-135} {\bibfield  {journal} {\bibinfo  {journal} {{Quantum}}\ }\textbf {\bibinfo {volume} {3}},\ \bibinfo {pages} {135} (\bibinfo {year} {2019})}\BibitemShut {NoStop}%
\end{thebibliography}%

\appendix

\section{Proof of Lemma \ref{cheb}}
\label{sec:chebyshev_proof}

\begin{proof}
Let $U_{A}$ be given as in Def. \ref{def:qubitizedoracle}. Then 
\begin{align}\label{eq:explicit_Uj}
U^{j}_{A}&=\bigoplus_{\gamma}\left(e^{-i\vartheta_{\gamma}Y_{\gamma}}\right)^{j}\notag\\
&=\bigoplus_{\gamma}{\begin{bmatrix}
\cos(j\arccos{\lambda_\gamma}) & -\sin(j\arccos{\lambda_\gamma})\\
\sin(j\arccos{\lambda_\gamma}) & \cos(j\arccos{\lambda_\gamma})
\end{bmatrix}}
.
\end{align}
Since $\cos(j\arccos{x})=\mathcal{T}_{j}(x)$, where $\mathcal{T}_j$ is the $j$-th Chebyshev polynomial of the first kind, it follows that 
\begin{equation}
\big(\bra{0}_a\otimes \mathds{1}_s\big)\,U_A^j\,\big(\ket{0}_a\otimes \mathds{1}_s\big)=\bigoplus_{\gamma}\mathcal{T}_{j}(\lambda_\gamma)\ketbra{\lambda_\gamma}_{s}=\mathcal{T}_j(A).
\end{equation}
\end{proof}

\section{Hadamard test for block encodings}
\label{sec:Hadamard_test}

In Algorithms 1 and 2, for each $j_\alpha$ or $j_\alpha,l_\alpha$ sampled, only one measurement shot is obtained. As preparation for proving the correctness of the algorithms (Theorem\ \ref{main_lemma}), in this appendix, we show what would result from the Hadamard tests if we performed several measurement shots with a fixed circuit, i.e., fixed $j_\alpha$ (Lemma\ \ref{lem:had_circuit_1}) or $j_\alpha,l_\alpha$ (Lemma\ \ref{lem:had_circuit_2}). Moreover, in Lemma\ \ref{lemma_hadamardsecondkind} we show how to extend the construction to Chebyshev polynomials of the second kind.
Throughout this section, we use the notation $\text{c}U$ to refer to a gate $U$ controlled by the state $\ket{1}$ of the ancilla $1$ (the top register in both circuits of Fig.\ \ref{fig:main}, i.e., the Hadamard qubit ancilla). We use  $\bar{\text{c}}U$ for when the action of the unitary $U$ is controlled by the ancilla $1$ being in state $\ket{0}$.

\begin{lemma}[Circuit for Algorithm 1]\label{lem:had_circuit_1}
A single-shot measurement on the Hadamard test of Fig.~\ref{fig:main}.a) yields a random variable $\mathbf{Had}(U^{j}_{A},\ket{\phi},\ket{\psi})\in\{-1,1\}$ that satisfies
\begin{equation}
\mathbb{E}\Big[\mathbf{Had}(U^{j}_{A},\ket{\phi},\ket{\psi})\Big]=\begin{cases}
\Re{\mel{\phi}{\mathcal{T}_j(A)}{\psi}}\text{, if }B=\mathds{1},\\
\Im{\mel{\phi}{\mathcal{T}_j(A)}{\psi}}\text{, if }B=S^{\dagger},
\end{cases}    
\end{equation}
\end{lemma}

\begin{proof}
Here we prove it for the case $B=\mathds{1}$, since the other follows in an analogous way. We have that,
\begin{equation}
\mathbb{E}\Big[\mathbf{Had}(U^{j}_{A},\ket{\phi},\ket{\psi})\Big]=\Tr{X_{1}\,G(\ketbra{+}_{1}\otimes\varrho_{a,s})G^{\dagger}},
\end{equation}
with $G=\Bar{\text{c}}U_{\phi}\,\text{c}U^{j}_{A}\,\text{c}U_{\psi}$ and $\varrho_{a,s}=\ketbra{0}_{a}\otimes\ketbra{0}_{s}$. Computing the expression inside the trace we get
\begin{multline}
X_{1}\,G(\ketbra{+}_{1}\varrho_{a,s})G^{\dagger}=\frac{1}{2}\ketbra{0}_{1}(U^{j}_{A}U_{\psi}\varrho_{a,s}U^{\dagger}_{\phi})+\\
+\frac{1}{2}\ketbra{1}_{1}(U_{\phi}\varrho_{a,s}U^{\dagger}_{\psi}U^{j\dagger}_{A})+\frac{1}{2}\ketbra{1}{0}_{1}(U_{\phi}\varrho_{a,s}U^{\dagger}_{\phi})+\\
+\frac{1}{2}\ketbra{0}{1}_{1}(U^{j}_{A}U_{\psi}\varrho_{a,s}U^{\dagger}_{\psi}U^{j\dagger}_{A}).
\end{multline}
Therefore,
\begin{multline}
\mathbb{E}\Big[\mathbf{Had}(U^{j}_{A},\ket{\phi},\ket{\psi})\Big]=\Re\{\Tr_{a,s}[U^{j}_{A}\ketbra{0}_{a}\ketbra{\psi}{\phi}_{s}]\}=\\
=\Re\{\Tr_{s}[\mathcal{T}_j(A)\ketbra{\psi}{\phi}_{s}]\}=\Re\{\mel{\phi}{\mathcal{T}_j(A)}{\psi}\}
\end{multline}
\end{proof}

\begin{lemma}[Circuit for Algorithm 2]\label{lem:had_circuit_2}
A single-shot measurement on the Hadamard test of Fig.~\ref{fig:main}.b) yields a random variable $\mathbf{Had}(U^{l}_{A},U^{j}_{A},O)$ that satisfies
\begin{equation}
\mathbb{E}\Big[\mathbf{Had}(U^{j}_{A},U^{l}_{A},O)\Big]=\Re\{\Tr[O\mathcal{T}_j(A)\varrho_{s}\mathcal{T}_{l}(A)]\}
\end{equation}
\end{lemma}

\begin{proof}
We have that
\begin{equation}
\mathbb{E}\Big[\mathbf{Had}(U^{j}_{A},U^{l}_{A},O)\Big]=\Tr{OX_{1}M(\ketbra{+}_{1}\otimes\varrho_{a,s})M^{\dagger}},
\end{equation}
where $M=\Bar{\text{c}}U^{l}_{A}\,\text{c}U^{j}_{A}$ and $\varrho_{a,s}=\ketbra{0}_{a}\otimes\varrho_{s}$. Note that
\begin{multline}
X_{1}M(\ketbra{+}_{1}\varrho_{a,s})M^{\dagger}=\frac{1}{2}\ketbra{0}_{1}(U^{j}_{A}\varrho_{a,s}U^{l\dagger}_{A})+\\
+\frac{1}{2}\ketbra{1}_{1}(U^{l}_{A}\varrho_{a,s}U^{j\dagger}_{A})+\frac{1}{2}\ketbra{1}{0}_{1}(U^{l}_{A}\varrho_{a,s}U^{l\dagger}_{A})+\\
+\frac{1}{2}\ketbra{0}{1}_{1}(U^{j}_{A}\varrho_{a,s}U^{j\dagger}_{A}),
\end{multline}
which implies
\begin{multline}
\mathbb{E}\Big[\mathbf{Had}(U^{j}_{A},U^{l}_{A},O)\Big]=\frac{1}{2}\Tr_{a,s}[O(U^{j}_{A}\varrho_{a,s}U^{l\dagger}_{A})]+\\
+\frac{1}{2}\Tr_{a,s}[O(U^{l}_{A}\varrho_{a,s}U^{j\dagger}_{A})]=\frac{1}{2}\Tr_{s}[O(\mathcal{T}_j(A)\varrho_{s}T^{\dagger}_{l}(A))]+\\
+\frac{1}{2}\Tr_{s}[O(\mathcal{T}_j(A)\varrho_{s}T^{\dagger}_{l}(A))]=\Tr_{s}[O(\mathcal{T}_j(A)\varrho_{s}\mathcal{T}_l(A))],
\end{multline}
since $A$ and $O$ are Hermitian.
\end{proof}

\begin{lemma}\label{lemma_hadamardsecondkind}  If we modify the Hadamard test in Fig.~\ref{fig:main}.a) by applying a $\Bar{\text{c}}NOT_{1,a}$ after any other gate, the random variable $\textbf{Had}_{alt}(U^{j}_{A},\ket{\phi},\ket{\psi})$ satisfies for $B=\mathds{1}$
\begin{equation}
\mathbb{E}\Big[\textbf{Had}_{alt}(U^{j}_{A},\ket{\phi},\ket{\psi})\Big]=\Re\{\mel{\phi}{\sqrt{\mathds{1}-A^{2}}\,\mathcal{U}_{j-1}(A)}{\psi}\}
\end{equation}
and
\begin{equation}
\mathbb{E}\Big[\textbf{Had}_{alt}(U^{j}_{A},\ket{\phi},\ket{\psi})\Big]=\Im\{\mel{\phi}{\sqrt{\mathds{1}-A^{2}}\,\mathcal{U}_{j-1}(A)}{\psi}\}
\end{equation}
for $B=S^{\dagger}$.
\end{lemma}

\begin{proof}
Again, here we only show the case $B=\mathds{1}$ since the other case is analogous. From the description of the circuit in Fig.~\ref{fig:main}.a) and the added modification we have that
\begin{equation}
\mathbb{E}\Big[\textbf{Had}_{alt}(U^{j}_{A},\ket{\phi},\ket{\psi})\Big]=\Tr[X_{1}G_{alt}(\ketbra{+}_{1}\varrho_{a,s})G_{alt}^{\dagger}]
\end{equation}
with $G_{alt}=\Bar{\text{c}}NOT_{1,a}\,\Bar{\text{c}}U_{\phi}\,\text{c}U^{j}_{A}\,\text{c}U_{\psi}$ and $\varrho_{a,s}=\ketbra{0}_{a}\otimes\ketbra{0}_{s}$. We have that
\begin{multline}
X_{1}G_{alt}(\ketbra{+}_{1}\varrho_{a,s})G_{alt}^{\dagger}=\frac{\ketbra{1}{0}_{1}\ketbra{1}_{a}\ketbra{\phi}_{s}}{2}+\\
+\frac{\ketbra{0}{0}_{1}U^{j}_{A}(\ketbra{0}{1}_{a}\ketbra{\psi}{\phi}_{s})}{2}+\frac{\ketbra{1}{1}_{1}(\ketbra{1}{0}_{a}\ketbra{\phi}{\psi}_{s})U^{j\dagger}_{A}}{2}+\\
+\frac{\ketbra{0}{1}_{1}U^{j}_{A}(\ketbra{0}_{a}\ketbra{\psi}_{s})U^{j\dagger}_{A}}{2},
\end{multline}
from which we obtain
\begin{equation}
\Tr[X_{1}G_{alt}(\ketbra{+}_{1}\varrho_{a,s})G_{alt}^{\dagger}]=\Re\{\mel{\phi}{\mel{1}{U^{j}_{A}}{0}_{a}}{\psi}_{s}\}.
\end{equation}
From Eq.\ \eqref{eq:explicit_Uj} and the trigonometric definition of the Chebyshev polynomial $\mathcal{U}_{j-1}(x)$ we have that
\begin{equation}
\big(\bra{1}_a\otimes \mathds{1}_s\big)\,U_A^j\,\big(\ket{0}_a\otimes \mathds{1}_s\big)=\sqrt{\mathds{1}-A^{2}}\,\mathcal{U}_{j-1}(A),
\end{equation}
proving the statement.
\end{proof}

\section{Proof of Theorem \ref{main_lemma}}
\label{sec:main_lemma_proof}

\begin{proof}
We provide different proofs for each $P$. For $P=1$, note that
\begin{equation}
\mel{\phi}{\Tilde{f}(A)}{\psi}=\sum^{k}_{j=0}a_{j}\mel{\phi}{\mathcal{T}_j(A)}{\psi},
\end{equation}
and from Lemma \ref{lem:had_circuit_1} for $B=\mathds{1}$,
\begin{multline}
\Re\{\mel{\phi}{\Tilde{f}(A)}{\psi}\}=\sum^{k}_{j=0}a_{j}\,\mathbb{E}\Big[\textbf{Had}(U^{j}_{A},\ket{\phi},\ket{\psi})\Big]\\
=\sum^{k}_{j=0}\sum_{b\in\{-1,1\}}p_{j}\,p(b|j)\,\norm{\vv{a}}_{1}\sgn(a_{j})\,b.
\end{multline}
The last line implies that, for $\alpha<S^{(1)}$, $\frac{1}{S^{(1)}}\sum_\alpha \norm{\vv{a}}_{1}\sgn(a_{j_\alpha})\,b^{(1)}_\alpha$ with $j_\alpha$ sampled from the distribution $p_j$ and $b^{(1)}_\alpha$ sampled from the conditional distribution $p(b|j)$ by running the Hadamard test in Figure \fig{main}a) is an unbiased estimator of $\Re\{\mel{\phi}{\Tilde{f}(A)}{\psi}\}$. The case for $B=S^{\dagger}$ is essentially the same, implying that for $S^{1}\ge\alpha<2S^{(1)}-1$, $\frac{1}{S^{(1)}}\sum_\alpha \norm{\vv{a}}_{1}\sgn(a_{j_\alpha})\,b^{(1)}_\alpha$ is an unbiased estimator of $\Im\{\mel{\phi}{\Tilde{f}(A)}{\psi}\}$. 

Similarly, for $P=2$ note that
\begin{equation}
\Tr[O\Tilde{f}(A)\varrho\Tilde{f}^{\dagger}(A)]=\sum^{k}_{j,l=0}a_{j}\,a_{l}\Tr[O\mathcal{T}_j(A)\varrho \mathcal{T}_j(A)],
\end{equation}
which by Lemma \ref{lem:had_circuit_2} implies that
\begin{multline}
\Tr[O\Tilde{f}(A)\varrho\Tilde{f}^{\dagger}(A)]=\sum^{k}_{j,l=0}a_{j}\,a_{l}\,\mathbb{E}\Big[\textbf{Had}(U^{j}_{A},U^{l}_{A},O)\Big]\\
=\sum^{k}_{j,l=0}\sum_{b}\sum_{\omega}p_{j}\,p_{l}\,p(b,\omega|j,l)\,\norm{\vv{a}}^{2}_{1}\sgn(a_{j}a_{l})\,b\,\omega,
\end{multline}
where $b\in\{-1,1\}$ and $\omega$ in the spectrum of $O$. The last line implies that $\frac{1}{S^{(2)}}\sum_\alpha \norm{\vv{a}}^{2}_{1}\sgn(a_{j_{\alpha}})\sgn(a_{l_{\alpha}})b^{(2)}_{\alpha}\omega_{\alpha}$ with $j_\alpha,l_\alpha$ both sampled from the distribution $p_j$ and $b^{(2)}_\alpha,\omega_\alpha$ sampled from the conditional distribution $p(b,\omega|j,l)$ by running the Hadamard test in Figure \fig{main}b) is an unbiased estimator of $\Tr[O\Tilde{f}(A)\varrho\Tilde{f}^{\dagger}(A)]$. 
\end{proof}

\section{Proof of Theorem \ref{main_lemma2}}
\label{sec:main_lemma2_proof}
Now we prove Theorem \ref{main_lemma2}.

\begin{proof}
Let $\Tilde{f}$ be a $\nu$-approximation to the function $f$. Then, for any two states $\ket{\psi},\ket{\phi}$ we have that
\begin{equation}
\begin{split}
&\abs{\mel{\phi}{f(A)}{\psi}-\mel{\phi}{\Tilde{f}(A)}{\psi}}\leq
\norm{f(A)-\Tilde{f}(A)}_{op}\le\nu.
\end{split}
\end{equation}
Here, $\norm{\cdot}_{op}$ is the operator norm, which for Hermitian operators coincides with the spectral norm $\norm{\cdot}$. From Theorem  \ \ref{main_lemma} we have $\mathbb{E}\big[\Tilde{z}^{(1)}\big]=\mel{\phi}{\Tilde{f}(A)}{\psi}$ and thus
\begin{equation}
\begin{split}
\abs{z^{(1)}- \Tilde{z}^{(1)}}&\leq\abs{z^{(1)}-\mathbb{E}\big[\Tilde{z}^{(1)}\big]}+\abs{\mathbb{E}\big[\Tilde{z}^{(1)}\big]-\Tilde{z}^{(1)}}\\
&=\nu+\abs{\Tilde{z}^{(1)}-\mathbb{E}\big[\Tilde{z}^{(1)}\big]},
\end{split}
\end{equation}
$\Tilde{z}^{(1)}$ will be an $\epsilon$-precise estimator of $z^{(1)}$ if and only if
\begin{equation}
\abs{\Tilde{z}^{(1)}-\mathbb{E}\big[\Tilde{z}^{(1)}\big]}\le\epsilon-\nu^{(1)},
\end{equation}
where $\nu^{(1)}=\nu$. Notice from the definition of the estimator $\Tilde{z}^{(1)}$ that the range of its real and imaginary parts each is $\{-\norm{\vv{a}}_{1},\norm{\vv{a}}_{1}\}$. Therefore, Hoeffding's inequality implies that the number of samples $S^{(1)}$ necessary to attain an $\epsilon$-precise estimate with $1-\delta$ confidence is
\begin{equation}
S^{(1)}=\frac{4\norm{\vv{a}}^{2}_{1}}{(\epsilon-\nu^{(1)})^{2}}\log(\frac{4}{\delta}),
\end{equation}
where the real and imaginary parts are sampled independently. By setting the approximation error $\nu^{(1)}$ to $\epsilon/2$ we arrive at the claimed result.

In the case of problem \ref{problem.2}, we have to bound $\abs{z^{(2)}-\Tilde{z}^{(2)}}$. First, notice that
\begin{equation}
\begin{split}
&\abs{\Tr[O\,f(A)\,\varrho\,f^{\dagger}(A)]-\Tr[O\,\Tilde{f}(A)\,\varrho\,\Tilde{f}^{\dagger}(A)]}\\
\leq&\abs{\Tr[O(f(A)-\Tilde{f}(A))\,\varrho\, f^{\dagger}(A)]}\\
+&\abs{\Tr[O(f(A)-\Tilde{f}(A))\,\varrho\,(f^{\dagger}(A)-\Tilde{f}^{\dagger}(A))]}\\
+&\abs{\Tr[O\,f(A)\,\varrho\,(f^{\dagger}(A)-\Tilde{f}^{\dagger}(A))]},
\end{split}
\end{equation}
where we used the triangle inequality after adding and subtracting the terms $\Tr[O\,f(A)\,\varrho\,\Tilde{f}(A)]$, $\Tr[O\,f(A)\,\varrho\, f(A)]$ and $\Tr[O\,\Tilde{f}(A)\,\varrho\, f(A)]$. 

By means of the Holder tracial matrix inequality
\begin{equation}
\abs{\Tr[AB^{\dagger}]}\leq\abs{\Tr[A]}\,\norm{B},
\end{equation}
we can now show that
\begin{equation}
\begin{split}
&\abs{\Tr[Of(A)\varrho f^{\dagger}(A)]-\Tr[O\Tilde{f}(A)\varrho\Tilde{f}^{\dagger}(A)]}\leq\\
&\norm{O}\norm{f(A)-\Tilde{f}(A)}\norm{f^{\dagger}(A)}+\norm{O}\norm{f(A)-\Tilde{f}(A)}^{2}+\\
&+\norm{O}\norm{f(A)-\Tilde{f}(A)}\norm{f(A)},
\end{split}
\end{equation}
where we used that $\abs{\Tr[\varrho]}=1$. Therefore, 
\begin{equation}
\abs{z^{(2)}-\Tilde{z}^{(2)}}\le\nu^{(2)}:=\nu\norm{O}\left(2\norm{f(A)}+\nu\right).
\end{equation}

Similarly to the first case, $\Tilde{z}^{(2)}$ will be an $\epsilon$-precise estimate of $z^{(2)}$ if and only if
\begin{equation}
\abs{\Tilde{z}^{(2)}-\mathbb{E}\Big[\Tilde{z}^{(2)}\Big]}\le\epsilon-\nu^{(2)}.
\end{equation}
Since the range of possible values of $\Tilde{z}^{(2)}$ lie in $\{-\norm{O}\norm{\vv{a}}^{2}_{1},\norm{O}\norm{\vv{a}}^{2}_{1}\}$, Hoeffding's inequality implies that the number of samples necessary to attain an $\epsilon$-precise estimate with $1-\delta$ confidence is
\begin{equation}
S_{2}=\frac{2\norm{O}^{2}\norm{\vv{a}}^{4}_{1}}{(\epsilon-\nu^{(2)})^{2}}\log(\frac{2}{\delta}).
\end{equation}
The claim then follows by setting the approximation error $\nu^{(2)}=\epsilon/2$.
\end{proof}

\section{Error robustness analysis} \label{sec:error_analysis_proof}

In this appendix, we prove Theorem \ref{thm:error_robust}. Before the proof, two auxiliary lemmas are presented. In Lemma\ \ref{lem:dep_chanel} we show that the action of a depolarizing channel can be commuted with the action of a generic unitary channel. In particular, we consider the error model where the unitary is called many times, and each time it acts it is followed by a depolarizing channel. We show that this is equivalent to acting with all the depolarizing channels first and then all the unitaries. Lemma\ \ref{lem:error_j} shows how the presence of a depolarizing noise in the block-encoding oracle affects the result of the Hadamard test in Fig. \ref{fig:main}a with a fixed $j$. With these two lemmas, the proof of Theorem \ref{thm:error_robust} is obtained straightforwardly. At the end of the section, we also analyze the influence on the estimate of Algo. 1 of coherent noise in the block-encoding of $A$.

\begin{lemma}\label{lem:dep_chanel}
Let $U$ be any unitary operator over $\mathcal{H}_{1}\otimes\mathcal{H}_{a}\otimes\mathcal{H}_{s}$. Then, if $\Lambda$ is the full depolarization channel with strength $p$, we have that, for any $\sigma\in\mathcal{H}_{1}\otimes\mathcal{H}_{a}\otimes\mathcal{H}_{s}$
\begin{equation}
(\Lambda\circ U)[\sigma]=(U\circ\Lambda)[\sigma],
\end{equation}
and, in particular,
\begin{equation}
(\Lambda\circ U)^{j}[\sigma]=U^{j}\circ\Lambda^{j}[\sigma],
\end{equation}
where $U[\sigma]$ is the action of $U$ as a unitary channel over $\sigma$.
\end{lemma}

\begin{proof}
First, we have that
\begin{equation}
\begin{split}
\Lambda\circ U[\sigma]&=\Lambda[\left(U\sigma U^{\dagger}\right)]\\
&=(1-p)\left(U\sigma U^{\dagger}\right)+p\frac{\mathds{1}}{D_\text{tot}}\\
&=U\left((1-p)\sigma+p\frac{\mathds{1}}{D_\text{tot}}\right)U^{\dagger}\\
&=U\circ\Lambda[\sigma],
\end{split}
\end{equation}
where $D_\text{tot}$ is the dimension of $\mathcal{H}_{1}\otimes\mathcal{H}_{a}\otimes\mathcal{H}_{s}$. Therefore, since $\Lambda$ and $U$ commute as channels, we can pass all $U$ to the left of the $\Lambda$ copies, proving the lemma.
\end{proof}

\begin{lemma}\label{lem:error_j}
Let $\textbf{Had}(U^{j}_{A},\Lambda,\ket{\phi},\ket{\psi})$ be the random variable obtained by the Hadamard test described in Fig.~\ref{fig:main}a when we use $\Lambda\circ \text{c}U_{A}$ in the place of $\text{c}U_{A}$. Then, for $B=\mathds{1}$
\begin{equation}
\mathbb{E}\Big[\textbf{Had}(U^{j}_{A},\Lambda,\ket{\phi},\ket{\psi})\Big]=(1-p)^{j}\Re\{\mel{\phi}{\mathcal{T}_j(A)}{\psi}\},
\end{equation}
and for $B=S^{\dagger}$
\begin{equation}
\mathbb{E}\Big[\textbf{Had}(U^{j}_{A},\Lambda,\ket{\phi},\ket{\psi})\Big]=(1-p)^{j}\Im\{\mel{\phi}{\mathcal{T}_j(A)}{\psi}\}.
\end{equation}
\end{lemma}

\begin{proof}
Note that, by Fig.~\ref{fig:main}.a)
\begin{equation}
\mathbb{E}\Big[\textbf{Had}(U^{j}_{A},\Lambda,\ket{\phi},\ket{\psi})\Big]=\Tr[X_{1}\left(N[\ketbra{+}_{1}\varrho_{a,s}]\right)],
\end{equation}
where 
\begin{equation}
N=\Bar{\text{c}}U_{\phi}\circ{(\Lambda\circ \text{c}U_{A})}^{j}\circ \text{c}U_{\psi}\circ B,
\end{equation}
and $\varrho_{a,s}=\ketbra{0}_{a}\ketbra{0}_{s}$. Now, using Lemma \ref{lem:dep_chanel} we have that
\begin{equation}
\begin{split}
N&=\Bar{\text{c}}U_{\phi}\circ \text{c}U_{A}^{j}\circ\Lambda^{j}\circ \text{c}U_{\psi}\circ B\\
&=\Bar{\text{c}}U_{\phi}\circ \text{c}U_{A}^{j}\circ \text{c}U_{\psi}\circ B\circ\Lambda^{j}=G\circ\Lambda^{j},
\end{split}
\end{equation}
where $G$ is the unitary channel $G[\sigma]=G\sigma G^{\dagger}$ with $G=\Bar{\text{c}}U_{\phi}\,\text{c}U^{j}_{A}\,\text{c}U_{\psi}\,B$. 
Therefore, 
\begin{equation}
\begin{split}
&N[\ketbra{+}_{1}\varrho_{a,s}]=G\circ\Lambda^{j}[\ketbra{+}_{1}\varrho_{a,s}]=\\
&(1-p)^j G(\ketbra{+}_{1}\varrho_{a,s})G^{\dagger}+(1-(1-p)^{j})\frac{\mathds{1}}{D_\text{tot}}.
\end{split}
\end{equation}
Now, by the proof of lem.\ref{lem:had_circuit_1}, and the fact that $X_{1}$ is trace-less, we have that
\begin{equation}
\begin{split}
\mathbb{E}\Big[\textbf{Had}(U^{j}_{A},\Lambda,\ket{\phi},\ket{\psi})\Big]=(1-p)^{j}\mathbb{E}\Big[\textbf{Had}(U^{j}_{A},\ket{\phi},\ket{\psi})\Big],
\end{split}
\end{equation}
which is equivalent to the statement of the lemma given any choice of $B$
\end{proof}

Now, we are ready for the proof of Theorem \ref{thm:error_robust}.

\begin{proof}[Proof of Theorem \ref{thm:error_robust}]
By definition, we have that
\begin{equation}
\abs{\mathbb{E}\big[\Tilde{z}^{(1)}\big]-\mathbb{E}\big[\Tilde{z}^{(1),\Lambda}\big]}=\abs{\sum^{k}_{j=0}a_{j}(1-(1-p)^{j})\mel{\phi}{\mathcal{T}_j(A)}{\psi}}.
\end{equation}
By Bernoulli's inequality $1-(1-p)^{j}\leq pj$ for $p<1$, implying
\begin{equation}
\abs{\mathbb{E}\big[\Tilde{z}^{(1)}\big]-\mathbb{E}\big[\Tilde{z}^{(1),\Lambda}\big]}\leq p\abs{\sum^{k}_{j=0}j\,a_{j}\mel{\phi}{\mathcal{T}_j(A)}{\psi}}.   
\end{equation}
Now, using the fact that $\abs{\mel{\phi}{\mathcal{T}_j(A)}{\psi}}\leq1$ for any $\ket{\phi},\ket{\psi}$ and the triangle inequality one gets
\begin{equation}
\abs{\mathbb{E}\big[\Tilde{z}^{(1)}\big]-\mathbb{E}\big[\Tilde{z}^{(1),\Lambda}\big]}\leq p\sum^{k}_{j=0}j\,\abs{a_{j}}=p\norm{\vv{a}}_{1}\mathbb{E}[j].
\end{equation}
\end{proof}

\begin{lemma}[Noise robustness of QSVT] Let $U_{\Tilde{f}(A)}$ be the block-encoding of $\Tilde{f}(A)/\|\Tilde{f}\|$ resulting from an ideal QSVT circuit, with $\|\Tilde{f}\|=\text{max}_{x\in[-1,1]}|\Tilde{f}(x)|$. Let $\Tilde{z}^{(1)}_{\text{fq}}=\|\Tilde{f}\| \,\textbf{Had}(U_{\Tilde{f}(A)},\ket{\phi},\ket{\psi})$ and $\Tilde{z}^{(1),\Lambda}_{\text{fq}}= \|\Tilde{f}\|\,\textbf{Had}(U_{\Tilde{f}(A)},\Lambda,\ket{\phi},\ket{\psi})$ be the random variable obtained by a Hadamard test on $U_{\Tilde{f}(A)}$ and its noisy counterpart with $\Lambda$ acting after each oracle query, respectively. Then, $\mathbb{E}[\Tilde{z}^{(1)}_{\text{fq}}]=\text{Re}\{\mel{\phi}{\Tilde{f}(A)}{\psi}\}$  and  
\begin{align}\label{eq:noise_robustnessQSP}
\abs{\mathbb{E}\big[\Tilde{z}^{(1)}_{\text{fq}}\big]-\mathbb{E}\big[\Tilde{z}^{(1),\Lambda}_{\text{fq}}\big]} &\le
p_\text{fq}\,E^{(1)}_\text{fq}\leq\, p_\text{fq}\,\|\Tilde{f}\|\,k \,,
\end{align}    
with $E^{(1)}_\text{fq} := k \abs{\text{Re}\{\mel{\phi}{\Tilde{f}(A)}{\psi}\}}$.
\end{lemma}

\begin{proof}
First notice that the Hadamard test will work as in Lemma~\ref{lem:had_circuit_1} with $U_A^j$ substituted by $U_{\Tilde{f}(A)}$.
    The rest of the proof follows analogously to the proof of Lemma~\ref{lem:error_j} by noticing that the circuit always makes $k$ calls to the noisy oracle.
\end{proof}

So far, we considered only the effect of incoherent error. Now we consider the case where the source of error is a coherent imperfection in the block-encoding oracle of $A$, as we define below.
\begin{definition}
Given an Hermitian operator $A$ with $\norm{A}\leq1$, an $\varepsilon$-approximate block-encoding of $A$ is a unitary operator $U_{A,\varepsilon}$ satisfying
\begin{equation}
\norm{(\bra{0}_{a}\otimes\mathds{1}_{s})U_{A,\varepsilon}(\ket{0}_{a}\otimes\mathds{1}_{s})-A}_{tr}<\varepsilon,
\end{equation}
where $\norm{\,\,}_{tr}$ is the trace norm.
\end{definition}

Under the approximate oracle assumption, the next lemma shows how the error propagates to our Algo. 1 estimate result.

\begin{lemma}
Let $U_{A,\varepsilon}$ be an $\varepsilon$-approximate block-encoding of $A$, and let $U_{A}$ be an exact block encoding of $A$. If $\Tilde{z}^{(1)}$ is as in Eq. (\ref{eq:ztildeP}a) and $\Tilde{z}^{(1),\varepsilon}$ is the unbiased estimator obtained from the Hadamard test in Fig.~\ref{fig:main}a when exchanging $\text{c}U^{j}_{A}$ by $\text{c}U^{j}_{A,\varepsilon}$, with $(k\,\varepsilon)\leq1/2$,
then 
\begin{equation}
\abs{\mathbb{E}\big[\Tilde{z}^{(1)}\big]-\mathbb{E}\big[\Tilde{z}^{(1),\varepsilon}\big]}\leq\, 2\,\varepsilon\,\norm{\vv{a}}_{1}\,\mathbb{E}[j]\,,
\end{equation}
\end{lemma}

\begin{proof}
If $U_{A,\varepsilon}$ is an $\varepsilon$-approximate block-encoding then,
\begin{equation}
\norm{(\bra{0}_{a}\otimes\mathds{1}_{s})U^{j}_{A,\varepsilon}(\ket{0}_{a}\otimes\mathds{1}_{s})-\mathcal{T}_{j}(A)}_\text{tr}\le 2\,j\,\varepsilon,
\end{equation}
which implies that
\begin{equation}
\abs{(\bra{0}_{a}\bra{\phi}_{s})U^{j}_{A,\varepsilon}(\ket{0}_{a}\ket{\psi}_{s})-\mel{\phi}{\mathcal{T}_{j}(A)}{\psi}}\le 2\,j\,\varepsilon.
\end{equation}
Therefore, denoting by $\delta \tilde{z}^{(1)}\equiv\abs{\mathbb{E}\big[\Tilde{z}^{(1)}\big]-\mathbb{E}\big[\Tilde{z}^{(1),\varepsilon}\big]}$ the error in the estimator $z^{(1)}$ from replacing the exact block-encoding $U_A$ by the faulty $U_{A,\varepsilon}$,
\begin{align}
\delta \tilde{z}^{(1)} &= \abs{\sum^{k}_{j=0}a_{j}\left((\bra{0}_{a}\bra{\phi}_{s})U^{j}_{A,\varepsilon}(\ket{0}_{a}\ket{\psi}_{s})-\mel{\phi}{\mathcal{T}_{j}(A)}{\psi}\right)} \notag\\    
&\le \varepsilon\abs{\sum^{k}_{j=0}ja_{j}}\leq2\,\varepsilon\norm{\vv{a}}_{1}\mathbb{E}[j].
\end{align}
\end{proof}

\section{Specific function approximation results} \label{sec:explicit_examples}

\begin{table*}[htpb!]
    \centering
    \begin{tabular}{||c||c|c|c|c||}
     \hline
     \rule{0pt}{3ex}
     \rule[-1.5ex]{0pt}{0pt}
     {\bf Function} & {\bf Domain} & {\bf Truncation order ($k$)} & {\bf $\norm{\vv{a}}_{1}$} & {\bf $\norm{f}_\infty$} \\
     \hline
     \hline
     \rule{0pt}{3ex}
     \rule[-1.5ex]{0pt}{0pt}
     $x^{t}$ ($t\in\mathbb{N}$) & $[-1,1]$ & $\sqrt{2t\,\log\frac{2}{\nu}}$ & $\mathcal{O}(1)$ & $\mathcal{O}(1)$\\
     \hline
     \rule{0pt}{4ex}
     \rule[-2ex]{0pt}{0pt}
     $e^{-\beta x/2}$ ($\beta\in\mathbb{R}$) & $[-1,1]$ & $\sqrt{2\,\max\left\{\frac{e^2\beta}{2},\log\Big(\frac{2\, e^{\beta/2}}{\nu}\Big)\right\}\,\log\Big(\frac{4\,e^{\beta/2}}{\nu}\Big)}$ & $\mathcal{O}(e^{\beta/2})$ & $\mathcal{O}(e^{\beta/2})$ \\
     \hline
     \rule{0pt}{4ex}
     \rule[-1.5ex]{0pt}{0pt}
     $x^{-1}$ & $[-1,-1/\kappa]\cup[1/\kappa,1]$ &  $\quad 2\kappa\,\sqrt{\log\frac{\kappa}{\nu} \, \log(\frac{4\kappa^{2}}{\nu}\log\frac{\kappa}{\nu})}+1\quad $ & $\mathcal{O}\left(\kappa\sqrt{\log\frac{\kappa}{\nu}}\right)$ & $\mathcal{O}(\kappa)$ \\
     \hline
     \rule{0pt}{3ex}
     \rule[-1.5ex]{0pt}{0pt}
     $\theta(x)$ & $[-1,-\xi]\cup[\xi,1]$ & $\mathcal{O}\left(\frac{1}{\xi}\log\frac{1}{\nu}\right)$ & {$\mathcal{O}\left(\log(\frac{1}{\xi}\log\frac{1}{\nu})\right)$} & $\mathcal{O}(1)$ \\
     \hline
    \end{tabular}
    \caption{Chebyshev $\nu$-approximation data (approximation interval $\mathcal{I}$, truncation degree $k$ and coefficient 1-norm $\norm{\vv{a}}_1$) for our exemplary functions. 
    When a singular point is present at the origin, the approximation interval excludes a ball of fixed radius around it. 
    For comparison, we also provide the $\infty$-norm $\norm{f}_\infty:=\max_{x\in I} |f(x)|$, which in the matrix function setting plays the role of spectral norm $\norm{f(A)}$ of the matrix $f(A)$. Namely, the difference between $\norm{\vv{a}}_1$ and $\norm{f}_\infty$ translates into the difference between sample complexities between our semi-quantum algorithm and the corresponding fully-quantum one.}
    \label{tab:chebyshev_data}
\end{table*}

The following general lemma allows us to compute the 1-norm $\norm{\vv{a}}_1$ in certain cases.

\begin{lemma}\label{lemma_1norm} 
Let $\vv{a}=\{a_0,\ldots,a_k\}$ be the vector of Chebyshev coefficients of the polynomial $\Tilde{f}$ as in \eq{ftilde}. Then $\norm{\vv{a}}_1=|\Tilde{f}(1)|$ if all $a_j$'s have the same sign and $\norm{\vv{a}}_1=|\Tilde{f}(-1)|$ if they alternate signs. 
\end{lemma}

\begin{proof}
First notice that if all the coefficients have the same sign (respectively, opposite signs) then either $a_j=|a_j|$ or $a_j=-|a_j|$ for all $j$ (respectively, either $a_j=(-1)^j |a_j|$ or $a_j=(-1)^{j+1} |a_j|$ for all $j$). 
The claim then follows immediately from the fact that $\mathcal{T}_j(1)=1$ and $\mathcal{T}_j(-1)=(-1)^j$ for all $j$. 
\end{proof}

\subsection{Monomial function} \label{sec:explicit_examples_monomials}

\subsubsection{Polynomial approximation and its parameters}
The following lemma gives the Chebyshev approximation to the monomial function that we use to prove the complexity of our algorithm. The lemma is a version of \cite[Lemma 3.3]{Vishnoi2013} with explicit coefficients.
\begin{lemma}\label{lem:approx_mon}
Let $f:\mathbb{R}\rightarrow\mathbb{R}$ be such that $f(x)=x^{t}$ for some $t\in\mathbb{N}$. Then, the function $\Tilde{f}:[-1,1]\rightarrow\mathbb{R}$ given by the Chebyshev polynomial 
\begin{equation}
\Tilde{f}(x)=
 2^{1-t} \sideset{}{'}\sum^{k}_{\substack{j=0\\ t-j\text{ even}}}\binom{t}{(t-j)/2}\mathcal{T}_{j}(x),
\end{equation}
is a $\nu$-approximation of $f$ on $[-1,1]$, if and only if 
\begin{equation}
k\geq\sqrt{2t\log(\frac{2}{\nu})}.    
\end{equation}
\end{lemma}

\begin{proof}
See \cite{sachdeva_approximation_2013}, Theorem 3.2.
\end{proof}

\begin{lemma}
Let $\textbf{a}_\text{mon}$ be the vector of coefficients of the Chebyshev polynomial $\Tilde{f}_\text{mon}$. Then the one-norm is
\begin{equation}
\norm{\textbf{a}_\text{mon}}_{1}=1-\nu.
\end{equation}
\end{lemma}
\begin{proof}
Follows as corollary of Lemma \ref{lemma_1norm}.
\end{proof}

\begin{lemma}\label{lem:av_Q_mon}
Let $\textbf{a}_\text{mon}$ be the vector of coefficients of the Chebyshev polynomial $\Tilde{f}_\text{mon}$. Then, the average query complexity of estimating $\Tilde{z}^{(1)}$ is
\begin{equation}
\mathbb{E}[j]\leq\frac{1}{1-\nu}\sqrt{\frac{2t}{\pi}\frac{t}{t-1}}\,.
\end{equation}
\end{lemma}

\begin{proof}
First notice that 
\begin{align}
\mathbb{E}[j] &= \frac{1}{\norm{\vv{a}_\text{mon}}}\sum_{j=1}^k j\,|a_j| \le \frac{1}{\norm{\vv{a}_\text{mon}}}\sum_{j=1}^t j\,|a_j|\,,
\end{align}
since $k\le t$ by assumption. The summation on the righthand side can be computed exactly, but its precise value depends on whether $t$ is even or odd. For $t$ even, the only non-vanishing Chebyshev coefficients are the even order ones, namely $a_{2\ell}=2^{1-t} \binom{t}{(t-2\ell)/2}$, and 
\begin{align}\label{eq:Ej_monomial_even}
\sum_{j=1}^t j\,|a_j| = \sum_{\ell=1}^{t/2} 2^{2-t}\,\ell\,\binom{t}{(t-2\ell)/2} = 
2^{-t} \,t\,\binom{t}{t/2} \le \sqrt{\frac{2t}{\pi}}\,,
\end{align}
where we have used the inequality $\binom{t}{t/2}\le \frac{2^{t}}{\sqrt{\pi t/2}}$. 
Similarly, for odd $t$ the only non-vanishing coefficients are the odd order ones, with $a_{2\ell-1}=2^{1-t} \binom{t}{(t-2\ell+1)/2}$, and 
\begin{align}\label{eq:Ej_monomial_odd}
\sum_{j=1}^t j\,|a_j| &= \sum_{\ell=1}^{(t+1)/2} 2^{1-t}\,(2\ell-1)\,\binom{t}{(t-2\ell+1)/2} \notag\\
&= 
2^{1-t}\,t\,\binom{t-1}{(t-1)/2} \le 
\sqrt{\frac{2t}{\pi}\frac{t}{t-1}}\,.
\end{align}
Notice that the bound \eq{Ej_monomial_even} contains the one in \eq{Ej_monomial_odd}, and therefore can be taken to hold for all $t$. The claim then follows by incorporating $\norm{\vv{a}_\text{mon}}$ from Lemma \ref{lem:av_Q_mon}.
\end{proof}

\subsubsection{Problem statement and complexities of the solution}

We  consider a  reversible and ergodic Markov chain on $N=2^n$ discrete states with row stochastic transition matrix $\mathcal{P}$ satisfying $\sum_{k}\mathcal{P}_{jk}=1$ whose stationary state, i.e the eigenstate corresponding to eigenvalue $\lambda=1$ has the Gibbs amplitudes $\bpi=(e^{-\beta E_0}/Z_\beta,\cdots,e^{-\beta E_{D-1}}/Z_\beta)$, with $\{E_l\}_{l\in[D]}$ the eigenvalues of the Hamiltonian $H$. 
The eigenvalues $\lambda$ of $\mathcal{P}$ lie in the interval $(-1,1]$ with $\lambda=1$ non-degenerate.
The so-called {\it discriminant matrix} $A$ is defined by $A_{jk}=\sqrt{\mathcal{P}_{jk}\mathcal{P}_{kj}}$. For the sake of our algorithm, the relevant properties of $A$ are Hermiticity and the fact that $A$ has the same spectrum of $\mathcal{P}$ with $\sqrt{\bpi}$ the eigenvector corresponding to eigenvalue $\lambda=1$.
In particular,  $A^t$ for $t=\mathcal{O}(\Delta^{-1})$ becomes approximately a projector over the state $\sqrt{\bpi}$ as we state
in the following lemma (here $\Delta=\min_{\{\lambda\}/1}(1-|\lambda|)$ is  the spectral gap of $A$). 

\begin{lemma}[Determining the monomial degree]\label{lemma_partition_function_power_method}
Let $\nu'>0$, $A$ a Markov-chain discriminant matrix with spectral gap $0<\Delta<1$, and $\ket{\vv{y}}$ a spin configuration with energy $E_{\vv{y}}$. For any integer $t\ge\frac{\log(1/\nu')}{\Delta}$, it holds that 
\begin{align}
\left|\!\mel{\vv{y}}{A^t}{\vv{y}} - \mel{\vv{y}}{\Pi_{\bpi}}{\vv{y}}\right| \le \nu'
\end{align}
\end{lemma}

\begin{proof}
Let $A=\sum_{\lambda}\lambda\ketbra{\lambda}{\lambda}=\Pi_{\vv{\pi}}+\sum_{\lambda\ne1}\lambda\ketbra{\lambda}{\lambda}$ be the spectral decomposition of $A$, where $\lambda$ denote its eigenvalues, $\{\ket{\lambda}\}$ the corresponding eigenstates, and $\Pi_{\vv{\pi}}=\ketbra{\sqrt{\bpi}}$ is the projector onto the "purified" Gibbs state. It follows that
\begin{align}
\left|\!\mel{\vv{y}}{A^t-\Pi_{\vv{\pi}}}{\vv{y}}\right| 
&\le \sum_{\lambda\ne1}|\lambda|^t\!\cdot\!|\!\braket{\vv{y}}{\lambda}\!|^2 \notag\\
&\le (1-\Delta)^t.
\end{align}
Hence, given some tolerated error $\nu'$, one can ensure that $\left|\!\mel{\vv{y}}{A^t}{\vv{y}} - \mel{\vv{y}}{\Pi_{\vv{\pi}}}{\vv{y}}\right| \le \nu'$ for any
\begin{align}\label{eq:tbound}
t \ge \frac{\log 1/\nu'}{\Delta}\ge \frac{\log 1/\nu'}{\log(1-\Delta)^{-1}} \,.
\end{align}
\end{proof}

\begin{lemma}\label{lem:estimate_z} Let $\tilde{z}^{(1)}$ be an estimate to $\mel{\vv{y}}{\Pi_{\bpi}}{\vv{y}}$ up to additive error $\epsilon<\mel{\vv{y}}{\Pi_{\bpi}}{\vv{y}}/2$.  If $\epsilon\leq \epsrel\big(2\,e^{\beta E_{\vv{y}}}Z_\beta\big)$, then $\Tilde{Z}_\beta:=\frac{e^{-\beta E_{\vv{y}}}}{\tilde{z}^{(1)}}$ is an estimate for the partition function $Z_\beta$ such that
\begin{equation}\label{eq:rel_error}
    \big|\Tilde{Z}_\beta-Z_\beta\big|<\epsrel Z_\beta.
\end{equation}
\end{lemma}
\begin{proof}
The estimate $\tilde{z}^{(1)}$ satisfies $|\tilde{z}^{(1)}-\mel{\vv{y}}{\Pi_{\bpi}}{\vv{y}}|\leq \epsilon$. It follows straightforwardly that 
    \begin{equation}
\begin{split}
\frac{e^{-\beta E_{\vv{y}}}}{\mel{\vv{y}}{\Pi_{\bpi}}{\vv{y}}}&\left(1+\frac{\epsilon}{\mel{\vv{y}}{\Pi_{\bpi}}{\vv{y}}}\right)^{-1}
<\frac{e^{-\beta E_{\vv{y}}}}{\Tilde{z}^{(1)}}\\
&\qquad<\frac{e^{-\beta E_{\vv{y}}}}{\mel{\vv{y}}{\Pi_{\bpi}}{\vv{y}}}\left(1-\frac{\epsilon}{\mel{\vv{y}}{\Pi_{\bpi}}{\vv{y}}}\right)^{-1}.
\end{split}    
\end{equation}
Now, using that $1-2\xi\leq(1+\xi)^{-1}$ and $1+2\xi\geq(1-\xi)^{-1}$ for $0\leq \xi\leq1/2$, we can derive the inequality:
\begin{equation}
    \begin{split}
\frac{e^{-\beta E_{\vv{y}}}}{\mel{\vv{y}}{\Pi_{\bpi}}{\vv{y}}}&\left(1-2\frac{\epsilon}{\mel{\vv{y}}{\Pi_{\bpi}}{\vv{y}}}\right)
<\frac{e^{-\beta E_{\vv{y}}}}{\Tilde{z}^{(1)}}\\
&\qquad<\frac{e^{-\beta E_{\vv{y}}}}{\mel{\vv{y}}{\Pi_{\bpi}}{\vv{y}}}\left(1+2\frac{\epsilon}{\mel{\vv{y}}{\Pi_{\bpi}}{\vv{y}}}\right),
\end{split}    
\end{equation}
from which we obtain \eqref{eq:rel_error} by using $Z_\beta=\frac{e^{-\beta E_{\vv{y}}}}{\mel{\vv{y}}{\Pi_{\bpi}}{\vv{y}}}$ and setting $\epsilon\leq \epsrel\big(2\,e^{\beta E_{\vv{y}}}Z_\beta\big)$.
\end{proof}

\begin{lemma}[Complexity]\label{lem:complexity_mon}
    Let $\epsrel>0$ and $\nu'=\nu/2=\frac{\epsrel}{12\,e^{\beta E_{\vv{y}}}Z_\beta}$. Let $\tilde{z}^{(1)}$ be the estimate obtained by Algorithm 1 with inputs $\ket{\psi}=\ket{\phi}=\ket{\vv{y}}$, error $\frac{\epsrel}{2\,e^{\beta E_{\vv{y}}}Z_\beta}$, and $\tilde{f}(A)$ a $\nu$-approximation to $A^t$ with $t=\frac{\log(1/\nu')}{\Delta}$. Then $\Tilde{Z}_\beta:=\frac{e^{-\beta E_{\vv{y}}}}{\tilde{z}^{(1)}}$ is an estimate to the partition function $Z_\beta$ up to relative error $\epsrel$. Moreover, the algorithm has maximal query depth $k=\sqrt{\frac{2}{\Delta}}\Big(\beta E_{\vv{y}}+\log\frac{12\,Z_\beta}{\epsrel}\Big)$ and sample complexity $S^{(1)}=64\,e^{2\beta E_{\vv{y}}} Z_\beta^2\,\frac{\log(2/\delta)}{\epsrel^2}$.
\end{lemma}
\begin{proof}
    The estimate $\tilde{z}^{(1)}$ has three sources of error relatively to $\mel{\vv{y}}{\Pi_{\bpi}}{\vv{y}}$. First, according to Lemma \ref{lemma_partition_function_power_method}, the fact that $t$ is finite causes an error $\nu'$. Second, only an approximation of $A^t$ is implemented, causing an error $\nu$. Finally, the finite number of samples incurs a statistical error $\epsilon_{st}$. Therefore, by triangular inequality, we have $|\tilde{z}^{(1)}-\mel{\vv{y}}{\Pi_{\bpi}}{\vv{y}}|\leq(\nu'+\nu+\epsilon_{st}):=\epsilon$. According to Lemma \ref{lem:estimate_z}, $\Tilde{Z}_\beta=\frac{e^{-\beta E_{\vv{y}}}}{\tilde{z}^{(1)}}$ is an estimate of $Z_\beta$ with relative error $\epsrel$ if the total error is chosen to be $\epsilon= \epsrel/\big(2\,e^{\beta E_{\vv{y}}}Z_\beta\big)$. Since each source of error is controlled individually, for convenience, we take $\nu'=\nu/2=\frac{\epsrel}{12\,e^{\beta E_{\vv{y}}}Z_\beta}$ and $\epsilon_{st}=\frac{\epsrel}{4\,e^{\beta E_{\vv{y}}}Z_\beta}$. The sample complexity is then given by Eq.\ \eqref{eq:SP1} with $\|\vv{a}\|_1=1-\nu$ and error $\epsilon$, while the maximal query depth is given by substituting the expression for $t$ in the expression for the truncation order $k=\sqrt{2\,t\,\log(2/\nu)}$.
\end{proof}

\subsection{Exponential function} \label{sec:explicit_examples_qite}

The following Lemma gives the explicit coefficients of the approximation used here and its maximal degree. It is obtained from Ref.\ \cite[Lemma 4.2]{sachdeva_approximation_2013}.

\begin{lemma}[Chebyshev approximation]\label{lem:approx_exp}
Let $f:\mathbb{R}\rightarrow\mathbb{R}$ be such that $f(x)=e^{-\beta x/2}$. Then, for $\nu\leq \frac{e^{\beta/2}}{2}$, the  polynomial  $\Tilde{f}_\text{exp}:[-1,1]\rightarrow\mathbb{R}$ defined by 
\begin{equation}\label{eq:serie_exp}
\Tilde{f}_\text{exp}(x):=\sideset{}{'}\sum^{k}_{j=0}\left[\sum^{t}_{\substack{l=j\\ l-j\text{ even}}}\frac{(-\beta/2)^{l}}{l!}2^{1-l}\binom{l}{(l-j)/2}\right] \mathcal{T}_{j}(x),
\end{equation}
is a $\nu$-approximation of $f(x)$ on $[-1,1]$ if
$t\geq\left\lceil\max\left\{\frac{\beta e^{2}}{2},\log(\frac{2\,e^{\beta/2}}{\nu})\right\}\right\rceil$, and $k\geq\left\lceil\sqrt{2\,t\log(\frac{4\,e^{\beta/2}}{\nu})}\right\rceil.$
\end{lemma}

\begin{proof}
Lemma 3.4 in Ref.\ \cite{sachdeva_approximation_2013} builds an approximation for the function $e^{\frac{-\beta-\beta x}{2}}$ with error $\nu'\leq 1/2$ in the interval $[-1,1]$. From this construction, it is straightforward to obtain the approximation for $f(x)=e^{-\beta x/2}$ with error $\nu=e^{\beta/2}\,\nu'\leq \frac{e^{\beta/2}}{2}$. The explicit expression for the coefficients is obtained from the truncated Taylor series of $e^{-\beta x}$ combined with the truncated Chebyshev representation of monomials $x^l$. 
\end{proof}

\begin{lemma}[$l_1$-norm of coefficients]\label{lem:1norm_exp}
Let $\vv{a}_\text{exp}$ be the vector of coefficients of the Chebyshev polynomial $\Tilde{f}_\text{exp}$. Then the one-norm satisfies
\begin{equation}
e^{\beta/2}-\nu\leq\norm{\vv{a}_\text{exp}}_{1}\leq e^{\beta/2}+\nu,
\end{equation}
\end{lemma}

\begin{proof}
 First, one can see that, by relabeling the summation index, the coefficients of $\Tilde{f}_{\text{exp}}$ in Eq.\ \eqref{eq:serie_exp} can be rewritten  as
 \begin{equation}\label{eq:exp_coef}
     a_j= 2^{1-\delta_{j,0}}\left(\frac{-\beta}{4}\right)^j \sum_{l=0}^{\left\lfloor\frac{t-j}{2}\right \rfloor} \frac{(\beta/16)^l}{l!(j+l)!}
 \end{equation}
  This form makes it more clear that the coefficients alternate sign. Hence,  Lem.\ref{lemma_1norm} implies the inequality because by definition $|\Tilde{f}_{\text{exp}}(1)-f(1)|=|\Tilde{f}_{\text{exp}}(1)-e^{\beta/2}|\leq\nu$.
\end{proof}

\begin{lemma}[Average query depth]\label{lem:query_exp}
 The average query complexity of estimating $\Tilde{z}^{(2)}$ with $\Tilde{f}_{\text{exp}}$ satisfies $2\,\mathbb{E}[j]\leq \sqrt{2\beta}$. 
\end{lemma}

\begin{proof}
    Notice that 
    \begin{align}
        2\,\mathbb{E}[j]&=\frac{2}{\norm{\vv{a}_\text{exp}}_{1}} \sum_{j=1}^{k}j\,|a_j| \notag\\
        &\leq 4 \,e^{-\beta/2} \sum^{k}_{j=1}\left[2\left(\frac{\beta}{4}\right)^j \sum_{l=0}^{\left\lfloor\frac{t-j}{2}\right \rfloor} \frac{(\beta/16)^l}{l!(j+l)!} \right] j \notag\\
        &\leq 4\, e^{-\beta/2} \sum^{k}_{j=1}\left[2\left(\frac{\beta}{4}\right)^j \sum_{l=0}^{\infty} \frac{(\beta/16)^l}{l!(j+l)!} \right] j \notag\\
        & = 4\, e^{-\beta/2} \sum^{k}_{j=1}\left[2\,I_j(\beta/2) \right] j,
    \end{align}
    where the first inequality comes from applying Lemma\ \ref{lem:1norm_exp} with the worst case of $\nu$ and Eq.\ \eqref{eq:exp_coef}. The second inequality is due to the addition of positive terms to the inner sum. In the last equality, we identified the coefficients of the outer sum with the power series of the modified Bessel functions $I_j(\beta/2)$. This also has the interesting feature of connecting the Chebyshev series in \cite{sachdeva_approximation_2013} with the widely known Jacobi-Anger expansion \cite{abramowitz1966}. Now, we  can use the identity $j\, I_j(\beta/2)=\frac{\beta}{4}(I_{j-1}-I_{j+1})$ to obtain:
    \begin{align}
        2\,\mathbb{E}[j]&\leq 2\,e^{-\beta/2} \beta \left[I_0(\beta/2)-I_{k+1}(\beta/2)\right] \notag\\
        &\leq  2\,e^{-\beta/2} \beta \,I_0(\beta/2)\leq \sqrt{2\,\beta}.
    \end{align}
\end{proof}

\subsection{Inverse function} \label{sec:explicit_examples_inverse}

The following Lemma gives the explicit coefficients of
the approximation used here and its maximal degree. It
is obtained from Ref. \cite{childs_quantum_2017}.

\begin{lemma}\label{lem:approx_inv}
Let $f:\mathbb{R}\rightarrow\mathbb{R}$ be such that $f(x)=1/x$. Then, the polynomial $\Tilde{f}_{inv}:[-1,1]\rightarrow\mathbb{R}$ given by 
\begin{equation}
\Tilde{f}_\text{inv}(x)=4\sum_{j=0}^{k}(-1)^{j}\left[\frac{{\sum}^{b}_{i=j+1}\binom{2b}{b+i}}{2^{2b}}\right]\mathcal{T}_{2j+1}(x),
\end{equation}
is $\nu$-approximation of $1/x$ on $[-1,-1/\kappa]\cup[1/\kappa,1]$ for any 
\begin{equation}
    k\geq\sqrt{b\log(\frac{4b}{\nu})}\quad\text{and}\quad b\geq\kappa^{2}\log(\frac{\kappa}{\nu}) \,.
\end{equation}
\end{lemma}

\begin{lemma}\label{lem:1norm_inv}
Let $\vv{a}_\text{inv}$ be the vector of coefficients of the polynomial $\Tilde{f}_{inv}$. Then $\frac{1}{2} \sqrt{b} - \frac{2bk}{2^{2b}} \binom{2b}{k}\le \norm{\vv{a}_\text{inv}}_{1} \le \frac{2}{\sqrt\pi}\sqrt{b}$. 
In particular, $\norm{\vv{a}_\text{inv}}_{1} \in\Theta\big(\sqrt{b}\big)$. 
\end{lemma}

\begin{proof}
First notice that
\begin{align}
    \norm{\vv{a}_\text{inv}}_1 
    &= \sum_{j=0}^k \left[ \frac{4}{2^{2b}}\sum_{i=j+1}^b \binom{2b}{b+i} \right] \notag\\
    &= \frac{4}{2^{2b}} \left[ \sum_{i=1}^{k+1} \sum_{j=0}^{i-1} \binom{2b}{b+i} + \sum_{i=k+2}^{b} \sum_{j=0}^k \binom{2b}{b+i} \right] \notag\\
    &= \frac{4}{2^{2b}} \left[ \sum_{i=1}^{k+1} i\,\binom{2b}{b+i} + \sum_{i=k+2}^{b} (k+1) \binom{2b}{b+i} \right] \notag\\      
    &= \frac{4}{2^{2b}} \left[ \sum_{i=1}^{b} i\,\binom{2b}{b+i} - \sum_{i=1}^{b-k-1} i\,\binom{2b}{b+i+k+1} \right] \,.
\end{align} 
One immediately gets the upper bound
\begin{align}\label{eq:norm_a_inv_upper}
    \norm{\vv{a}_\text{inv}}_1 
    &\le \frac{4}{2^{2b}} \sum_{i=1}^{b} i\, \binom{2b}{b+i} 
    =\frac{2b}{2^{2b}}\binom{2b}{b}
    \le \frac{2}{\sqrt{\pi}}\sqrt{b}\,,
\end{align} 
where we have used the inequality $\binom{2b}{b}\le \frac{2^{2b}}{\sqrt{\pi b}}$. 
Similarly, for the lower bound we get
\begin{align}
   \norm{\vv{a}_\text{inv}}_1 
    & \ge \frac{4}{2^{2b}}  \sum_{i=1}^{b-k-1} i\left[\binom{2b}{b+i} - \binom{2b}{b+i+k+1} \right] \notag\\
    &= \frac{4}{2^{2b}}  \sum_{i=1}^{b-k-1} i\,\binom{2b}{b+i}\left[1 - \prod_{\ell=0}^k\frac{b-i+1-\ell}{b+i+\ell} \right] \notag\\  
    &\ge \frac{4}{2^{2b}}  \sum_{i=1}^{b-k-1} i\,\binom{2b}{b+i}\left[1 - \left(\frac{b}{2b-1}\right)^{k+1} \right] \notag\\  
    &\ge \frac{2}{2^{2b}}  \sum_{i=1}^{b-k-1} i\,\binom{2b}{b+i} \notag\\
    &= \frac{b}{2^{2b}}\binom{2b}{b} - \frac{2}{2^{2b}}\sum_{i=b-k}^{b} i\,\binom{2b}{b+i} \notag\\
    &\ge \frac{1}{2} \sqrt{b} - \frac{2bk}{2^{2b}} \binom{2b}{k}\,,
\end{align}
where in the second line we have made repeated use of the identity $\binom{n}{m}=\frac{n-m+1}{m}\binom{n}{m-1}$ and the last inequality follows from $\binom{2b}{b}\ge \frac{2^{2b-1}}{\sqrt{b}}$. Since $\binom{2b}{k}\le \left(\frac{2eb}{k}\right)^k$, it follows that in the asymptotic limit $b\gg k\gg 1$ of interest here the second term is is exponentially small in $b$ and, therefore, $\norm{\vv{a}_\text{inv}}_1 \in\Omega\big(\sqrt{b}\big)$. Together with the upper bound \eq{norm_a_inv_upper}, this implies $\norm{\vv{a}_\text{inv}}_1 \in\Theta\big(\sqrt{b}\big)$. 
\end{proof}

\begin{lemma}
Let $\vv{a}_\text{inv}$ be the vector of coefficients of the Chebyshev polynomial $\Tilde{f}_\text{inv}$. Then, the average query complexity of estimating $\Tilde{z}^{(1)}$ satisfies $\mathbb{E}[j]\le b/\norm{\vv{a}_{inv}}_1$. In particular, $\mathbb{E}[j]\in\mathcal{O}\big(\sqrt{b}\big)$.
\end{lemma}

\begin{proof}
The proof parallels that of Lemma \ref{lem:1norm_inv}. Namely, 
\begin{align*}
    \mathbb{E}[j] &= \frac{1}{\norm{\vv{a}_{inv}}_1}\sum_{j=0}^k (2j+1)\left[ \frac{4}{2^{2b}}\sum_{i=j+1}^b \binom{2b}{b+i} \right] \notag\\
    &= \frac{4}{2^{2b}\norm{\vv{a}_{inv}}_1} \left[ \sum_{i=1}^{k+1} i^2\,\binom{2b}{b+i} + \sum_{i=k+2}^{b} (k+1)^2 \binom{2b}{b+i} \right] \notag\\      
    &\le \frac{4}{2^{2b}\norm{\vv{a}_{inv}}_1} \sum_{i=1}^{b} i^2\, \binom{2b}{b+i} \notag\\
    &=\frac{b}{\norm{\vv{a}_{inv}}_1}\,.
\end{align*}      
Lemma \ref{lem:1norm_inv} then implies that $\mathbb{E}[j]\in\mathcal{O}\big(\sqrt{b}\big)$. 
\end{proof}

\subsection{Heaviside step function} \label{sec:explicit_examples_step}

\begin{lemma}\label{lemma_theta}
Let $f:\mathbb{R}\rightarrow\mathbb{R}$ be such that $f(x)=\theta(x)$, where $\theta(x)$ is the unit step function. Then, the function $\Tilde{f}_{s}:[-1,1]\cross[-1,1]\rightarrow\mathbb{R}$ given by the Chebyshev expansion
\begin{equation}
\Tilde{f}_{s}(y-x)=\frac{1}{2}+\sum^{k}_{j=0}a_{j}(y)\mathcal{T}_{2j+1}(x)+b_{j}(y)\sqrt{1-x^{2}}U_{2j}(x)
\end{equation}
where
\begin{equation}
\begin{split}
a_{j}(y)&=-\sqrt{\frac{2\sigma}{\pi}}e^{-\sigma}\frac{I_{j}(\sigma)+I_{j+1}(\sigma)}{2j+1}\sqrt{1-y^{2}}U_{2j}(y),\\
b_{j}(y)&=\sqrt{\frac{2\sigma}{\pi}}e^{-\sigma}\frac{I_{j}(\sigma)+I_{j+1}(\sigma)}{2j+1}\mathcal{T}_{2j+1}(y),
\end{split}
\end{equation}
with $I_{n}$ being the modified Bessel function of the first kind and $U_{2j}$ is a Chebyshev polynomial of the second-kind, is a $\nu$-approximation of $f(y-x)$ on all $(x,y)\in[-1,1]\cross[-1,1]$ such that $x\in[-1,-\xi+y]\cup[\xi+y,1]$ if 
\begin{equation}
k=\mathcal{O}\left(\tfrac{1}{\xi}\log\big(\tfrac{1}{\nu}\big)\right)\text{, and }\sigma=\mathcal{O}\left(\tfrac{1}{\xi^2}\log\big(\tfrac{1}{\nu}\big)\right).
\end{equation}
\end{lemma}

\begin{proof}
We take the Chebyshev approximation 
\begin{equation}
\Tilde{f}(x)=\frac{1}{2}+\sqrt{\frac{2\sigma}{\pi}}e^{-\sigma}\sum^{k}_{j=0}(-1)^{j}\frac{I_{j}(\sigma)+I_{j+1}(\sigma)}{2j+1}\mathcal{T}_{2j+1}(x)
\end{equation}
described in \cite{wan_randomized_2022} and extend it to the two-variable case using the identity
\begin{equation}
\theta(y-x)=\theta(\sin(\arcsin{y}-\arcsin{x})),
\end{equation}     
that follows from fact that $\sin$ and $\arcsin$ are monotonous functions, and expand
\begin{equation}
\mathcal{T}_{2j+1}(\sin(\arcsin{y}-\arcsin{x}))
\end{equation}
using trigonometric identities and the trigonometric definitions of $\mathcal{T}_{2j+1}$ and $U_{2j}$. The conditions on $k$ and $\sigma$ needed for obtaining the $\nu$-approximation end up being the same as those for the single-variable version.
\end{proof}

\begin{lemma}
[$l_1$-norm of coefficients]\label{lem:1norm_step}
Let $\vv{a}_{s}(y_*)$ and $\vv{b}_{s}(y_*)$ be the vector of coefficients of the Chebyshev expansion $\Tilde{f}_{s}(y_*-x)$ for $y_*=\frac{1}{\sqrt{2}}$. Then 
\begin{equation}
\begin{split}
\norm{\vv{a}_{s}(y_*)}_1=\norm{\vv{b}_{s}(y_*)}_1 \in \Theta(\log(k)).
\end{split}
\end{equation}
\end{lemma}

\begin{proof}
Let $\norm{\vv{a}_{s}(y_*)}_1=\norm{\vv{b}_{s}(y_*)}_1=N(k,\sigma)$. To prove the asymptotic growth orders we need to investigate the behavior of the one-norms when $\xi,\nu\rightarrow0$, which is equivalent to $k,\sigma\rightarrow\infty$. Since, by Lemma \ref{lemma_theta}, $\sigma$ grows quadratically faster then $k$ as $\xi\rightarrow0$, we first investigate the behavior of $N(k,\sigma)$ as $\sigma$ grows. Computing the limit
\begin{equation}
\Bar{N}(k)=\lim_{\sigma\rightarrow\infty}N(k,\sigma)=\lim_{\sigma\rightarrow\infty}\sum^{k}_{j=0}\sqrt{\frac{\sigma}{\pi}}e^{-\sigma}\frac{I_{j}(\sigma)+I_{j+1}(\sigma)}{2j+1}
\end{equation}
gives us
\begin{equation}
\Bar{N}(k)=\frac{1}{\pi\sqrt{2}}\sum^{k}_{j=0}\frac{1}{2j+1},
\end{equation}
where we have used $\lim_{\sigma\rightarrow\infty}\sqrt{\sigma}e^{-\sigma}I_{l}(\sigma)=\frac{1}{\sqrt{2\pi}}$. 
Now, note that
\begin{equation}
1+\frac{H_{k}}{2}<\sum^{k}_{j=0}\frac{1}{2j+1}<\frac{H_{k+1}}{2},
\end{equation}
where $H_{k}=\sum^{k}_{j=1}\frac{1}{j}$ is the $k$-th harmonic number. 
By the properties of $H_k$, we have that there exists $k_{0}$ such that for all $k>k_{0}$,
\begin{equation}
\frac{1}{2\pi\sqrt{2}}\ln(k+1)<\Bar{N}(k)<\frac{\sqrt{2}}{\pi}\ln(k+1),
\end{equation}
which proves our statements.
\end{proof}

\begin{lemma}\label{lemma_step_query}
Let $\vv{a}_{s}(y_*)$ and $\vv{b}_{s}(y_*)$ be the vector of coefficients of the Chebyshev expansion $\Tilde{f}_{s}(y_*-x)$ for $y_*=\frac{1}{\sqrt{2}}$. Then, the average query complexity of estimating $\Tilde{z}^{(1)}$ is $\mathbb{E}[j]\in\mathcal{O}\left(\tfrac{1}{\xi}\sqrt{\log\big(\tfrac{1}{\nu}\big)}\Big/\!\log\big(\frac{1}{\xi}\log\big(\frac{1}{\nu}\big)\big)\right)$.  
\end{lemma}

\begin{proof}
By definition of $\mathbb{E}[j]$, and using the fact that $\abs{a_{j}(y_*)}=\abs{b_{j}(y_*)}$ for all $j$, we have that
\begin{align}
\mathbb{E}[j] &=\frac{1}{\norm{\vv{a}_{s}(y_*)}_1}\sum^{k}_{j=0}\sqrt{\frac{4\sigma}{\pi}}e^{-\sigma}(I_{j}(\sigma)+I_{j+1}(\sigma))\notag\\
&\le \frac{1}{\norm{\vv{a}_{s}(y_*)}_1}\sum^{\infty}_{j=0}\sqrt{\frac{4\sigma}{\pi}}e^{-\sigma}(I_{j}(\sigma)+I_{j+1}(\sigma))\notag\\
&=\frac{1}{\norm{\vv{a}_{s}(y_*)}_1}\sqrt{\frac{4\sigma}{\pi}},
\end{align}
where in the last step $\sum^{\infty}_{j=0}\big(I_{j}(\sigma)+I_{j+1}(\sigma)\big)=e^\sigma$ follows from the Jacobi-Anger expansion. 
Since $\norm{\vv{a}_{s}(y_*)}=\Omega(\log(k))$ for large $k$, it follows that 
\begin{equation}
    \mathbb{E}[j] = \mathcal{O}\left(\frac{\sqrt{\sigma}}{\log(k)}\right) = \mathcal{O}\left(\frac{\tfrac{1}{\xi}\sqrt{\log\big(\tfrac{1}{\nu}\big)}}{\log\big(\frac{1}{\xi}\log\big(\frac{1}{\nu}\big)\big)}\right).
\end{equation}
\end{proof}

\begin{lemma}
Let $\vv{a}_{s}(y_*)$ and $\vv{b}_{s}(y_*)$ be the vector of coefficients of the Chebyshev expansion $\Tilde{f}_{s}(y_*-x)$ for $y_*=\frac{1}{\sqrt{2}}$. Then, the error damping factor for estimating $\Tilde{z}^{(1)}$ is $\mathcal{O}\left(\tfrac{1}{\xi}\sqrt{\log\big(\tfrac{1}{\nu}\big)}\Big/\!\log\big(\frac{1}{\xi}\log\big(\frac{1}{\nu}\big)\big)\right)$.
\end{lemma}

\begin{proof}
Follows immediately from Lemmas \ref{lemma_theta} and \ref{lemma_step_query}.
\end{proof}

\begin{definition}
Given $\Tilde{F}_{\varrho}(y)=\Tr[\varrho\,\Tilde{f}_{s}(y\mathds{1}-H)]$, its T-part is the expression
\begin{equation}
\sum^{k}_{j=0}a_{j}(y)\Tr[\mathcal{T}_{2j+1}(H)\,\varrho],
\end{equation}
while its U-part is the expression
\begin{equation}
\sum^{k}_{j=0}b_{j}(y)\Tr[\sqrt{1-H^{2}}\,U_{2j}(H)\,\varrho].
\end{equation}

\end{definition}

\begin{lemma}\label{lemma_ychoice}
Let $y_*=\frac{1}{\sqrt{2}}$, and $\Tilde{z}^{(1),T}(y)$ and $\Tilde{z}^{(1),U}(y)$ be the unbiased estimators of the $T$-part and the $U$-part of $\Tilde{F}_{\varrho}(y)$ obtained from algorithm 1, for any point $y$, respectively. Then, for any point $y$ we have that
\begin{equation}
\Tilde{F}_{\varrho}(y)=\frac{1}{2}+\mathbb{E}_{w^{T}(y)}\left[\Tilde{z}^{(1),T}(y_*)\right]+\mathbb{E}_{w^{T}(y)}\left[\Tilde{z}^{(1),U}(y_*)\right],
\end{equation}
where $\mathbb{E}_{w^{T}(y)}$ and $\mathbb{E}_{w^{U}(y)}$ are the weighted expected values with weight functions
\begin{equation}
\begin{split}
w^{T}(y,j)=\frac{a_{j}(y)}{a_{j}(y_*)},\\
w^{U}(y,j)=\frac{b_{j}(y)}{b_{j}(y_*)}.
\end{split}
\end{equation}
\end{lemma}

\begin{proof}
To show this, just notice that by the definition of $\mathbb{E}_{w^{T}(y)}$ and $\Tilde{z}^{(1,T)}(y_*)$ that $\mathbb{E}_{w^{T}(y)}\left[\Tilde{z}^{(1),T}(y_*)\right]$ is equal to
\begin{equation}
\begin{split}
&\sum^{k}_{j=0}\sum_{b}w^{T}(y,j)p_{j}(p(b|j)\norm{\vv{a}_{s}(y_*)}\sgn(a_{j}(y_*))b)\\
=&\sum^{k}_{j=0}w^{T}(y,j)a_{j}(y_*)\mel{\psi}{\mathcal{T}_{2j+1}(H)}{\psi}\\
=&\sum^{k}_{j=0}a_{j}(y)\mel{\psi}{\mathcal{T}_{2j+1}(H)}{\psi}=\mathbb{E}\Big[\Tilde{z}^{(1),T}(y)\Big],
\end{split}
\end{equation}
while a similar computation also shows that
\begin{equation}
\mathbb{E}_{w^{U}(y)}\left[\Tilde{z}^{(1),U}(y_*)\right]=\mathbb{E}\Big[\Tilde{z}^{(1),U}(y)\Big].
\end{equation}
which proves the lemma.
\end{proof}

\begin{corollary}
To sample $\Tilde{F}_{\varrho}(y)$ for $\mathcal{O}(\log(\xi^{-1}))$ values of $y$ it is only necessary to sample $\Tilde{F}_{\varrho}(y_*)$ with $1-\mathcal{O}(\delta/\log(\xi^{-1}))$ confidence, and for each desired value of $y$, to multiply each sample by the appropriate weight factor.
\end{corollary}

\begin{proof}
Follows directly from the previous lemma.
\end{proof}

\section{Extension to QSVT} \label{sec:QSVT}

For simplicity, in the main text, we considered only Hermitian matrices. In that case, there is an operation to build a qubitized oracle whose powers directly lead to the Chebyshev polynomials of the block-encoded matrix. This section shows that our algorithm can be applied to randomize the more general quantum singular value transformation technique.

Let us consider a unitary $U_A$ such that 
$\big(\bra{0}_a\otimes \mathds{1}_s\big)\,U_A\,\big(\ket{0}_a\otimes \mathds{1}_s\big) = A$ as in Eq.\ \eqref{eq:block_encoding}, except that now the block-encoded matrix $A$ does not need to be Hermitian. In fact, $A$ could even be rectangular. In this case, $U_A$ is a bock-encoding of a square matrix built from $A$ by filling the remaining entries with zeroes. The input matrix assumes a singular value decomposition, and we are interested in singular value transformations as in the definition below (Def.\ 16 in Ref.\ \cite{Gilyen2019}):

\begin{definition}[Singular value transformations] \label{def:SVT} Let $f:\mathbb{R}\rightarrow \mathbb{C}$ be an even or odd function. Let $\sum_{j=1}^{d} \sigma_j\ketbra*{\tilde{\psi}_j}{\psi_j}$ be a singular value decomposition of  the $d\times d$ matrix $A$ with right eigenvectors $\{\ket{\psi_j}\}$, left eigenvectors $\{\ket*{\tilde{\psi}_j}\}$ and singular values $\{{\sigma_j}\}$(including zero values). The singular value transformation of $A$ by  $f$ is defined as 
\begin{equation}
    f^{(SV)}(A) := \begin{cases}
        \sum_{j=1}^{d} f(\sigma_j)\ketbra{{\psi}_j}{\psi_j},\quad \text{$f$ even,}\\
        \sum_{j=1}^{d} f(\sigma_j)\ketbra*{\tilde{\psi}_j}{\psi_j},\quad \text{$f$ odd.}
    \end{cases}
\end{equation}
\end{definition}

It was shown in Ref.\ \cite{Gilyen2019}  that a quantum singular value transformation (QSVT) could be implemented from calls to $U_A$ and $U_A^\dagger$, and single qubit rotations controlled by $\ket{0}_a$. In Theorem 17 and Corollary 18, the authors state that  the singular value transformation by a real polynomial of degree $k$ with even or odd parity is always achieved by some sequence of qubit rotations $\Phi=\{\phi_1,\cdots,\phi_k\}\in\mathbb{R}$ and a circuit making $k$ queries to $U_A$ and $U_A^\dagger$, as long as the polynomial is normalized to $1$. In particular, to implement the Chebyshev polynomial $T^{(SV)}_k(A)$, the sequence of angles is simply $\phi_1=(1-k)\frac{\pi}{2}$ and $\phi_j=\frac{\pi}{2}$ for $j\neq1$.

Under singular value transformation, Probs. \ref{problem.1} and \ref{problem.2} can be phrased in the same way, except that now the Hermitean $A$ is substituted by a general matrix $A$ and the function is meant to be a singular value transformation according to Def.\ \ref{def:SVT}. Moreover, the problems are well-posed only for  functions with  defined parity.  Therefore, among the use cases presented in the main text, the ones that are extendable to general matrices are the inverse function and matrix monomials. 

The randomized QSVT solutions to Probs. \ref{problem.1} and \ref{problem.2} follow as before by sampling polynomial degrees according to their weight in the Chebyshev expansion of the target function. A Hadamard test is then run for each polynomial degree draw, the Chebyshev polynomials being implemented using the circuit in Fig.~\ 1 in \cite{Gilyen2019} with the angle sequence described above. The algorithm works for the same reason as for Hermitian matrices: QSVT implements a block-encoding of the polynomial singular value transformation, while the Hadamard test has the useful feature of always selecting the correct block of a unitary block-encoding.

\section{Non-normalized matrices}\label{app:non_norm}

Suppose that, instead of having a normalized matrix $A$ (i.e., satisfying $\norm{A}\leq1$), we have a general Hermitian matrix $A$ as input. We can define the normalized matrix $A'=\frac{A}{\alpha}$, with $\alpha\geq\norm{A}$ to be block-encoded instead of $A$. Therefore, using our framework we can simulate a function $f'(A')$ defined such that 
\begin{equation}
f(A)=f(\alpha\, A')=f'(A'),
\end{equation}
where $f:\mathcal{D}\rightarrow\mathbb{R}$ and $f':\frac{\mathcal{D}}{\alpha}\subseteq[-1,1]\rightarrow\mathbb{R}$ have the same image. Therefore, introducing a sub-normalization will require a Chebyshev expansion for a new function whose domain is also sub-normalized relative to the original. The contribution of the sub-normalization factor $\alpha$ to the sample and query complexities of the randomized algorithm needs to be analyzed on a case-by-case basis. In the following, we analyze the maximal query depth and the $l_1$-norm of the Chebyshev coefficients for each use-case we showed in the main text. From them and  Theorem \ref{main_lemma2}, the complexities of particular instances of Problem 1 or Problem 2 can be obtained.

In the case of monomials, we have that $f(x)=x^t$ yields $f'(x')=\alpha^t {x'}^t$. This implies that an $\nu$-approximation to $A^{n}$ is $\nu/\alpha^{n}$-approximation to $H^{'n}$. To attain an error $\nu$ in approximating $f(x)$, the approximation in Lemma \ref{lem:approx_mon} shall be used with error $\nu/\alpha^n$ and also the $l_1$-norm of the coefficients will gain a factor $\alpha^n$. Therefore, The maximal query depth and coefficients are ${k}=\sqrt{2t\log(\frac{2\alpha^{t}}{\nu})}$ and $\norm{\vv{a}}_{1}=\alpha^{t}-\nu$, respectively.

In the case of the exponential function, notice that $e^{-\beta A}=e^{-(\beta\alpha)(A/\alpha)}$. Therefore, we just need to re-scale  $\beta$ with a factor  $\alpha$  , obtaining ${k}=\mathcal{O}\left(\sqrt{\alpha\beta}\log(\frac{e^{\alpha\beta}}{\nu})\right)$ and 
$\norm{\vv{a}}_{1}=e^{\alpha\beta}$ from Lemmas \ref{lem:approx_exp} and \ref{lem:1norm_exp}, respectively.

In the case of the inverse function, $f(x)=x^{-1}$ needs to be approximated in the domain $\mathcal{D}=[-\norm{A},-1/\norm{A^{-1}}]\cup[1/\norm{A^{-1}},\norm{A}]$. In this case, the randomized algorithm should implement an approximation to $f'(A')=\alpha^{-1}{A'}^{-1}$ in the domain $\mathcal{D}/\alpha=[-1,-1/\kappa]\cup[1/\kappa,1]$, with $\alpha=\norm{A}$ and $\kappa=\norm{A}\,\norm{A^{-1}}$ as usual. It suffices to apply Lemma \ref{lem:approx_inv} to obtain a $\alpha\nu$-approximation to ${A'}^{-1}$ with resulting maximal degree  
$k=\mathcal{O}\left(\kappa\,{\log\frac{\norm{A^{-1}}}{\nu}}\right)$ and coefficients $l_1$-norm 
$\norm{\vv{a}}_{1}=\mathcal{O}\left(\norm{A^{-1}}{\log\frac{\norm{A^{-1}}}{\nu}}\right)$.

Lastly, in the case of the Heaviside step function, we have that $f(A)=f'(A')$. However, in order to approximate $\theta(A)$ in the domain $D=(-1,-\xi]\cup[\xi,1)$, we need to be able to approximate $\theta(A')$ in the larger domain $D=(-1,-\xi/\alpha]\cup[\xi/\alpha,1)$. Therefore, Lemmas \ref{lemma_theta} and \ref{lem:1norm_step} give $k=\mathcal{O}\left(\frac{\alpha}{\xi}\,\log\frac{1}{\nu}\right)$ and $\norm{\vv{a}}_1=\mathcal{O}\left(\log\big(\frac{\alpha}{\xi}\,\log\frac{1}{\nu}\big)\right)$.

\section{GSEE for FeMoco}
\label{app:gate_count}

Here, we provide details on the complexity analysis of GSEE for the specific case of FeMoco. We use the sparse qubitized Hamiltonian model of \cite{Berry2019qubitizationof} with the corrections given in \cite{Lee2021} and the active space of \cite{Li_2019}.

The molecule's Hamiltonian is written as a linear combination $H'=\sum_{l=1}^L b_l\,V_l$ of unitary operators $V_l$. 
A block-encoding of $H=H'/\|\boldsymbol{b}\|_1$ is obtained using the method of linear combination of unitaries (LCU) \cite{Childs_2012}, i.e., encoding the coefficients $b_l$ in the amplitudes of a state of $\order{\log(L)}$ ancillary qubits, which are used to control the action of each unitary $V_l$ in the sum. 
In order to estimate the ground-state energy of $H'$ up to chemistry precision of $0.0016$ $E_\text{h}$ (hartree), we can run the GSEE algorithm on $H$ with precision $\xi=0.001/\|b\|_1$, leaving the remaining $0.0006/\|b\|_1$ tolerated error to imprecision in the block-encoding procedure.

The number of terms $L$ in the Hamiltonian, the $1$-norm $\|b\|_1$ of their coefficients, and also the total number of qubits required depend on the active space used. 
In the case we consider, this is composed of $n=152$ orbitals, each one corresponding to one qubit in the system. 
In Ref.\ \cite[App. A]{Lee2021} it was shown that $\|\boldsymbol{b}\|_1=1547$ $E_\text{h}$ and $L=440501$. 
Moreover, the authors show how to construct a qubitized block-encoding oracle for $H$ explicitly using LCU with $2446$ logical qubits and $1.8\times 10^4$ Toffoli gates per oracle query. It is worth noting that, in principle, only $n+a$ qubits are strictly required for LCU,  with $a=\lceil \log(L)\rceil=19$. Therefore, the majority of the qubits are used to upload the Hamiltonian coefficients into the $a$ ancillas state and could be removed if a more efficient upload is used. 
They apply quantum phase estimation (QPE), which requires $Q=\lceil\frac{\pi}{2\xi}\rceil\approx2.4\times 10^6$ oracle queries and $2\,\log(Q+1)\approx 43$ additional control qubits. 
In our case, only one control qubit is necessary for the Hadamard test, while the number of oracle calls $Q$ is to be compared to $k$ or $\mathbb{E}[j]$.

At last, to implement our algorithm, we need an initial state with a sufficiently large overlap $\eta$ with the target ground state. 
In Ref.\ \cite{Lee_2023} it was shown that $\eta^2=10^{-7}$ can be achieved using Slater determinants. 
Using this and the value of $\xi$ given above, we explicitly construct the series approximating the Heaviside function $\theta(H)$. 
The truncation degree and the average degree obtained are $k=1.12\times10^7$ and $\mathbb{E}[j]=4.01\times 10^5$, respectively. When multiplied by the above-mentioned gate complexity per oracle query, we get a total of $2.03\times10^{11}$ and $7.26\times10^9 $ Toffoli gates per circuit, respectively. 
Notice that, when compared to phase estimation, the coherent QSP algorithm demands more coherent queries to the oracle while the randomized processing is less costly on average. 
At the same time, all three algorithms require a similar number of samples, proportional to $\eta^{-2}$. We summarize the complexities of the three methods in Table\ \ref{tab:femoco}.

\begin{table}[b]
    \centering
    \begin{tabular}{||c||c|c|c||}
     \hline
     \rule{0pt}{3ex}
     \rule[-1.5ex]{0pt}{0pt}
     {\bf Method} & {\bf Calls to $U_A$} & {\bf Toffoli count} & {\bf Qubit count} \\
     \hline
     \hline
     \rule{0pt}{3ex}
     \rule[-1.5ex]{0pt}{0pt}
     QPE \cite{Lee2021} & $2.43\times 10^6$ & $4.4\times10^{10}$ & $2489$ \\
     \hline
     \rule{0pt}{4ex}
     \rule[-2ex]{0pt}{0pt}
     QSP & $1.12\times10^7$ & $2.0\times10^{11}$ & $2448$  \\
     \hline
     \rule{0pt}{4ex}
     \rule[-1.5ex]{0pt}{0pt}
     This work & $4.01\times 10^5$ &  $7.3\times10^9 $ & $2447$ \\
     \hline
    \end{tabular}
    \caption{Comparison of complexities of GSEE for FeMoco using the three different algorithms described in the text. The number of ancillas for QSP is increased by $2$ to account for the function-implementation ancilla and the Hadamard test ancilla.}
    \label{tab:femoco}
\end{table}

Another interesting feature of the randomized scheme is that, due to its variable circuit depth, it allows for variable quantum error-correcting code distances. 
Consequently, the average number of physical qubits required for quantum error correction (QEC) and also the runtime of the typical circuits, which have even lower depth than the average, can be highly reduced further. 
To give a more quantitative idea, we estimated the QEC overhead using the surface code with the spreadsheet provided in Ref. \cite{Gidney2019efficientmagicstate}. 
We consider a total number of Toffoli's obtained using the average degree $\mathbb{E}[j]$, (i.e., not taking into account the possible further due to typical circuits with even lower depths). 
With a level-1 code distance $d_1=19$ and a level-2 code distance $d_2=31$, assuming an error rate per physical gate of $0.001$, the error-corrected algorithm uses $8$ million qubits to obtain a global error budget of $1\%$. 
This means that all circuits with a depth smaller than the average will have an extremely high fidelity. 
The circuits corresponding to larger degrees will have an increasingly higher chance of failure, but they contribute less to the final estimate. 
Meanwhile, the same QEC overhead only ensures roughly $20\%$ chance of errors to a circuit corresponding to the truncation order $k$.

\end{document}